\documentclass[ba,preprint]{imsart}
\pubyear{2024}
\volume{TBA}
\issue{TBA}
\firstpage{1}
\lastpage{1}

\usepackage{amsmath,amsthm,amssymb}
\usepackage{natbib}
\usepackage[colorlinks,citecolor=blue,urlcolor=blue,filecolor=blue,backref=page]{hyperref}
\usepackage{graphicx}
\usepackage{subfig}

\usepackage{bbm}

\usepackage{comment}

\startlocaldefs
\numberwithin{equation}{section}

\theoremstyle{plain}

\newtheorem{prop}{Proposition}

\endlocaldefs

\usepackage{xcolor} 

\begin{document}

\begin{frontmatter}
\title{Model Selection-Based Estimation for Generalized Additive Models Using Mixtures of g-priors: Towards Systematization}

\runtitle{Generalized Additive Models Using Mixtures of g-priors}

\begin{aug}
\author{\fnms{Gyeonghun} \snm{Kang}\thanksref{addr1}}
\and
\author{\fnms{Seonghyun} \snm{Jeong}\thanksref{addr2,addr3,t1}}

\runauthor{G. Kang and S. Jeong}

\address[addr1]{Department of Statistical Science, Duke University, Durham, North Carolina, USA
}

\address[addr2]{Department of Statistics and Data Science, Yonsei University, Seoul, Korea
}

\address[addr3]{Department of Applied Statistics, Yonsei University, Seoul, Korea
}

\thankstext{t1}{Corresponding author: sjeong@yonsei.ac.kr}

\end{aug}

\begin{abstract}	
We explore the estimation of generalized additive models using basis expansion in conjunction with Bayesian model selection. Although Bayesian model selection is useful for regression splines, it has traditionally been applied mainly to Gaussian regression owing to the availability of a tractable marginal likelihood. We extend this method to handle an exponential family of distributions by using the Laplace approximation of the likelihood. Although this approach works well with any Gaussian prior distribution, consensus has not been reached on the best prior for nonparametric regression with basis expansions. Our investigation indicates that the classical unit information prior may not be ideal for nonparametric regression. Instead, we find that mixtures of g-priors are more effective. We evaluate various mixtures of g-priors to assess their performance in estimating generalized additive models. Additionally, we compare several priors for knots to determine the most effective strategy. Our simulation studies demonstrate that model selection-based approaches outperform other Bayesian methods.
\end{abstract}

\begin{keyword}[class=MSC]
\kwd[Primary ]{62G08}
\kwd[; secondary ]{62J12}
\end{keyword}

\begin{keyword}
\kwd{Bayesian nonparametrics}
\kwd{exponential family models}
\kwd{mixtures of g-priors}
\kwd{nonparametric regression}
\kwd{regression splines}
\end{keyword}

\end{frontmatter}

\section{Introduction}
\label{sec:intro}

Since its inception, the generalized additive model (GAM) has been pivotal in statistics and machine learning, garnering significant attention from both theorists and practitioners.
The GAM represents an interpretable semiparametric approach that balances between parametric generalized linear models (GLMs) and fully nonparametric regression with multidimensional smoothing.  Specifically, the GAM describes the relationship between multiple predictor variables and a (possibly non-Gaussian) response variable through an additive structure of univariate functions \citep{hastie1986generalized}. This approach sacrifices the flexibility of multidimensional smoothing for a clear interpretation of each predictor variable's contribution to the mean as a univariate function. 

Several estimation methods have been proposed for nonparametric regression and additive models, from both frequentist and Bayesian perspectives. Common Bayesian techniques for estimating univariate smooth functions include Gaussian process priors \citep{williams1995gaussian}, Bayesian P-splines \citep{lang2004bayesian}, and basis expansion methods with model selection \citep{smith1996nonparametric,denison1998automatic,dimatteo2001bayesian}. Among these approaches, basis expansion with Bayesian model selection (BMS), which we call the BMS-based approach to nonparametric regression, stands out for its theoretical benefits and empirical success \citep{smith1996nonparametric,denison1998automatic,dimatteo2001bayesian,rivoirard2012posterior,de2012adaptive,shen2015adaptive}. BMS-based approaches determine suitable basis functions by comparing Bayes factors, thereby selecting more plausible basis terms in a data-driven manner. These methods are also useful for multidimensional smoothing, as seen in Bayesian multivariate adaptive regression splines \citep{denison1998bayesian} and Bayesian additive regression trees \citep{chipman2010bart,jeong2023art}.

Despite their conceptual simplicity, BMS-based methods can be computationally challenging owing to the need for marginal likelihood calculations. This limitation has typically restricted the application of BMS to Gaussian regression within nonparametric regression contexts. For GLMs and GAMs, marginalization is often impractical even with conjugate priors on the coefficients \citep{chen2003conjugate}.
The most feasible scenario often involves cases with available latent variable expressions, such as probit regression \citep[e.g.,][]{jeong2017analysis,sohn2022functional}.
When marginalization is not analytically tractable, BMS-based methods require numerical marginalization of the coefficients using Markov chain Monte Carlo (MCMC) algorithms, such as reversible jump MCMC \citep{green1995reversible}. These methods can be significantly less efficient than using Bayes factors unless a well-designed proposal distribution is available. This challenge has contributed to the early preference for P-spline-based Bayesian methods for estimating GAMs \citep[e.g.,][]{fahrmeir2001bayesian,brezger2006generalized}.

A practical solution to this issue is to use an approximation of the likelihood, such as the Laplace approximation, which allows for the calculation of marginal likelihood with a Gaussian prior distribution on the coefficients  \citep{li2018mixtures}. This approach enables the application of BMS-based methods to estimate GAMs with distributions from the exponential family. While the Laplace approximation has been occasionally used in BMS-based methods \citep[e.g.,][]{dimatteo2001bayesian}, it is more widely accepted in the literature for improving computational efficiency in Bayesian P-splines for GAM estimation	\citep{bove2015objective,gressani2021laplace}. 

When a Gaussian or Gaussian mixture prior is used for the coefficients, the Laplace approximation allows for a straightforward derivation of a closed-form expression for the marginal likelihood.  However, the optimal prior distribution for basis determination remains unclear. Literature on variable selection indicates that mixture priors often outperform the classical Gaussian prior, known as Zellner's g-prior, and its variants \citep{liang2008mixtures,li2018mixtures}. Such mixture priors, also known as mixtures of g-priors, are preferred because of their desirable properties and ability to resolve issues associated with the g-prior \citep{liang2008mixtures}.
Various mixtures of g-priors have been proposed within the framework of linear regression \citep[e.g.,][]{zellner1980posterior,liang2008mixtures,maruyama2011fully,bayarri2012criteria,womack2014inference}, and some attempts have been made to extend them to GLMs \citep{bove2011hyper,held2015approximate,fouskakis2018power}. Recently, \citet{li2018mixtures} provided a comprehensive framework for mixtures of g-priors for the GLM. 
However, the best-performing mixture prior for the BMS-based GAM estimation remains uncertain. To address this, understanding how mixtures of g-priors affect the penalization of nonparametric functions is essential. Even within Gaussian additive regression, determining the best mixture prior remains unresolved.

Another important consideration in BMS-based methods is selecting a prior distribution for the intrinsic basis terms. Since spline basis functions are often determined by the knot locations, this implies a prior on the knots. Various prior distributions have been proposed to balance computational efficiency and estimation quality  \citep[e.g.,][]{smith1996nonparametric,denison1998automatic,dimatteo2001bayesian,shen2015adaptive}. These priors can be categorized based on their underlying principles; generally, more flexible priors offer better approximation but come with greater computational costs. The most suitable class of prior distribution for determining optimal spline knots remains unclear.

This study makes three key contributions. First, we systematize BMS-based approaches for GAM estimation by using the Laplace approximation and a unified framework for mixtures of g-priors, as proposed by \citet{li2018mixtures}. In doing so, we enhance computational efficiency by introducing a new form of natural cubic spline function specifically tailored for BMS-based methods. Second, among various mixtures of g-priors within a general class, we identify the default mixture prior for GAM estimation. We deepen our understanding of how mixtures of g-priors penalize the model during GAM estimation and evaluate the empirical performance of different mixture priors through extensive simulations. Our findings suggest that the traditional g-prior, also known as the unit information prior \citep{kass1995reference}, may be less suitable. Instead, a mixture of g-priors is recommended. Finally, we categorize prior distributions for knots into three groups and assess which class is most effective for GAM estimation. Our investigation reveals that a prior distribution balancing flexibility and computational efficiency performs best. Specifically, while the most flexible prior, the free-knot spline \citep{denison1998automatic,dimatteo2001bayesian}, may be excessive in practice, a less flexible but computationally efficient approach based on variable selection \citep{smith1996nonparametric} yields better empirical results with fast mixing. We support these findings with various numerical results. The R package implementing the sampling algorithms for our GAM estimation is available on the first author's GitHub page.\footnote{https://github.com/hun-learning94/gambms}

The remainder of this paper is organized as follows. Section~\ref{sec2} introduces the construction of GAMs using spline basis expansion with natural cubic splines. Section~\ref{sec:mixgprior} discusses mixtures of g-priors for BMS within a unified framework and compares these priors for GAM estimation, interpreting them as penalty functions for nonparametric regression. Section~\ref{sec:knotprior} categorizes prior distributions for knots into three strategies and evaluates their effectiveness for BMS-based approaches. Section~\ref{sec:sims} presents comprehensive simulations and numerical studies to identify the best prior distribution and compare BMS-based methods with other approaches for GAM estimation. 
Section~\ref{sec:realdata} applies the BMS-based method to the Pima diabetes dataset.
Finally, Section~\ref{sec:discussion} concludes the study with a discussion. Supplementary material includes proofs of propositions, additional simulation studies, and instructions for installing the R package.

\section{Generalized additive models via basis expansion}
\label{sec2}

For given predictor variables $x_{i}=(x_{i1},x_{i2},\dots,x_{ip})^{T}\in\mathbb{R}^{p}$, suppose the response variable $Y_i\in\mathbb R$ follows a distribution from the exponential family. The density of $Y_i$ is given by
\begin{align}\label{eqn:glm}
	y_i\mapsto p(y_{i};\theta_i,\phi)=\exp\bigg(\dfrac{y_{i}\theta_{i}-b(\theta_{i})}{\phi}+c(y_{i},\phi)\bigg),\quad i=1,\dots,n,
\end{align}
where $\theta_{i}$ is the natural parameter modeled by $x_i$, $\phi$ is a scale parameter, and $b$ and $c$ are known functions. The dependence of $\theta_i$ on $x_i$ is clarified below. Although we focus primarily on cases where the dispersion parameter $\phi$ is known, we also consider Gaussian regression with an unknown $\phi$ in Section~S5 of the supplementary material.
Assuming $b$ is twice differentiable with $b''(\theta_i)>0$, the expected value and variance of $Y_i$ are $E(Y_{i})=b'(\theta_{i})$ and $Var(Y_{i})=\phi b''(\theta_{i})$, respectively. We use a monotonically increasing link function $h$ to parameterize the natural parameter as $\theta_{i}= (h\circ b')^{-1}(\eta_i)$, where $\eta_i$ is an additive predictor defined as
\begin{align}\label{eqn:gam}
	\eta_{i}=\alpha +\sum_{j=1}^{p}f_{j}(x_{ij}) ,\quad i=1,\dots, n,
\end{align}
with a global mean $\alpha$ and univariate functions $f_j:\mathbb R\to \mathbb R$, $j=1,\dots,p$. To ensure identifiability, we assume that the functions $f_j$ satisfy the restriction $\sum_{i=1}^n f_j(x_{ij})=0$, $j=1,\dots,p$. 

The key aspect of the model specification is determining how to characterize the nonparametric functions $f_j$.
In this study, the functions $f_j$ are parameterized using a spline basis representation. Specifically, $f_j$ are expressed as linear combinations of $K_j$ basis functions $b_{j1},\dots,b_{jK_j}$; that is, with coefficients $\beta_{jk}\in\mathbb R$,
$$
f_j(\cdot)=\sum_{k=1}^{K_j}\beta_{jk}b_{jk}(\cdot),\quad j=1,\dots,p.
$$
To satisfy the identifiability condition $\sum_{i=1}^n f_j(x_{ij})=0$, we assume that each basis function satisfies $\sum_{i=1}^n b_{jk}(x_{ij})=0$, $j=1,\dots,p$. This can be achieved by centering unrestricted basis functions 
$b_{jk}^\ast$ as 
\begin{align}
	b_{jk}(\cdot) = b_{jk}^\ast(\cdot) - \frac{1}{n}\sum_{i=1}^n b_{jk}^\ast(x_{ij}), \quad j=1,\dots,p,\quad k=1,\dots, K_j.
	\label{eqn:basiscentering}
\end{align}
Let $B_j\in\mathbb R^{n\times K_j}$ be the matrix whose $(i,k)$th component is $b_{jk}(x_{ij})$. The centering procedure is achieved by the projection $B_j=(I_n-n^{-1}1_n1_n^T)B_j^\ast$ with the unrestricted basis matrix $B_j^\ast$ defined with $b_{jk}^\ast$ for its $(i,k)$th component.
We define  $B=[B_1,\dots,B_p]\in\mathbb R^{n\times J}$ and a vector of full coefficients $\beta=(\beta_{11},\dots,\beta_{1K_1}, \dots ,\beta_{p1},\dots,\beta_{pK_p})^T\in\mathbb R^{J}$, where $J=\sum_{j=1}^p K_j$. 
The vector of additive predictors $\eta=(\eta_1,\dots,\eta_n)^T$ can then be written as $\eta=\alpha 1_n + B\beta$.

Various classes of basis functions can be used to estimate smooth functions. In this study, we employ natural cubic spline basis functions to avoid erratic behavior near the boundaries. This approach is equivalent to using any piecewise polynomial basis function (including B-splines) with appropriate natural boundary conditions, provided that the prior distribution remains invariant under linear bijections of the design matrix.
For boundary knots $\{t^L,t^U\}$ and a set of $M$ interior knots $\{t_1,\dots,t_M\}$ satisfying $-\infty<t^L<t_1<\dots<t_M<t^U<\infty$, we define the natural cubic spline basis functions $N_k:\mathbb R\to \mathbb R$, $k=1,\dots,M+1$, as follows:
\begin{align}
	\begin{split}
		N_1(u)&=u, \\
		N_{k+1}(u)&= N( u; t^L, t^U, t_k)\\
		&\equiv\dfrac{(u-t_k)_+^3 - (u-t^U)_+^3}{t^U - t_k} - \dfrac{(u-t^L)_+^3 - (u-t^U)_+^3}{t^U - t^L},\quad k=1,\dots,M.
	\end{split}
	\label{eqn:ncs}
\end{align}
Combined with the constant term $N_0(u)=1$, the basis functions in \eqref{eqn:ncs} generate piecewise cubic functions. These functions are linear beyond the boundary knots $\{t^L,t^U\}$, enhancing stability near the boundaries owing to the constraints imposed at $\{t^L,t^U\}$. The constant term is excluded from \eqref{eqn:ncs}, as it is redundant given the intercept term. Our findings, consistent with well-known observations, show that natural cubic splines significantly reduce estimation bias near boundaries compared to cubic splines without natural conditions.

Although the basis construction in \eqref{eqn:ncs} is based on a truncated power series, our definition differs slightly from the truncated power natural cubic splines typically used in the literature, such as those in Equations (5.4) and (5.5) of \citet{hastie2009elements}. The basis terms in \eqref{eqn:ncs} span the same piecewise cubic polynomial space with natural boundary conditions, as demonstrated. However, our definition in \eqref{eqn:ncs} has an additional advantageous property: inserting a new knot-point $t_\ast \in(t^L ,t^U)$ simply adds a new basis term $N( \cdot; t^L, t^U, t_\ast)$ to the set $\mathcal N=\{N_k,k=0,1,\dots,M+1\}$ without altering the existing basis terms in $\mathcal N$. Similarly, removing a knot-point simply deletes an existing basis term in $\mathcal N$.
This feature may not be present in other natural cubic spline basis functions, such as the natural cubic B-spline basis or those in Equations (5.4) and (5.5) of \citet{hastie2009elements}, where a single basis term might depend on more than two knots, and adding or removing a knot-point could alter other basis terms. This characteristic makes the basis terms in \eqref{eqn:ncs} more attractive for model selection-based approaches (see Sections~\ref{sec:prior2} and \ref{sec:prior3}) because it allows faster computation by reducing the time spent expanding the design matrix at each iteration. For related simulation results, see Section~S7 of the supplementary material.
To the best of our knowledge, this is the first study to use the form of natural cubic splines in \eqref{eqn:ncs}. These properties are formalized as follows.

\begin{prop} 
	The set $\mathcal N=\{N_k,k=0,1,\dots,M+1\}$ is a basis for the cubic spline space with natural boundary conditions.
	\label{prop:ncspace}
\end{prop}

\begin{prop} 
	The addition of a new interior knot-point $t_\ast \in(t^L ,t^U)$ introduces the corresponding basis term $N(\cdot;t^L,t^U,t_\ast)$ into $\mathcal N$. Similarly, the elimination of an existing interior knot-point $t_k\in t$ eliminates the corresponding basis term $N(\cdot;t^L,t^U,t_k)$ in $\mathcal N$.
	\label{prop:selection}
\end{prop}

Proofs are provided in Section~S2 of the supplementary material.
We select our basis terms $b_{jk}^\ast$ using the natural cubic spline basis functions defined in \eqref{eqn:ncs}. Specifically, for each $j$, we set the boundary knots to $\xi_j^L=\min_{1\le i\le n}x_{ij}$ and $\xi_j^U=\max_{1\le i\le n}x_{ij}$ based on the observed design points. With a given set of knots $\xi_j=\{\xi_{j1},\dots,\xi_{jL_j}\}$ where $\xi_j^L<\xi_{j1}<\dots<\xi_{jL_j}<\xi_j^U$, the uncentered basis terms are chosen as 
\begin{align}
	b_{j1}^\ast(\cdot)&=N_1(\cdot), \quad
	b_{j,k+1}^\ast(\cdot) = N(\cdot ;\xi_j^L,\xi_j^U, \xi_{jk}),\quad  k=1,\dots,L_j.
	\label{eqn:basis}
\end{align} 

The class of spline functions is highly dependent on the placement of knots $\xi = \{\xi_1,\dots,\xi_p\}$. Therefore, choosing suitable knot locations is crucial for accurately capturing both local and global functional characteristics while avoiding overfitting.
From a Bayesian perspective, a natural approach is to let the data select the most appropriate knots $\xi$ from a predetermined set $\Xi$ using BMS. This approach is well-established in the literature  \citep[e.g.,][]{smith1996nonparametric,denison1998automatic,dimatteo2001bayesian,rivoirard2012posterior,de2012adaptive,shen2015adaptive,jeong2016bayesian,jeong2017analysis}.
A set $\Xi$ can be a countable or uncountable collection of knots. A richer $\Xi$ allows for more flexible estimation of regression spline functions but may result in computational inefficiency. Within the Bayesian framework, specifying $\Xi$ via a predetermined law is akin to assigning a prior distribution to $\xi$ over an infinite-dimensional space with restricted support $\Xi$. The key to success is placing a prior on $\xi$ with appropriately restricted support $\Xi$. 
Several options for specifying a prior for $\xi$ are discussed in Section~\ref{sec:knotprior}.

An additional advantage of the formulation in \eqref{eqn:basis} is its ability to easily characterize a fully linear relationship. Specifically, if $\xi_j$ is empty, the basis consists only of the linear term $b_{j1}^\ast$. This is particularly useful when a predictor variable is binary or assumed to have a linear effect. In such cases, we can assign a point mass prior to empty $\xi_j$. As a result, generalized additive partial linear models (GAPLMs) with both parametric and nonparametric additive terms \citep{wang2011estimation} are naturally accommodated by our construction without modification. Additionally, while not explored in this study, variable selection could be incorporated by introducing additional latent variables for the linear basis term $b_{j1}^\ast$. A related idea is discussed in \citet{jeong2021bayesian}.

One major advantage of BMS-based approaches to nonparametric regression is that they provide model-averaged estimates rather than relying on specific knot locations. Our goal is to examine the model-averaged estimates of a functional $\mathcal L:(\alpha,f_1,\dots,f_p)\mapsto \mathcal L(\alpha,f_1,\dots,f_p)$, which is parameterized by coefficients $\alpha$ and $\beta$. For example, we may be interested in a pointwise evaluation of the additive predictor $\alpha+\sum_{j=1}^p f_j(x_j)$ or the univariate function $f_j(x_j)$, $j=1,\dots,p$, at a given point $x=(x_1,\dots,x_p)^T$.
The model-averaged posterior of a functional is given by
\begin{align}
	\label{eqn:functional}
	\pi\big(\mathcal L(\alpha,f_1,\dots,f_p)\mid Y\big) = \int_{\Xi} \pi\big(\mathcal L(\alpha,f_1,\dots,f_p)\mid \xi, Y\big) d\Pi(\xi\mid Y).
\end{align}
A key aspect of our Bayesian procedure is assigning a prior distribution for model selection and exploring the posterior distribution of $\xi$, $\Pi(\xi\mid Y)$.
To highlight the dependency on $\xi$, we use the notation $B_\xi=B$, $\beta_\xi=\beta$, $J_\xi=J$, and $\eta_\xi=\alpha 1_n + B_\xi\beta_\xi$. Note that $J_\xi = p +\sum_{j=1}^p |\xi_j|$, where $|\xi_j|$ represents the number of knots $\xi_j$, $j=1,\dots,p$.

\section{Mixtures of g-priors for generalized additive models}
\label{sec:mixgprior}

Our main objective is to explore the posterior distribution of a functional $\mathcal L(\alpha,f_1,\dots,f_p)$. To obtain a model-averaged estimate, we need to numerically evaluate the integral in \eqref{eqn:functional}, which involves exploring the posterior distribution $\Pi(\alpha,\beta_\xi,\xi\mid Y)$.
Therefore, we need to specify a prior distribution $\Pi(\alpha,\beta_\xi,\xi)$ over the parameter space. The possible priors for $\xi$, $\Pi(\xi)$, are discussed in Section~\ref{sec:knotprior}. A critical aspect is determining a prior for the knot-specific coefficients $\beta_\xi$, that is, $\Pi(\beta_\xi\mid \xi)$. This study employs mixtures of g-priors for this purpose. In this section, we explain the use of mixtures of g-priors in BMS-based approaches to GAMs and discuss the resulting posteriors. Additionally, we provide a toy example to illustrate how mixed priors penalize GAMs.

\subsection{Mixtures of g-priors for exponential family models}
\label{sec:mixgexp}

We specify the prior distribution as $\Pi(\alpha,\beta_\xi\mid\xi) = \Pi(\alpha)\Pi(\beta_\xi\mid \xi)$. In line with common practice, we assign an improper uniform prior to the intercept parameter $\alpha$, that is,
\begin{align}
	\pi(\alpha)\propto 1.
	\label{eqn:alphaprior}
\end{align}
This improper prior has been justified in the literature \citep{berger1998bayes,bayarri2012criteria}.
Next, we discuss $\Pi(\beta_\xi\mid \xi)$.
For model selection in linear regression, Zellner's g-prior is often preferred owing to its computational efficiency and invariance to linear transformations \citep{zellner1986assessing}. 
In our spline setup, this invariance is particularly valuable because it ensures that the procedure remains unaffected by specific choices of basis functions, as long as the target spline space is correctly generated. Therefore, the invariance property of the g-prior supports the spline basis system defined in \eqref{eqn:basis}.
However, the computational advantage of the g-prior is typically diminished in GAMs because Gaussian priors are not conjugate to non-Gaussian models, making it impossible to obtain a closed-form expression for the marginal likelihood $p(Y \mid  \xi)$.
This complicates the computation of the posterior distribution in \eqref{eqn:functional} owing to the intractability of the marginal likelihood. 
To address this issue, we consider approximating the likelihood using the Laplace approximation with a suitable variant of the g-prior. 

Let $\theta=(h\circ b')^{-1}$, and define $\mathcal{J}_n(\hat{\eta}_\xi) = \text{diag}(- Y_i \theta''(\hat{\eta}_{\xi,i}) + (b \circ \theta)''(\hat{\eta}_{\xi,i}),i=1,\dots, n)$ as the observed information matrix of ${\eta}_\xi$ evaluated at $\hat{\eta}_\xi$ (the Hessian matrix of the negative log-likelihood), where $\hat{\eta}_{\xi} =(\hat{\eta}_{\xi,1},\dots, \hat{\eta}_{\xi,n})^T = \hat\alpha_\xi 1_n + B_\xi \hat\beta_\xi$ with the maximum likelihood estimators $\hat\alpha_\xi$ and $\hat\beta_\xi$ (assuming they exist). 
We focus on cases where $\mathcal{J}_n(\hat{\eta}_\xi)$ is positive definite, which is generally true except in extreme situations like complete separation in logistic regression \citep{li2018mixtures}. Among the variants of the g-prior for exponential family models, we use the form proposed by \citet{li2018mixtures},
\begin{align}\label{eqn:gprior}
	\beta_\xi\mid g, \xi &\sim \text{N}\big(0, g (\tilde B_\xi^T \mathcal{J}_n(\hat{\eta}_\xi)\tilde B_\xi)^{-1} \big), 
\end{align}
where $g>0$ serves as a dispersion factor that controls the influence of the prior, and $\tilde B_\xi=[I_n-\text{tr}(\mathcal{J}_n(\hat\eta_\xi))^{-1}1_n1_n^T \mathcal{J}_n(\hat\eta_\xi)]B_\xi$ is the matrix consisting of the columns of $B_\xi$ centered by the weighted average with the diagonal elements of $\mathcal{J}_n(\hat{\eta}_\xi)$. The prior in \eqref{eqn:gprior} requires that $\tilde B_\xi^T \mathcal{J}_n(\hat{\eta}_\xi)\tilde B_\xi$ be invertible. This condition is satisfied if and only if $B_\xi$ has full-column rank (observe that $\mathcal{J}_n(\hat{\eta}_\xi)$ is positive definite and $\text{rank}(B_\xi)=\text{rank}(\tilde B_\xi)$, where $\text{rank}(\cdot)$ is the rank of a matrix). Thus, a full-column rank condition will be imposed on $\Pi(
\xi)$ in Section~\ref{sec:knotprior}. 
Although the prior in \eqref{eqn:gprior} could be extended using a generalized inverse, we do not pursue this approach here (for further discussion, see Section~2.5 of \citet{li2018mixtures}).

In addition to the prior in \eqref{eqn:gprior}, many other variants of the g-prior exist for exponential family models \citep[e.g.,][]{hansen2003minimum,wang2007adaptive,gupta2009information,bove2011hyper,held2015approximate}. 
We note that the prior in \eqref{eqn:gprior} depends on the observation vector $Y$, which means it does not strictly adhere to the pure Bayesian philosophy. 
Some methods address this issue by using the expected information matrix instead of $\mathcal J_n(\hat{\eta}_\xi)$, while substituting $\eta_\xi=\alpha 1_n$ based on the null model \citep{bove2011hyper, held2015approximate,castellanos2021model, garcia2023model}. 
However, within this framework, the marginal likelihood is not available in closed form unless $\alpha$ is fixed and a specific prior on $g$ is used \citep{bove2011hyper, held2015approximate}. In contrast, the prior in \eqref{eqn:gprior} provides a convenient expression for the approximate marginal likelihood, enabling relatively fast computation.
Moreover, our prior captures
the large-sample covariance structures and local geometry better than other variants of the g-prior \citep{li2018mixtures}.

By integrating the second-order Taylor expansion of the likelihood with the priors specified in \eqref{eqn:alphaprior} and \eqref{eqn:gprior}, we obtain
\begin{align}\label{eqn:fixedg}
	p(Y\mid g, \xi) 
	&\approx p(Y\mid \hat\eta_\xi)\text{tr}(\mathcal{J}_n(\hat{\eta}_\xi))^{-1/2}
	(g+1)^{-{J_\xi}/{2}} 
	\exp \!\left(-\frac{Q_\xi}{2(g+1)}\right),
\end{align}
where $p(Y\mid \hat\eta_\xi)$ represents the likelihood evaluated at $\hat\eta_\xi$ for a given $\xi$ and $Q_\xi = \hat\beta_\xi^T \tilde B_\xi^T \mathcal{J}_n(\hat{\eta}_\xi)\tilde B_\xi\hat\beta_\xi$ is the Wald statistic; see Section~S3 of the supplementary material for the derivation of \eqref{eqn:fixedg}.
The expression in \eqref{eqn:fixedg} shows that when $g$ is treated as a fixed hyperparameter, the marginal likelihood becomes highly sensitive to its value. Determining an appropriate choice for $g$ has been widely discussed in the literature. The most common approach is to set $g=n$, known as the unit information prior \citep{kass1995reference}. 
This concept is also frequently used in the literature on nonparametric regression using BMS \citep[e.g.][]{gustafson2000bayesian, dimatteo2001bayesian, kohn2001nonparametric}.
From a Bayesian perspective, the unit information prior can be viewed as a point mass prior at $g=n$, expressed as $\Pi(g)=\delta_n(g)$, where $\delta_b$ denotes the Dirac measure at $b$.
However, research has shown that using a suitable prior distribution for $g$, known as a mixture of g-priors, enhances empirical performance and addresses paradoxes in BMS \citep{liang2008mixtures,li2018mixtures}.
To unify various mixtures of g-priors, we adopt a general family that encompasses various mixture distributions. Specifically, following \citet{li2018mixtures}, we assign the truncated compound confluent hypergeometric (tCCH) distribution to $(g+1)^{-1}$ \citep{gordy1998generalization}, that is,
\begin{align}
	\frac{1}{g+1} \sim \text{tCCH}\bigg(\frac{a}{2},\frac{b}{2},r,\frac{s}{2},\nu,\kappa\bigg),\quad a,b,\kappa>0, \quad r,s\in\mathbb R,\quad \nu\ge1.
	\label{eqn:tcchprior}
\end{align}
The tCCH distribution is a type of generalized beta distribution characterized by five parameters, which allow it to exhibit multi-modal or long-tailed density. Parameters $a$ and $b$ behave similarly to those in a beta distribution, while parameters $r$, $s$, and $\kappa$ control the skewness of the density. Parameter $\nu$ determines the support of the distribution. For a detailed discussion, including the density function and moments of the tCCH distribution, see Section~S1 of the supplementary material.

\begin{table}[t!]
	\centering
	\begin{tabular}{lccccccc}
		\hline
		& $a$ & $b$ & $r$ & $s$ & $\nu$ & $\kappa$ & Concentration\\ 
		\hline
		Uniform & $2$ & $2$ & $0$ & $0$ & $1$ & $1$ & $g=O(1)$\\
		Hyper-g & $1$ & $2$ & $0$ & $0$ & $1$ & $1$ & $g=O(1)$\\
		Hyper-g/n & $1$ & $2$ & $1.5$ & $0$ & $1$ & $n^{-1}$ & $g=O(n)$\\
		Beta-prime & $0.5$ & $n-J_\xi - 1.5$ & $0$ & $0$ & $1$ & $1$ & $g=O(n)$\\
		ZS-adapted & $1$ & $2$ & $0$ & $n+3$ & $1$ & $1$ & $g=O(n)$\\
		Robust & $1$ & $2$ & $1.5$ & $0$ & $\frac{n+1}{J_\xi+1}$ & 1 & $g=O(n)$ \\
		Intrinsic  & $1$ & $1$ & $1$ & $0$ & $\frac{n+J_\xi+1}{J_\xi+1}$ & $\frac{n+J_\xi+1}{n}$ & $g=O(n)$ \\
		\hline
	\end{tabular}
	\caption{Distributions belonging to the tCCH family.}
	\label{table:mixg}
\end{table}

Table~\ref{table:mixg} presents several distributions from the tCCH family, including the uniform prior (on $(g+1)^{-1}$), the hyper-g and hyper-g/n priors \citep{liang2008mixtures}, the beta-prime prior \citep{maruyama2011fully}, the Zellner Siow (ZS)-adapted prior \citep{held2015approximate}, the robust prior \citep{bayarri2012criteria}, and the intrinsic prior \citep{womack2014inference}.
Note that the beta-prime prior is only proper if $J_\xi <n-1$, so this constraint needs to be incorporated into $\Pi(\xi)$ when using the beta-prime prior.
According to \citet{li2018mixtures}, prior distributions can be classified into two categories based on their concentration: $g=O(1)$ and $g=O(n)$. (This notation can be misleading, as it refers to the concentration order of the distribution rather than the actual value of $g$; \citet{maruyama2011fully} uses the same notation.) Figure~\ref{plot:gprior_log} illustrates the concentration behavior of each prior distribution on $g$.

\begin{figure}[t!]
	\centering
	\includegraphics[width = 11.5cm]{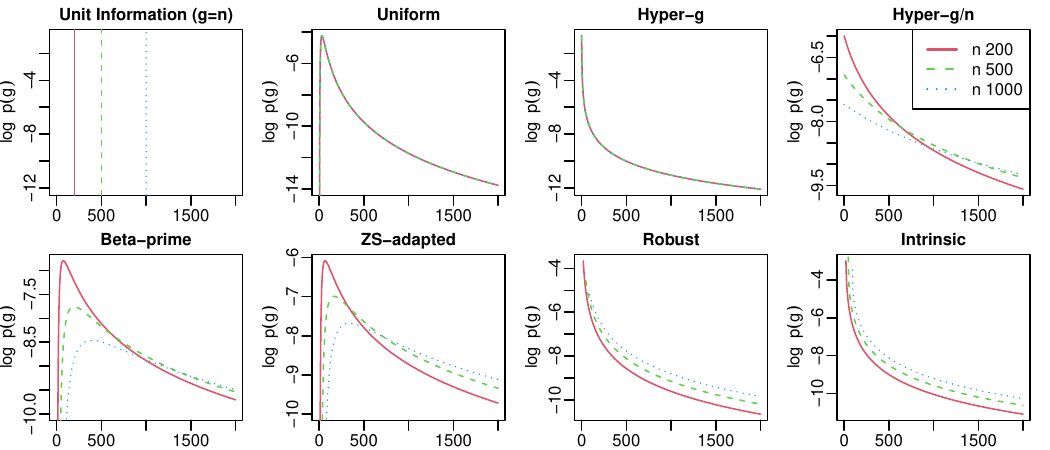}
	\caption{Distributions belonging to the tCCH family for $n=200, 500, 1000$, with $J_\xi = 10$ if required.}
	\label{plot:gprior_log}
\end{figure}

We define the confluent hypergeometric function of two variables \citep{gordy1998generalization} as
$\Phi_1(\alpha,\beta,\gamma, x, y) = B(\alpha, \gamma - \alpha)^{-1} \int_0^1 u^{\alpha-1} (1-u)^{\gamma-\alpha-1}(1-yu)^{-\beta} e^{xu}du$, for $\gamma>\alpha>0$, $\beta>0$, $x\in\mathbb R$, and $y<1$.\footnote{These parameter ranges ensure that $\Phi_1$ is finite, positive, and real; see Theorem~1 of \citet{gordy1998generalization}.} The resulting marginal likelihood is expressed as 
\begin{align}\label{eqn:mixedg}
	\begin{split}
		p(Y\mid \xi)
		&= p(Y\mid\hat{\eta}_{\xi}) \text{tr}(\mathcal{J}_n(\hat{\eta}_\xi))^{-1/2} \nu^{-J_\xi / 2} \exp \!\left(-\frac{Q_\xi}{2\nu}\right)\frac{B((a+J_\xi)/2, b/2) }{B(a/2, b/2) }\\
		&\quad\times\Phi_1\!\left(\frac{b}{2}, r, \frac{a+b+J_\xi}{2}, \frac{s+Q_\xi}{2\nu}, 1-\kappa\right)   \bigg/ \Phi_1\!\left(\frac{b}{2}, r, \frac{a+b}{2}, \frac{s}{2\nu}, 1-\kappa\right),
	\end{split}
\end{align}
where $B(\cdot,\cdot)$ denotes the beta function. The derivation of \eqref{eqn:mixedg} is detailed in Section~S3 of the supplementary material.
Generally, $\Phi_1$ cannot be evaluated analytically and requires numerical approximation. For this purpose, we utilize the Gaussian-Kronrod quadrature routine available in the Boost \texttt{C++} library.

The approximate posterior for $((g+1)^{-1},\alpha,\beta_\xi)$ conditional on $\xi$ is given by
\begin{align}
	\begin{split}
		\frac{1}{g+1} \mid Y, \xi &\sim \text{tCCH}\!\left(\frac{a+J_\xi}{2}, \frac{b}{2}, r, 
		\frac{s+Q_\xi}{2}, \nu, \kappa\right),\\
		\beta_\xi \mid Y, g,\xi &\sim 
		\text{N}\!\left(\frac{g}{g+1}\hat\beta_\xi, \frac{g}{g+1}(\tilde B_\xi^T J(\hat\eta_\xi) \tilde B_\xi)^{-1}\right),\\
		\alpha \mid Y, g, \beta_\xi,\xi &\sim 
		\text{N}\!\left(\hat \alpha_\xi - \text{tr}(\mathcal{J}_n(\hat\eta_\xi))^{-1}1_n^T \mathcal{J}_n(\hat\eta_\xi) B_\xi(\beta_\xi-\hat\beta_\xi), \text{tr}(\mathcal{J}_n(\hat{\eta}_\xi))^{-1}\right).
	\end{split}
	\label{eqn:posterior}
\end{align}
The derivation of \eqref{eqn:posterior} is detailed in Section~S3 of the supplementary material.
This expression is also applicable for the unit information prior by substituting the first line with the point mass posterior $\Pi(g \mid Y,\xi)=\delta_n(g)$.
Sampling from tCCH distributions can be performed using MCMC, but exact sampling is possible with certain prior specifications. Specifically, if the uniform prior, hyper-g prior, ZS-adapted prior, or robust prior is used, the first line of \eqref{eqn:posterior} simplifies to a truncated gamma distribution, making  exact sampling straightforward.
For the remaining priors, slice sampling with data augmentation can be employed. Details on the sampling procedures are provided in Section~S4 of the supplementary material.
The joint posterior $\Pi(\alpha,\beta_\xi,\xi,g\mid Y)$ is fully specified by the posterior in \eqref{eqn:posterior} and the marginal posterior of $\xi$, $\Pi(\xi\mid Y)$.
The latter is obtained by specifying a prior $\Pi(\xi)$ as described in Section~\ref{sec:knotprior} and using the approximate marginal likelihood $p(Y\mid \xi)$ in \eqref{eqn:mixedg} (or $p(Y\mid g,\xi)$ in \eqref{eqn:fixedg} for the unit information prior). The posterior distribution of a functional in \eqref{eqn:functional} can then be evaluated either by directly marginalizing $\xi$ or by using MCMC for Monte Carlo integration of $\xi$, depending on the prior specified for $\xi$ in Section~\ref{sec:knotprior}.

\subsection{Behavior of the Bayes factor}
\label{sec:penalty}

The choice of $g$ is crucial for achieving appropriate sparsity in model selection with the g-prior \citep{kass1995bayes}. A large value of $g$ tends to favor sparse models, while a small value of $g$ supports more complex models. This choice is particularly important in our additive model setup, as it directly influences the smoothness of the additive functions. In the literature on nonparametric regression with basis expansion, many studies use the unit information prior, which corresponds to setting $g=n$ \citep[e.g.][]{gustafson2000bayesian, dimatteo2001bayesian, kohn2001nonparametric}.
However, as noted earlier, a mixture of g-priors can offer improved empirical performance in BMS \citep{liang2008mixtures,li2018mixtures}. 
While attempts have been made to assign a prior to $g$ in nonparametric regression \citep{jeong2016bayesian, jeong2017analysis, francom2018sensitivity, francom2020bass, jeong2021bayesian}, a thorough investigation into how these approaches differ from the unit information prior is still lacking.
In this section, we explore how mixtures of g-priors compare to the unit information prior  and discuss why the unit information prior might not be the optimal choice for estimating GAMs.

Our investigation utilizes Bayes factors.
For two sets of knots $\xi_{(1)}$ and $\xi_{(2)}$, the Bayes factor of $\xi_{(1)}$ to $\xi_{(2)}$ is defined as $	BF[\xi_{(1)};\xi_{(2)}]=p(Y \mid \xi_{(1)})/p(Y \mid \xi_{(2)})$.
For exponential family models with a known $\phi$, the marginal likelihood $p(Y\mid \xi)$ is given by \eqref{eqn:fixedg} with $g=n$ for the unit information prior and by \eqref{eqn:mixedg} for mixtures of g-priors induced by tCCH priors on $(g+1)^{-1}$.
To understand how the Bayes factor penalizes model complexity, we consider two knots $\xi_{(1)}$ and $\xi_{(2)}$ such that $J_{\xi_{(1)}} = J_{\xi_{(2)}}+1$ and $\hat\eta_{\xi_{(1)}}=\hat\eta_{\xi_{(2)}}$. In other words, both knots contribute equally to the model fit, but $\xi_{(1)}$ has one additional redundant knot-point compared to $\xi_{(2)}$. 
Accordingly, the Bayes factor satisfies $ BF[\xi_{(1)};\xi_{(2)}]<1$, indicating that the larger model $\xi_{(1)}$ is never preferable over the smaller model $\xi_{(2)}$ owing to the same model fit. The Bayes factor $ BF[\xi_{(1)};\xi_{(2)}]$ quantifies the relative preference for the larger model $\xi_{(1)}$ over the smaller model $\xi_{(2)}$. For example, if $ BF[\xi_{(1)};\xi_{(2)}]=1/2$, the larger model $\xi_{(1)}$ is only half as preferred as the smaller model $\xi_{(2)}$ (or equivalently, the smaller model $\xi_{(2)}$ is twice as preferred). 
As $BF[\xi_{(1)};\xi_{(2)}]$ approaches 1, the preference for the two models becomes equal.

We examine how the Bayes factor behaves with changes in $J_{\xi_{(1)}}$ and the goodness-of-fit. In Gaussian regression, the goodness-of-fit is naturally assessed using the coefficient of determination. For the exponential family models, the pseudo-$R^2$, defined as $1-\exp(-D/n)$ with the usual deviance statistic $D$, can be used alternatively \citep{cox1989analysis, magee1990r}, with the caveat that its maximum value may be less than 1 depending on the specific model \citep{nagelkerke1991note}. To relate the Bayes factor to the pseudo-$R^2$, we use the fact that $Q_\xi$ is asymptotically equivalent to the deviance $D$ under mild conditions \citep{held2015approximate, li2018mixtures}. Therefore, we define $R_{\xi,\text{pseudo}}^2=1-\exp(-Q_\xi/n)$ to gauge the goodness-of-fit for exponential family models.

	\begin{figure}[t!]
		\centering 
		\includegraphics[width=10cm]{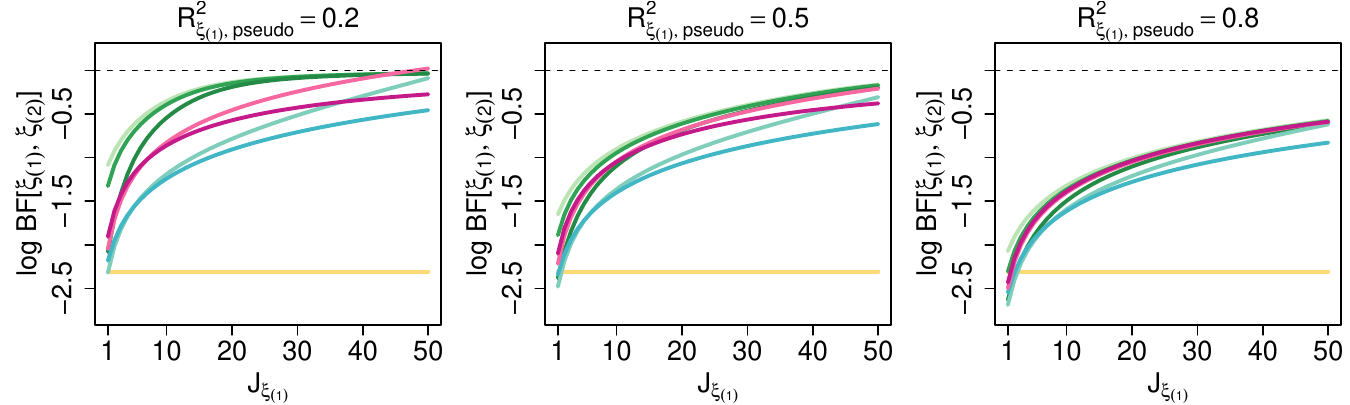} \\
		\smallskip
		\includegraphics[width=10cm]{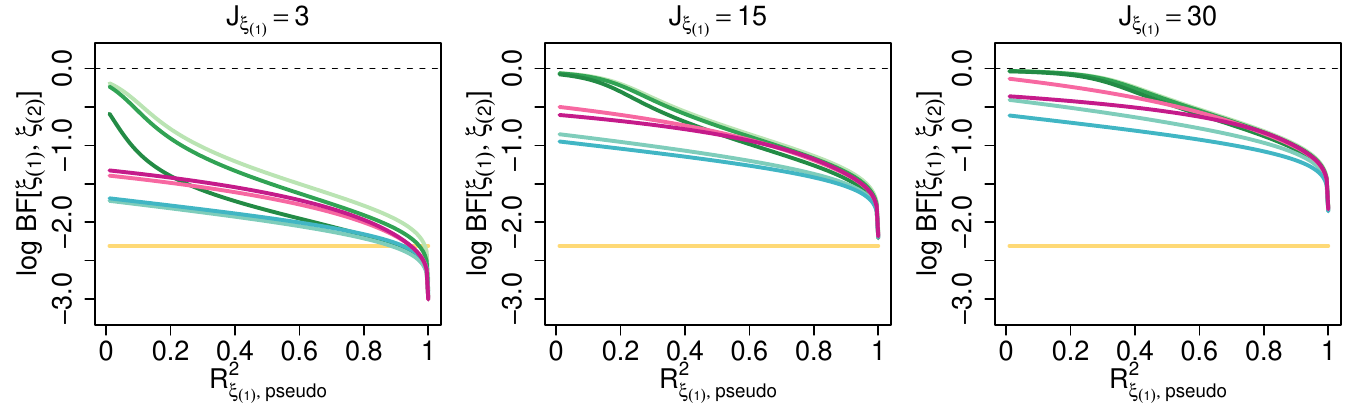}\\
		\smallskip
		\includegraphics[width=11cm]{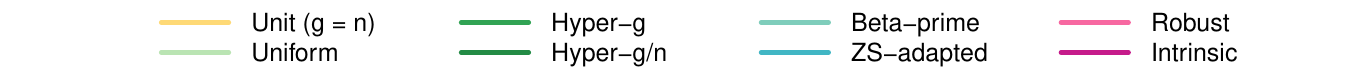}
		\caption{Change in $\log BF[\xi_{(1)};\xi_{(2)}]$ as a function of $J_{\xi_{(1)}}$ ($=J_{\xi_{(2)}}+1$) and $R_{\xi_{(1)},\text{pseudo}}^2$ ($=R_{\xi_{(2)},\text{pseudo}}^2$) for $n=1000$. The black dashed lines denote zero values, indicating equal preference for $\xi_{(1)}$ and $\xi_{(2)}$, i.e., $BF[\xi_{(1)};\xi_{(2)}]=1$.}
		\label{plot:logBFexp}
	\end{figure}

Figure~\ref{plot:logBFexp} presents a toy example with $n=1000$, illustrating how $\log BF[\xi_{(1)};\xi_{(2)}]$ changes with $J_{\xi_{(1)}}$ (where $J_{\xi_{(1)}} = J_{\xi_{(2)}} + 1$) and $R_{\xi_{(1)},\text{pseudo}}^2$ (where $R_{\xi_{(1)},\text{pseudo}}^2 = R_{\xi_{(2)},\text{pseudo}}^2$). 
Similar patterns were observed with other values of $n$.
The unit information prior consistently produces a constant Bayes factor, regardless of $J_{\xi_{(1)}}$ and $R_{\xi_{(1)},\text{pseudo}}^2$. 
 In contrast, the first row of
Figure~\ref{plot:logBFexp} shows that mixture priors cause $\log BF[\xi_{(1)};\xi_{(2)}]$ to increase as the model size $J_{\xi_{(1)}}$ increases. This indicates that when comparing two small models (i.e., models with both small $J_{\xi_{(1)}}$ and $J_{\xi_{(2)}}$), mixtures of g-priors significantly penalize the larger model $\xi_{(1)}$ unless the marginal likelihood exhibits a notable improvement. 
Conversely, when comparing two large models (i.e., models with relatively large $J_{\xi_{(1)}}$ and $J_{\xi_{(2)}}$), the larger model $\xi_{(1)}$ is less likely to be penalized, even in the absence of a substantial benefit. This property of mixture priors can enhance GAM estimation, as it needs the comparison of large models generated through basis expansion to detect both local and global signals in the target functions that might be otherwise overlooked.
The second row of Figure~\ref{plot:logBFexp} shows that mixture priors cause $\log BF[\xi_{(1)};\xi_{(2)}]$ to decrease as $R_{\xi_{(1)},\text{pseudo}}^2$ increases. This aligns with intuition: with a sufficiently high goodness-of-fit, a more complex model may not be necessary, and a simpler model is often preferable unless it provides a significant improvement in the marginal likelihood. 
The unit information prior does not account for these characteristics in GAM estimation.

Figure~\ref{plot:logBFexp} illustrates the differences between mixtures of g-priors. The beta-prime and ZS-adapted priors show similar behavior, with the smallest values of $\log BF[\xi_{(1)};\xi_{(2)}]$, indicating the weakest inclination towards the larger model $\xi_{(1)}$. In contrast, the robust and intrinsic priors exhibit comparable decay patterns and result in larger values of $\log BF[\xi_{(1)};\xi_{(2)}]$, reflecting a stronger relative preference for $\xi_{(1)}$ compared to the beta-prime and ZS-adapted priors. 
The two $O(1)$-type priors (uniform and hyper-g) demonstrate a more pronounced preference for $\xi_{(1)}$ compared to the $O(n)$-type priors. Interestingly, the hyper-g/n prior, although categorized as an $O(n)$ prior, behaves similarly to the 
$O(1)$-type priors.
Based on thee discussion in the preceding paragraph, the latter three mixture priors may be mistakenly considered suitable for GAM estimation. However, when $R_{\xi_{(1)},\text{pseudo}}^2$ is small, these priors tend to drive $\log BF[\xi_{(1)};\xi_{(2)}]$ towards zero. This suggests that the preferences for the smaller and larger models may become undesirably similar, which may lead to overfitting. In contrast, other mixture priors appear to be less affected by this issue. The question of which mixture prior performs best for GAMs remains unresolved. Our numerical studies in Section~\ref{sec:sims} suggest that robust and intrinsic priors are the most effective. The following proposition provides a basic interpretation of where the differences among mixtures of g-priors may arise.

\begin{prop}
	For the model in \eqref{eqn:glm} and \eqref{eqn:gam} with the priors in \eqref{eqn:alphaprior} and \eqref{eqn:gprior}, consider two knots $\xi_{(1)}$ and $\xi_{(2)}$ such that $J_{\xi_{(1)}} = J_{\xi_{(2)}}+k$ and $\hat\eta_{\xi_{(1)}}=\hat\eta_{\xi_{(2)}}$, where $k$ is a positive integer.
	For any positive integer $k$,
	\begin{align*}
		BF[\xi_{(1)};\xi_{(2)}] =
		\begin{cases} (1+b)^{-k/2}, & \text{if $g=b$},\\ E[(1+g)^{-k/2}\mid\xi_{(2)}, Y], & \text{if $g$ has a tCCH prior}.
		\end{cases}
	\end{align*}
	\label{prop:bf}
\end{prop}
The proof can be found in Section~S2 of the supplementary material. This proposition implies that the Bayes factor $BF[\xi_{(1)};\xi_{(2)}]$ represents the conditional posterior mean of $(1+g)^{-k/2}$, as induced by the unit information prior or tCCH priors. 
The proposition clarifies why the Bayes factor with the unit information prior remains constant.
Differences in Bayes factors with mixture priors arise from variations in the posterior means of the shrinkage factor $(1+g)^{-k/2}$.

In conjunction with a specified prior for knots, $\Pi(\xi)$, the actual model comparison for determining the basis terms relies on the posterior odds $\Pi(\xi_{(1)}\mid Y)/\Pi(\xi_{(2)}\mid Y)$, rather than solely on the Bayes factor. 
Instead of employing mixtures of g-priors, one may consider using the unit information prior and adjusting the posterior odds with a suitable prior $\Pi(\xi)$. However, this approach is generally less favorable. This is because, for the posterior odds to reflect changes in goodness-of-fit with the unit information prior, $\Pi(\xi)$ would need to be excessively data-dependent. Thus, using a mixture of g-priors along with standard priors for knots is a more natural and practical choice.

\section{Priors for knots}
\label{sec:knotprior}

A prior $\Pi(\alpha,\beta_\xi\mid\xi)$ on the coefficients was specified in Section~\ref{sec:mixgprior}. To complete the Bayesian framework, we need to specify $\Pi(\xi)$ for the knots. 
Our prior for $\beta_\xi$ requires that $B_\xi$ be of full-column rank (see \eqref{eqn:gprior} above). Therefore, we choose $\Pi(\xi)$ under the condition that $B_\xi$ has full-column rank with a prior probability of one; that is, $\Pi(\text{rank}(B_\xi)=J_\xi)=1$.
This condition is typically satisfied by ensuring $J_\xi < n$, provided the knots and design points are well distributed.

Intuitively, $\xi_j$ can be any set of singletons within the interval $(\xi_j^L,\xi_j^U)$, indicating that the intrinsic parameter space for $\xi_j$ is infinite-dimensional. However, for computational reasons, a finite truncation to a restricted support may be beneficial. 
As previously mentioned, we denote $\Xi$ as the induced support for $\Pi(\xi)$. The support $\Xi$ restricts the function class generated by the natural cubic spline basis terms. A smaller space reduces model complexity but may fail to capture both local and global features of the target function. Thus, the choice of restricted support $\Xi$ is crucial for balancing estimation quality and computational efficiency. Various strategies have been proposed for specifying $\Xi$ for $\Pi(\xi)$. In this section, we discuss widely accepted methods for constructing $\Xi$, classifying them into three categories. These approaches are described in detail in Sections~\ref{sec:prior1}--\ref{sec:prior3}. Figure~\ref{plot:plotknots} provides a graphical summary of these strategies. 
A comparison of the empirical performance of these three approaches is presented in Section~\ref{sec:sims} based on a numerical study.

\begin{figure}[t!]
	\centering
	\includegraphics[width = 10cm]{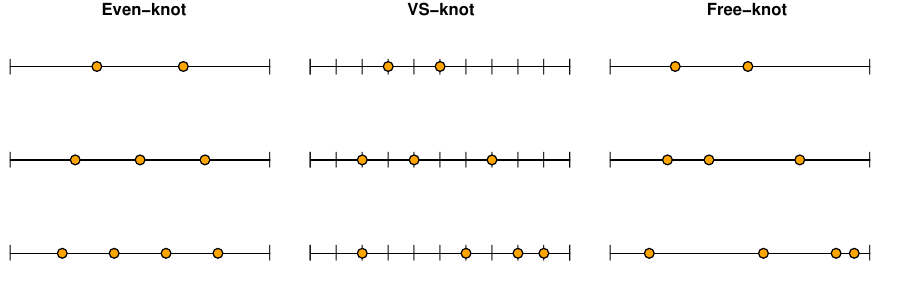}
	\caption{A graphical illustration of the three strategies for constructing $\Xi$ discussed in Sections~\ref{sec:prior1}--\ref{sec:prior3}. For even-knot splines, the locations of knots are deterministically ascertained once $|\xi_j|$ is chosen. VS-knot splines select knot-points from a pre-determined set of locations. Free-knot splines are the most flexible and have no such limitation.}
	\label{plot:plotknots}
\end{figure}

\subsection{Even-knot splines: equidistant knots}
\label{sec:prior1}

The simplest yet powerful Bayesian adaptation arises from the assumption that the number of knots is not fixed but their locations are determined by an intrinsic law.
The idea has been extensively considered in the literature and has been empirically and theoretically successful  \citep[e.g.,][]{rivoirard2012posterior,de2012adaptive,shen2015adaptive}. We refer to this approach as {\em even-knot splines}. The name should be carefully understood because evenness may be assessed by an empirical measure rather than a geometric distance.

	In this approach, a prior is assigned a number $|\xi_j|$ of knots $\xi_j$ for $j=1,\dots, p$.
The specific locations of these knots are then determined based on a predefined rule. For example, for a given number $|\xi_j|$, the knots $\xi_j$ may be equally spaced or chosen based on the quantiles of the design points $x_{ij}$, $i=1,\dots,n$. We prefer the latter approach for its stability. 
Using the quantiles also ensures that $B_\xi$ has full-column rank, provided that $J_\xi< n$ and the design points $x_{ij}$, $i=1,\dots,n$, are distinct. 
To avoid issues with duplicated quantile values for discretized design points, we recommend using only the unique quantile values.
For computational efficiency, limiting each $|\xi_j|$ such that $|\xi_j|\le M_j$ for a predetermined $M_j$, is computationally useful. The induced support is defined as 
\begin{align*}
	\Xi_{EK}=\Big\{\xi:  \text{rank}(B_\xi)=J_\xi, \, |\xi_j|\le M_j, \, \xi_{jk} = Q_{jk}, \, j=1,\dots,p,\, k=1,\dots,|\xi_j| \Big\},
\end{align*}
where $Q_{jk}$, $k=1,\dots,|\xi_j|$, are the unique quantile values of $x_{ij}$, $i=1,\dots,n$.\footnote{These unique quantile values are obtained by removing duplicates from the usual quantiles of $x_{ij}$, $i=1,\dots,n$, with equal probability. When ties are absent, they correspond to the usual quantiles.} Examples of knots in $\Xi_{EK}$ are shown in Figure~\ref{plot:plotknots}. 
 With a density $q_j:\mathbb \{0,1,\dots,M_j\}\to (0,\infty)$ on $|\xi_j|$, the prior can be formally expressed as
\begin{align}
	\pi_{EK}(\xi)\propto \prod_{j=1}^p q_j(|\xi_j|), \quad \xi \in \Xi_{EK}.
	\label{eqn:prior1}
\end{align}
Further discussion of the density $q_j$ is presented in Section~\ref{sec:q-density}.

The key advantage of the prior in \eqref{eqn:prior1} is its low model complexity. Specifically, for a moderately large $p$, all possible models can be enumerated because $|\Xi_{EK}|\le \prod_{j=1}^p (1+M_j)$. This allows for MCMC-free posterior computation in relatively low-dimensional problems. If $p$ is too large to enumerate all possibilities, the Metropolis-Hastings algorithm can be employed to explore the model space by proposing changes in $|\xi_j|$. In such cases, the computation can be streamlined by storing the value of the marginal likelihood $p(Y\mid \xi)$ for the current $\xi$ and reusing it when the same $\xi$ is revisited. We observe that this storage approach is effective unless $p$ is extremely large.

Despite its advantages, the even-knot spline approach has a significant drawback due to its deterministic rules. Specifically, it cannot accommodate functions with spatially adaptive smoothness, such as Doppler functions. This limitation highlights the need for a more flexible construction, which is addressed in the following two subsections.

\subsection{VS-knot splines: knot selection}
\label{sec:prior2}

The limitations of even-knot splines described in Section~\ref{sec:prior1} can be mitigated by using a prior that induces a richer $\Xi$ allowing for spatial adaptation. This can be achieved by allowing knot placement as well as the number of knots to be data-driven. A common approach is to set a large set of candidate basis functions and select the most important ones using Bayesian variable selection. This idea was introduced by \citet{smith1996nonparametric} and has been widely adopted in the literature on nonparametric regression \citep[e.g.,][]{kohn2001nonparametric,chan2006locally,jeong2016bayesian,jeong2017analysis,park2018analysis,jeong2021bayesian}. We refer to this approach as {\em VS-knot splines}.

Consider a set $\xi_j^c=\{\xi_{j1}^c,\dots, \xi_{jM_j}^c\}$ of knot candidates such that $\xi_j^L<\xi_{j1}^c<\dots<\xi_{jM_j}^c<\xi_j^U$ with large enough $M_j<n$. Similar to Section~\ref{sec:prior1}, the candidates $\xi_j^c$ can be equidistant or determined using the unique values of the quantiles of $x_{ij}$, $i=1,\dots,n$. We prefer the latter setup for its stability. The actual knots $\xi_j$ are selected as a subset of $\xi_j^c$ (including an empty set) using BMS. Consequently, the support consists of all possible subsets of $\{\xi_1^c,\dots,\xi_p^c\}$ with the restriction $\text{rank}(B_\xi)=J_\xi$, that is
\begin{align*}
	\Xi_{VS}=\Big\{\xi:  \text{rank}(B_\xi)=J_\xi, \, \xi_j\subset \xi_j^c, \, j=1,\dots,p \Big\}.
\end{align*}
As in Section~\ref{sec:prior1}, we assign the density $q_j:\mathbb \{0,1,\dots,M_j\}\to (0,\infty)$ to $|\xi_j|$. We then assign equal weights to all knot locations conditional on $|\xi_j|$.
The resulting prior is
\begin{align}
	\pi_{VS}(\xi)\propto\prod_{j=1}^p q_j(|\xi_j|)\binom{M_j}{|\xi_j|}^{-1},\quad \xi \in \Xi_{VS}.
	\label{eqn:prior2}
\end{align}

The VS-knot spline approach has proven effective in adapting to spatially inhomogeneous smoothness \citep[e.g.,][]{chan2006locally,jeong2016bayesian,jeong2017analysis}.
The cardinality $|\Xi_{VS}|\le 2^{\sum_{j=1}^p M_j}$ indicates that enumerating all possible models is usually impractical, highlighting the usefulness of MCMC methods for exploring model spaces. Standard Gibbs sampling and  Metropolis-Hastings algorithms are well-suited for this setup \citep{dellaportas2002bayesian}. 
Sampling efficiency can be enhanced using block updates \citep{kohn2001nonparametric,jeong2021bayesian} or adaptive sampling \citep{nott2005adaptive,ji2013adaptive}. Additionally, since $\Xi_{VS}$ is finite-dimensional, storing the marginal likelihood, as discussed in Section~\ref{sec:prior1}, appears feasible. However, our experience shows that this approach is only effective when $p$ is very small, due to memory constraints (e.g., $p\le 2$). Consequently, we do not pursue this direction.

We emphasize that the basis system in \eqref{eqn:basis} is particularly useful for the VS-knot spline approach. According to Proposition~\ref{prop:selection},
knot selection naturally translates into basis selection.
This property simplifies computation using the basis system in \eqref{eqn:ncs} and \eqref{eqn:basis}: one can generate a full basis matrix $B_j^c\in\mathbb R^{n\times (M_j+1)}$ whose $(i,k)$th component is $b_{jk}(x_{ij})$ constructed with the knot candidates $\xi_j^c=(\xi_{j1}^c,\dots, \xi_{jM_j}^c)$, and then choose important columns of $B_j^c$, while always including the first column for the linear term.
As previously noted, this approach cannot be applied to other natural cubic spline basis functions, such as natural cubic B-splines or those in (5.4) and (5.5) of \citet{hastie2009elements}. For these basis functions, the basis term $b_{j,k+1}^\ast$ may be specified with more than one knot-point for some $k$. Consequently, inserting or deleting a knot-point may alter multiple basis terms, leading to a conflict between knot selection and basis selection.

\subsection{Free-knot splines}
\label{sec:prior3}

The VS-knot spline strategy discussed in Section~\ref{sec:prior2} selects important knot locations from a set of predetermined candidates. As a result, the knots are not equally spaced, which allows for spatially varying degrees of smoothness. Despite this flexibility, further relaxation of the restriction imposed by the discrete set of knot candidates remains a topic of interest. This can be achieved with a fully nonparametric approach by allowing knots to be any singleton set within the specified range, provided that the induced $B_\xi$ is of full-column rank. 
This approach is known as {\em free-knot splines} \citep{denison1998automatic,dimatteo2001bayesian}.

As in Section~\ref{sec:prior1}, capping each $|\xi_j|$ so that $|\xi_j|\le M_j$ for a predetermined $M_j$ can be computationally beneficial.
The resulting support for $\xi$ is
\begin{align*}
	\Xi_{FK}=\Big\{\xi: \text{rank}(B_\xi)=J_\xi,  \, |\xi_j|\le M_j , \, \xi_j^L<\xi_{j1}<\dots<\xi_{j|\xi_j|}<\xi_j^U, \, j=1,\dots,p \Big\}.
\end{align*}
Clearly, the set $\Xi_{FK}$ is uncountable.
The prior is specified similarly to the one in \eqref{eqn:prior2}. However, because the mapping $|\xi_j|\mapsto\xi_j$ is a surjection rather than a bijection, the conditional prior density of $\xi_j$ given $|\xi_j|$, denoted by $\tilde q_j(\cdot\mid |\xi_j|)$, must be defined on the corresponding support. Following \citet{dimatteo2001bayesian}, $\tilde q_j$ is chosen based on a uniform prior on the $|\xi_j|$-simplex by scaling $(\xi_j^L,\xi_j^U)$ to $(0,1)$.
With density $q_j:\mathbb \{0,1,\dots,M_j\}\to (0,\infty)$, the prior on $\xi_j$ is formally expressed as:
\begin{align}
	\pi_{FK}(\xi)\propto \prod_{j=1}^p q_j(|\xi_j|) \tilde q_{j}(\xi_j\mid |\xi_j|), \quad \xi \in \Xi_{FK}.
	\label{eqn:fkprior}
\end{align}
While the original approach in \citet{dimatteo2001bayesian} requires that at least one knot be always included, our free-knot spline prior in \eqref{eqn:fkprior} extends it by allowing the possibility of an empty knot, which can account for a completely linear effect. 
To explore the posterior distribution, reversible jump MCMC with birth, death, and relocation proposals can be used \citep{dimatteo2001bayesian}. This approach is generally more computationally demanding than methods for VS-knot splines. 
Despite the increased flexibility of the free-knot spline prior compared to the one in \eqref{eqn:prior2}, our experience indicates that this flexibility does not significantly improve performance in most practical cases. 
The inherent inefficiency of reversible-jump MCMC further underscores the importance of avoiding unnecessary use of free-knot splines.
Our simulation study in Section~\ref{sec:sims} shows that while performance measures for free-knot splines are comparable to those for VS-knot splines, the sampling efficiency (measured as the ratio of effective sample size to runtime) is notably lower for free-knot splines.

Similar to the VS-knot spline approach, the basis construction in \eqref{eqn:ncs} and \eqref{eqn:basis} is useful for free-knot splines. According to Proposition~\ref{prop:selection}, adding or removing a knot-point corresponds to adding or removing the corresponding basis term. Consequently, reversible-jump MCMC can be implemented by modifying the matrix columns without needing to reconstruct the entire basis term.

\subsection{Prior distribution on $|\xi_j|$}
\label{sec:q-density}
The priors described in Sections~\ref{sec:prior1}--\ref{sec:prior3} require specifying the  density $q_j$ for $|\xi_j|$.
Previous studies have shown that to achieve optimal properties in nonparametric regression, priors used in BMS-based methods must have appropriately decaying tail properties \citep[e.g.,][]{shen2015adaptive}. Priors with guaranteed tail properties include Poisson and geometric distributions, with suitable truncation as needed. 
To leverage both the theoretical benefits and practical performance, we select our default prior as a mixture of a point mass at $|\xi_j|=0$ and a truncated geometric distribution for $|\xi_j| > 0$.
Specifically, the density is given by
\begin{align}
	q_j(u) = 
	\begin{cases}
		\lambda_j, &\quad u = 0,\\
		(1-\lambda_j)(1-\varpi_j)^u /\sum_{\ell=1}^{M_j} (1-\varpi_j)^\ell ,&\quad u=1,\dots,M_j,
	\end{cases}
	\label{eqn:truncgeom}
\end{align}
where the hyperparameter $\lambda_j\in [0,1]$ represents the prior belief regarding a linear effect, while $\varpi_j\in [0,1]$ governs the tail behavior. A reasonable default choice for $\lambda_j$ is $1/2$.
Selecting a value close to zero for $\varpi_j$ makes the second part of the prior in \eqref{eqn:truncgeom} closely resemble a discrete uniform distribution while still ensuring the desired tail property for optimality. 
However, we find that using a moderately small value for $\varpi_j$ improves the stability in estimating knot specifications. 
Consequently, we set $\varpi_j=0.2$ as the default value.

The density $q_j$ in \eqref{eqn:truncgeom} is also useful in GAPLMs, where some predictor variables are expected to have linear effects (e.g., binary variables). This is achieved by fixing specific $f_j$ to include only the linear basis term $N_1$. Accordingly, for predictor variables with linear effects, we set $\lambda_j=1$ for the linear additive components and $\lambda_j=1/2$ for the nonparametric additive components.

\section{Numerical study}
\label{sec:sims}
The primary goal of this study is to investigate the behavior of mixtures of g-priors in BMS-based approaches for estimating GAMs. While Section~\ref{sec:penalty} provides some foundational insights, the optimal mixture prior for GAMs remains unclear. Additionally, we evaluate three strategies for specifying priors for knots, as discussed in Section~\ref{sec:knotprior}, and compare them with other function estimation methods. This section introduces the simulation study designed to address these objectives.

\subsection{Comparison among the mixtures of g-priors}\label{sec:sim1}

We first conduct a simulation study to examine the differences in performance between mixtures of g-priors for estimating GAMs. For the synthetic functions, we consider the following four uncentered functions $f_j^\ast:[-1,1]\to \mathbb R$, $j=1,2,3,4$:
\begin{align}
	\begin{split}
	{f}_1^\ast(x) &= 0.5(2 x^5 + 3 x^2 + \cos (3\pi x) - 1), \\
	{f}_2^\ast(x) &= \frac{21(3x+1.5)^3}{8000}   + \frac{21(3x-2.5)^2 }{400e^{-3x - 1.5}} \sin\! \left(\frac{17\pi(3x+1.5)^2 }{32}\right) \mathbbm 1_{(-0.5,0.85)}(x),\\
	{f}_3^\ast(x) &= x,\\
	{f}_4^\ast(x) &= 0,
	\label{eqn:functions}
	\end{split}
\end{align}
where $\mathbbm 1_A$ is the indicator function of a set $A$.
Specifically, ${f}_1^\ast$ is a nonlinear function that is not a polynomial, ${f}_2^\ast$ is a nonlinear function with locally varying smoothness, ${f}_3^\ast$ is a linear function, and $f_4^\ast$ is a constant function.
The two nonlinear functions ${f}_1^\ast$ and ${f}_2^\ast$ are adapted from \citet{gressani2021laplace} and \citet{francom2020bass}, respectively. 
These functions are illustrated in Figure~\ref{plot:ber1} with appropriate centering.

\begin{figure}[t!]
	\centering
\subfloat[\centering Pointwise posterior mean estimates of $f_1$ in $100$ replications]{\includegraphics[width=1\textwidth]{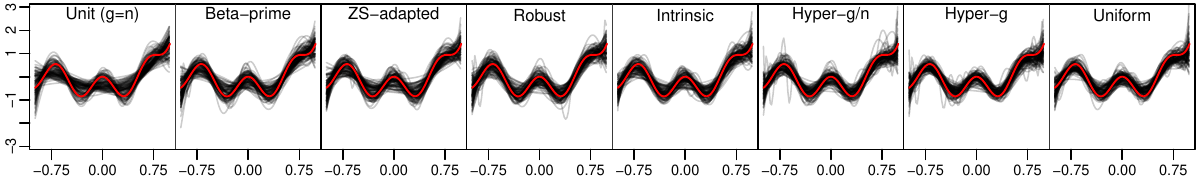}}\\
\subfloat[\centering Pointwise posterior mean estimates of $f_2$ in $100$ replications]{\includegraphics[width=1\textwidth]{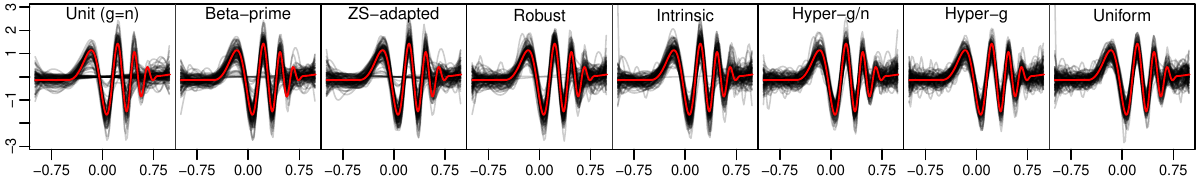}}\\
\subfloat[\centering Pointwise posterior mean estimates of $f_3$ in $100$ replications]{\includegraphics[width=1\textwidth]{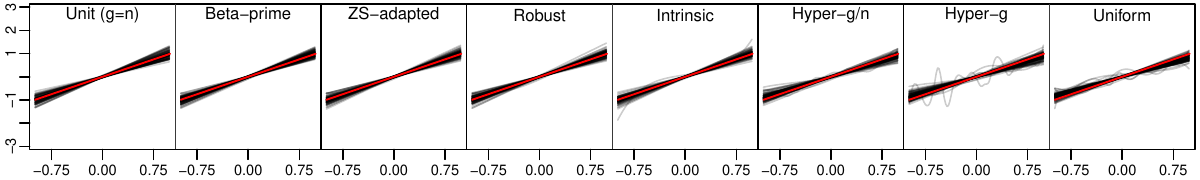}}\\
\subfloat[\centering Pointwise posterior mean estimates of $f_4$ in $100$ replications]{\includegraphics[width=1\textwidth]{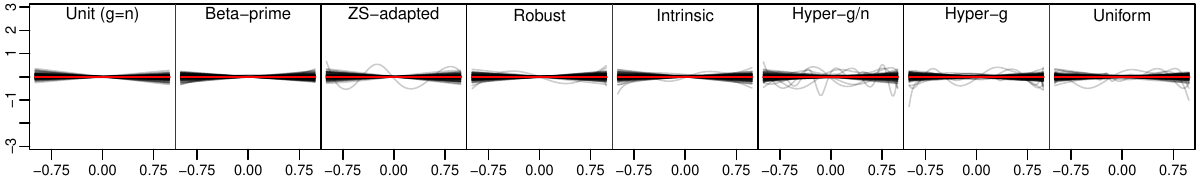}}
\caption{Pointwise posterior means (gray) of $f_1$, $f_2$, $f_3$ and $f_4$ in the nonparametric logistic regression model with $n=1000$, obtained from randomly chosen 100 replicated datasets, along with the true function (red).}
	\label{plot:ber1}
\end{figure}

The simulation datasets are generated using the exponential family model with the additive predictor $\eta_i=\sum_{j=1}^4 f_j^\ast(x_{ij})=\alpha+\sum_{j=1}^4 f_j(x_{ij})$, where $x_{ij}$ are drawn independently from $\text{Unif}(-1,1)$, $f_j$ represents the centered version of $f_j^\ast$, and $\alpha$ is the induced intercept. 
In this section, we present the simulation results for the nonlinear logistic regression model, where $Y_i\sim\text{Bernoulli}(e^{\eta_i}/(1+e^{\eta_i}))$. Section~S6 of the supplementary material includes a simulation study for Poisson regression $Y_i\sim \text{Poi}(e^{\eta_i})$ and Gaussian regression $Y_i\sim N(\eta_i,\sigma^2)$.

As noted in Section~\ref{sec:sim_others}, the VS-knot spline approach performs reasonably well compared to the other strategies for choosing $\Xi$ described in Section~\ref{sec:knotprior}. Therefore, we focus specifically on VS-knot splines in this section.
For each value of $n=500,1000,2000$, we generate 500 data replications and estimate $f_j$ using VS-knot splines with $M_j=30$ knot candidates. We employ the unit information prior and the mixture priors summarized in Table~\ref{table:mixg}, with the prior in \eqref{eqn:truncgeom} with $\varpi=0.2$ and $\lambda = 1/2$.
For each prior distribution, we run a Markov chain of length 10,000 to explore the posterior distribution, ensuring convergence after an appropriate burn-in period.
We then calculate the root mean squared error (RMSE) and the coverage probabilities of 95\% pointwise credible bands.

\begin{figure}[t!]
\centering
\subfloat[\centering Logarithm of RMSE for $f_1$]{\includegraphics[width = 0.47\textwidth]{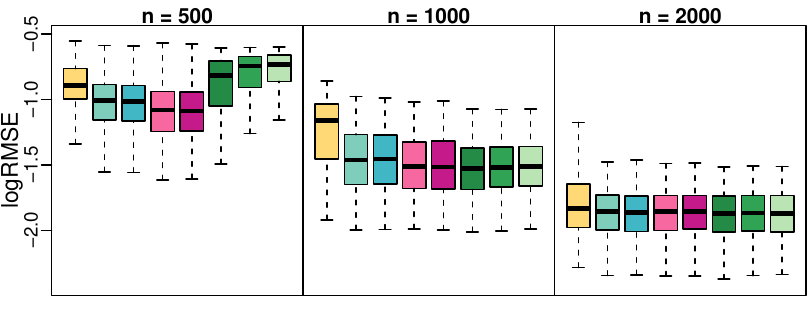}}
~
\subfloat[\centering Coverage probabilities for $f_1$]{\includegraphics[width = 0.47\textwidth]{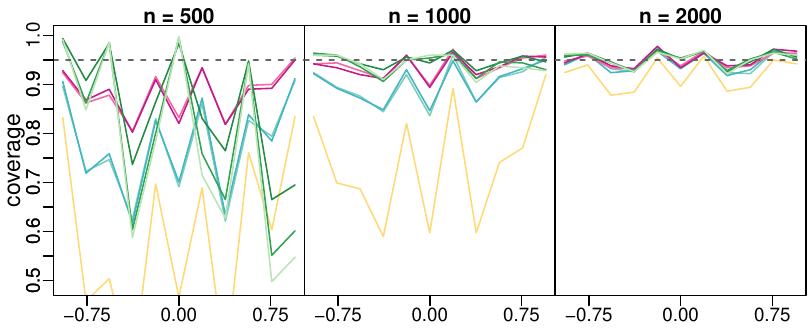}}
\\
\subfloat[\centering Logarithm of RMSE for $f_2$]{\includegraphics[width = 0.47\textwidth]{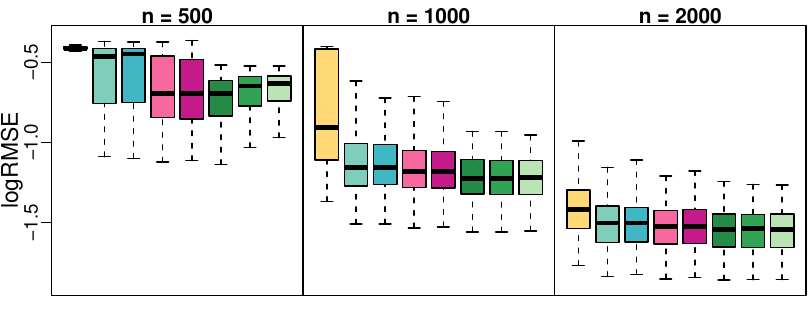}}
~
\subfloat[\centering Coverage probabilities for $f_2$]{\includegraphics[width = 0.47\textwidth]{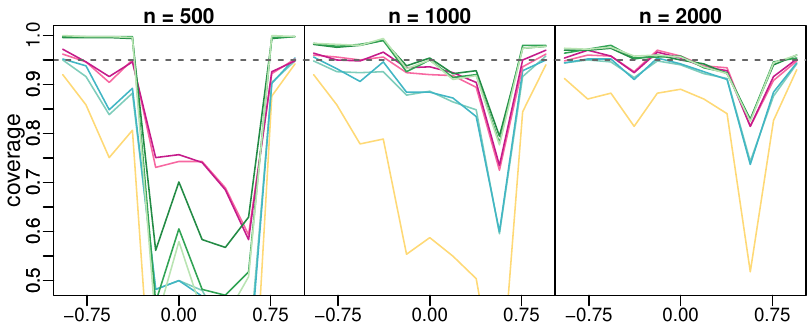}}
\\
\subfloat[\centering Logarithm of RMSE for $f_3$]{\includegraphics[width = 0.47\textwidth]{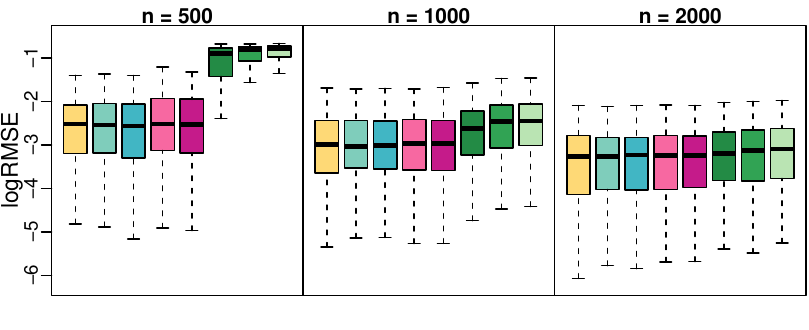}}
~
\subfloat[\centering Coverage probabilities for $f_3$]{\includegraphics[width = 0.47\textwidth]{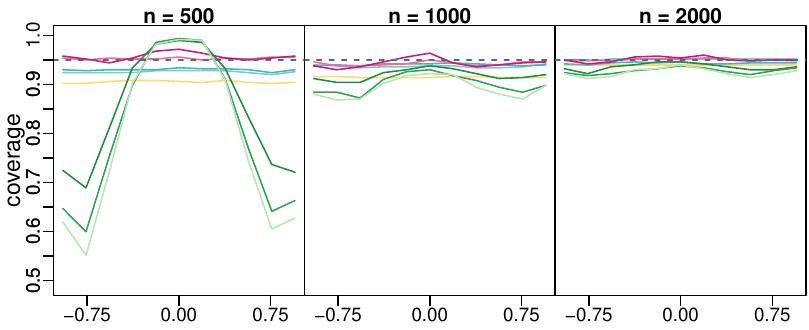}}
\\
\subfloat[\centering Logarithm of RMSE for $f_4$]{\includegraphics[width = 0.47\textwidth]{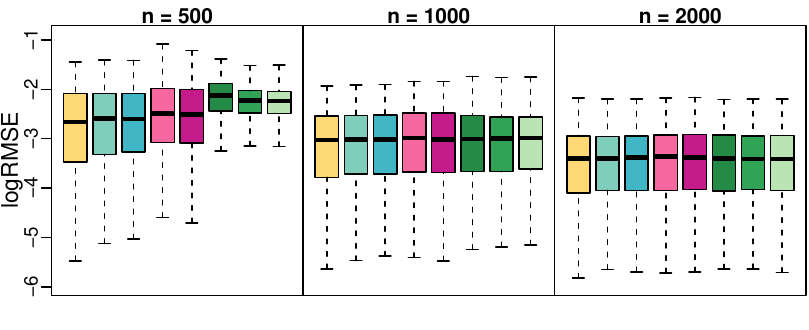}}
~
\subfloat[\centering Coverage probabilities for $f_4$]{\includegraphics[width = 6.5cm]{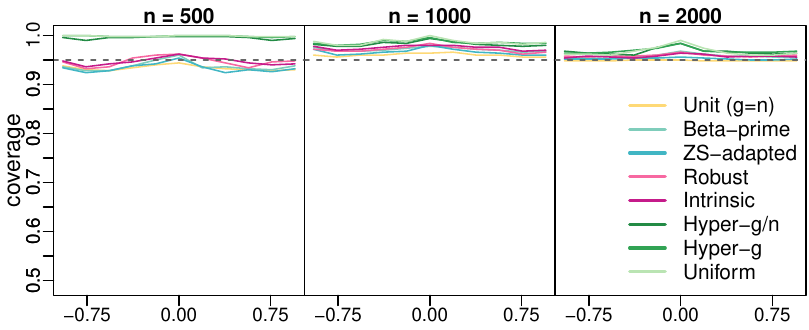}}
\caption{Logarithm of RMSE and coverage probabilities for $f_1$, $f_2$, $f_3$ and $f_4$ in the nonparametric logistic regression models with $n=500,1000, 2000$, obtained from 500 replicated datasets. Outliers are excluded to improve visualization.}
\label{plot:ber2}
\end{figure}

Figures~\ref{plot:ber1} and \ref{plot:ber2} display the simulation results. 
As discussed in Section~\ref{sec:penalty}, the unit information prior behaves quite differently from the mixture priors. It generally underperforms in nonlinear function estimation and exhibits clear signs of underfitting. This indicates that using the unit information prior for function estimation may be unsuitable.
The main challenge is determining the most appropriate mixture prior for function estimation. 
Although differences between mixtures of g-priors become less pronounced with larger sample sizes, intrinsic and robust priors consistently outperform other priors in finite samples. The beta-prime and ZS-adapted priors tend to exhibit slight underfitting, while the uniform, hyper-g, and hyper-g/n priors tend to exhibit overfitting. This aligns with the expectations discussed in Section~\ref{sec:penalty}.
Across all functions, the robust and intrinsic priors achieve moderate RMSE and coverage properties, with the intrinsic prior being slightly more accurate for smaller samples.
The simulation results for Poisson and Gaussian regressions in the supplementary material support this conclusion. We recommend using the intrinsic or robust prior as the default choice for $g$.

\subsection{Comparison with other methods}\label{sec:sim_others}

We now compare BMS-based methods for GAMs with other approaches to GAM estimation. We consider the three strategies described in Section~\ref{sec:knotprior}: even-knot splines, VS-knot splines, and free-knot splines, alongside several competitors available in R packages: \texttt{R2BayesX} \citep{umlauf2012structured}, \texttt{Blapsr} \citep{gressani2021laplace}, \texttt{mgcv} \citep{wood2017generalized}, and \texttt{bsamGP} \citep{jo2019bsamgp}.
Based on the results in Section~\ref{sec:sim1}, the three BMS-based approaches use the intrinsic prior. 
Among the competitors, \texttt{mgcv} is the only frequentist method, while the others are Bayesian. Specifically,  \texttt{R2BayesX} and \texttt{Blapsr} are based on Bayesian P-splines \citep{lang2004bayesian}.
\texttt{R2BayesX} offers conventional MCMC estimates, whereas \texttt{Blapsr} provides an option for the Laplace approximation, which can improve computational efficiency when the number of additive components is small. In contrast, \texttt{bsamGP} uses a second-order Gaussian process to estimate nonparametric functions in GAMs.

To ensure a fair comparison, we carefully select simulation specifications. For the BMS-based methods (i.e., even-knot, VS-knot, and free-knot splines), the maximum number of knots $M_j$ is consistently set to $30$ for each $j=1,2,3,4$.  Similarly, for the competitors relying on penalized splines (i.e., \texttt{R2BayesX}, \texttt{Blapsr}, and \texttt{mgcv}), we use $M_j=30$, ensuring comparable least-penalized models across both BMS-based and penalized spline approaches. The \texttt{mgcv} package offers an option for locally adaptive smooth functions. We explore both the standard version with a single smoothness parameter (\texttt{mgcv-ps}) and a variant with local adaptation (\texttt{mgcv-ad}). For \texttt{bsamGP}, the number of cosine basis functions in the spectral representation of the Gaussian process priors is set equal to $M_j$ for each $f_j$.
The simulation settings follow those in Section~\ref{sec:sim1}, using the functions specified in \eqref{eqn:functions}. This section presents the simulation results for nonparametric logistic regression, with results for Poisson and Gaussian regression available in Section S6 of the supplementary material.
For each Bayesian method relying on MCMC, we generate a Markov chain of length 10,000 to explore the posterior distribution, ensuring convergence after a suitable burn-in period.
We then calculate the RMSE and 95\% pointwise credible bands for selected points for each method.

\begin{figure}[t!]
	\centering
	\subfloat[\centering Pointwise posterior mean estimates of $f_1$ in $100$ replications]{\includegraphics[width=1\textwidth]{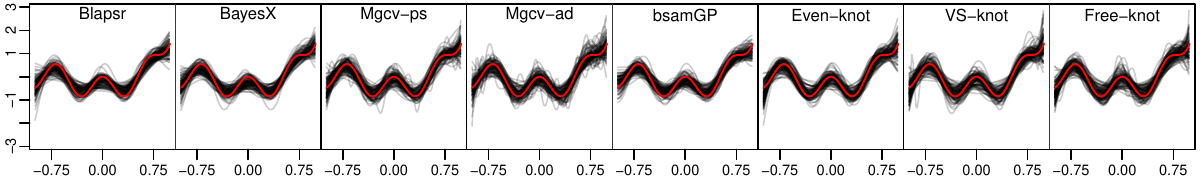}}\\
	\subfloat[\centering Pointwise posterior mean estimates of $f_2$ in $100$ replications]{\includegraphics[width=1\textwidth]{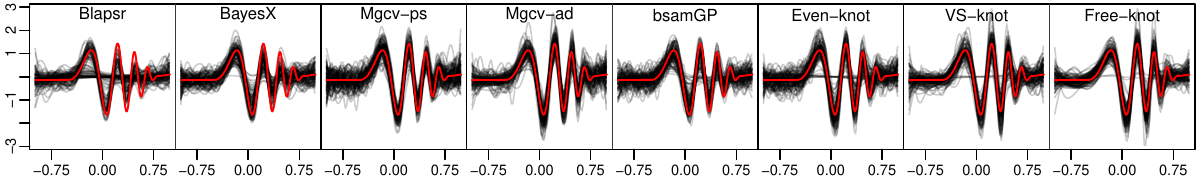}}\\
	\subfloat[\centering Pointwise posterior mean estimates of $f_3$ in $100$ replications]{\includegraphics[width=1\textwidth]{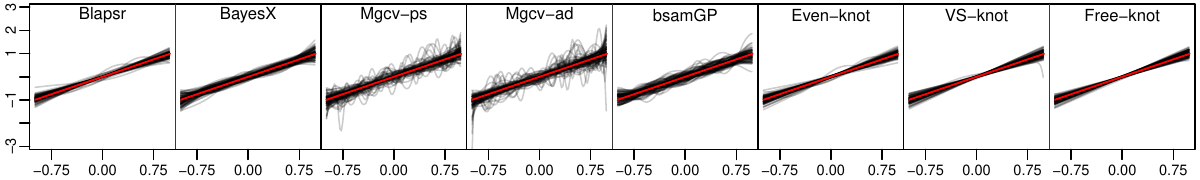}}\\
	\subfloat[\centering Pointwise posterior mean estimates of $f_4$ in $100$ replications]{\includegraphics[width=1\textwidth]{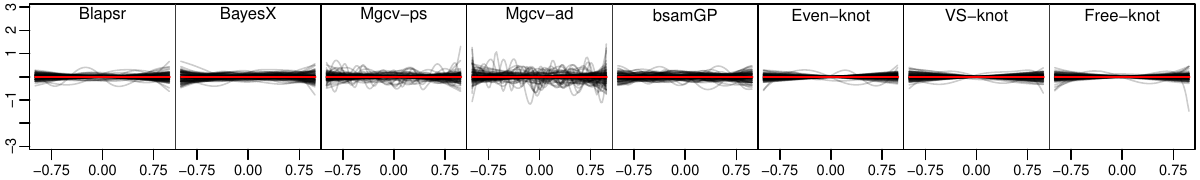}}
	\caption{Pointwise posterior means (gray) of $f_1$, $f_2$, $f_3$ and $f_4$ in the nonparametric logistic regression model with $n=1000$, obtained from randomly chosen 100 replicated datasets, along with the true function (red).}
	\label{plot:ber3}
\end{figure}

\begin{figure}[t!]
	\centering
	\subfloat[\centering Logarithm of RMSE for $f_1$]{\includegraphics[width = 0.47\textwidth]{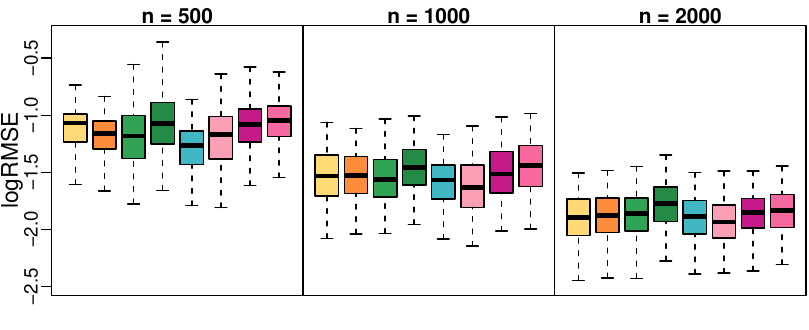}}
~
	\subfloat[\centering Coverage probabilities for $f_1$]{\includegraphics[width = 0.47\textwidth]{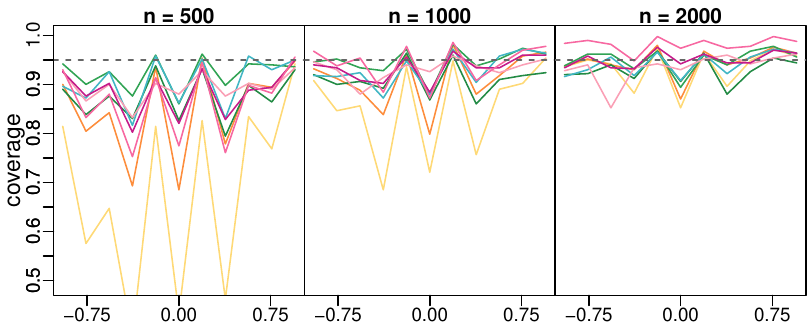}}
	\\
	\subfloat[\centering Logarithm of RMSE for $f_2$]{\includegraphics[width = 0.47\textwidth]{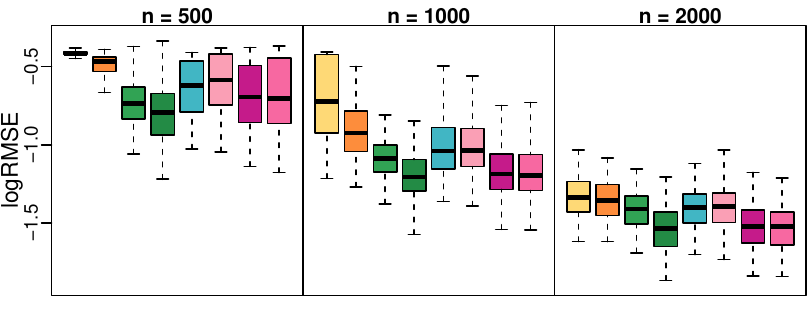}}
~
	\subfloat[\centering Coverage probabilities for $f_2$]{\includegraphics[width = 0.47\textwidth]{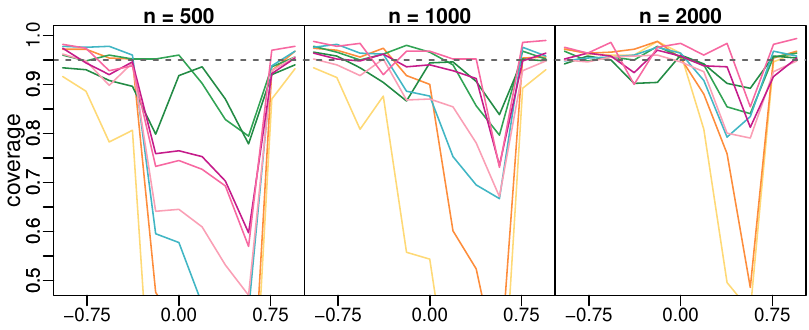}}
	\\
	\subfloat[\centering Logarithm of RMSE for $f_3$]{\includegraphics[width = 0.47\textwidth]{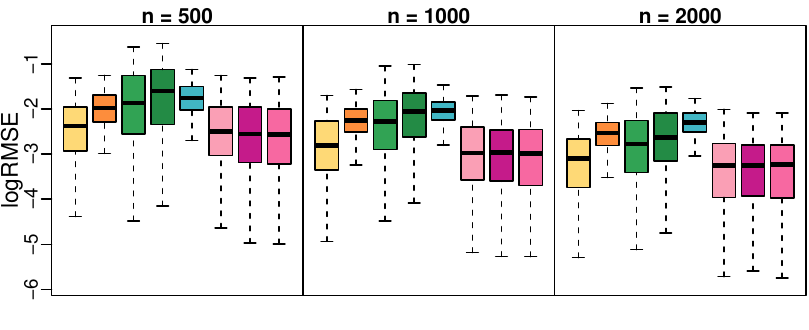}}
~
	\subfloat[\centering Coverage probabilities for $f_3$]{\includegraphics[width = 0.47\textwidth]{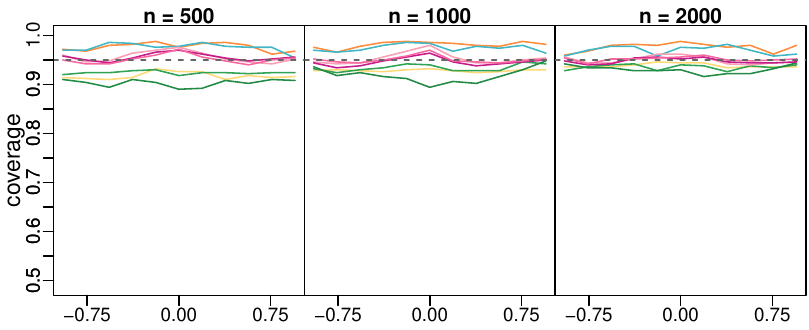}}
	\\
	\subfloat[\centering Logarithm of RMSE for $f_4$]{\includegraphics[width = 0.47\textwidth]{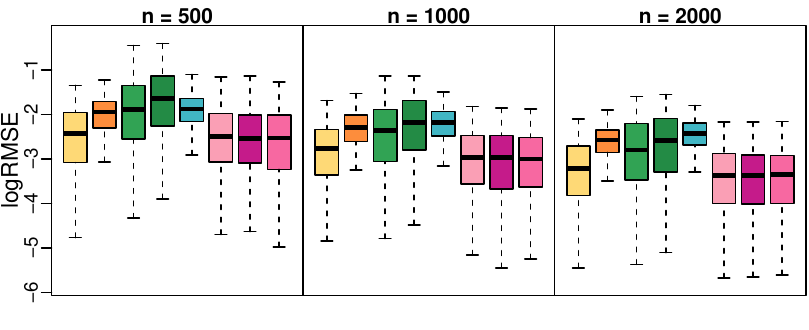}}
~
	\subfloat[\centering Coverage probabilities for $f_4$]{\includegraphics[width = 0.47\textwidth]{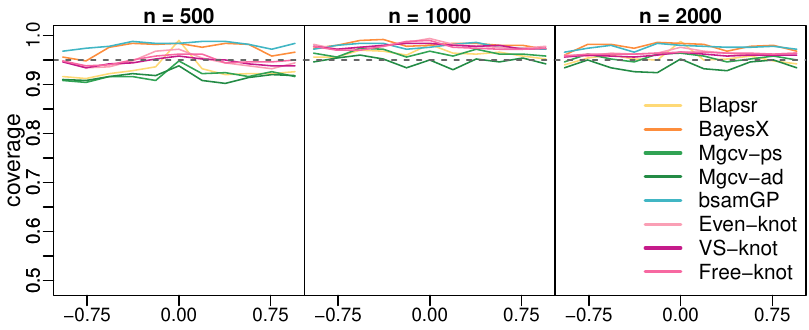}}
	\caption{Logarithm of RMSE and coverage probabilities for $f_1$, $f_2$, $f_3$ and $f_4$ in the nonparametric logistic regression models with $n=500,1000, 2000$, obtained from 500 replicated datasets. Outliers are excluded to improve visualization.}
	\label{plot:ber4}
\end{figure}

Figures~\ref{plot:ber3} and \ref{plot:ber4} summarize the simulation results for the nonlinear logistic regression models. Observations reveal that \texttt{R2BayesX} and \texttt{Blapsr} tend to oversmooth the target functions owing to excessive penalization. In contrast, \texttt{mgcv} produces highly oscillatory estimates for the linear and constant functions, reflecting a tendency to overfit simpler functions.
Both \texttt{R2BayesX} and \texttt{Blapsr} struggle with locally varying smoothness, as penalized splines are not inherently designed for such adaptability without significant modifications \citep{crainiceanu2007spatially, jullion2007robust, scheipl2009locally}.
While \texttt{mgcv} with local adaptation performs well in estimating the locally varying smoothness of $f_2$, the performance for $f_1$, $f_3$, and $f_4$ suggests that adaptive estimation using \texttt{mgcv} may lead to higher RMSEs and incorrect coverage probabilities. 
A major drawback of \texttt{mgcv} is the challenge of accurately specifying whether adaptive estimation will achieve optimal performance, given the unknown characteristics of the target function.

Among the BMS-based approaches, even-knot splines exhibit limitations in adapting to the locally varying smoothness of $f_2$ owing to their construction with equidistant knots. In contrast, both the VS-knot and free-knot splines effectively identify the local features of $f_2$. The results show that the RMSEs for these two adaptive estimation methods are comparable, although free-knot splines tend to slightly overestimate the coverage probabilities.
Similar to the comparison between \texttt{mgcv-ps} and \texttt{mgcv-ad}, even-knot splines perform better than the VS-knot and free-knot splines in estimating $f_1$. However, the BMS-based methods are comparable in estimating $f_3$ and $f_4$, indicating that BMS-based methods are generally less sensitive to whether adaptive estimation is employed. Given that VS-knot splines typically outperform other methods and effectively handle local adaptation, we recommend using them as the default option. Nonetheless, even-knot splines are faster than other BMS-based methods and eliminate the need for MCMC when $p$ is relatively small.

We also evaluated the computational efficiency of the Bayesian methods based on MCMC, excluding \texttt{mgcv} (which follows a frequentist approach) and \texttt{Blapsr} (which does not use MCMC). We measured the effective sample size of the posterior per second of runtime, after accounting for appropriate burn-in periods. 
Figure~\ref{plot:berESS} displays the efficiency measures across 500 replicates. Contrary to the common belief that BMS-based methods might be inefficient, they perform comparably to other Bayesian methods. In particular, even-knot splines are the most efficient, owing to their fast mixing despite their slower overall processing. Although VS-knot splines are slower than even-knot splines, the trade-off is justified by their capability for local adaptation.

\begin{figure}[t!]
	\centering
	\includegraphics[width = 9.5cm]{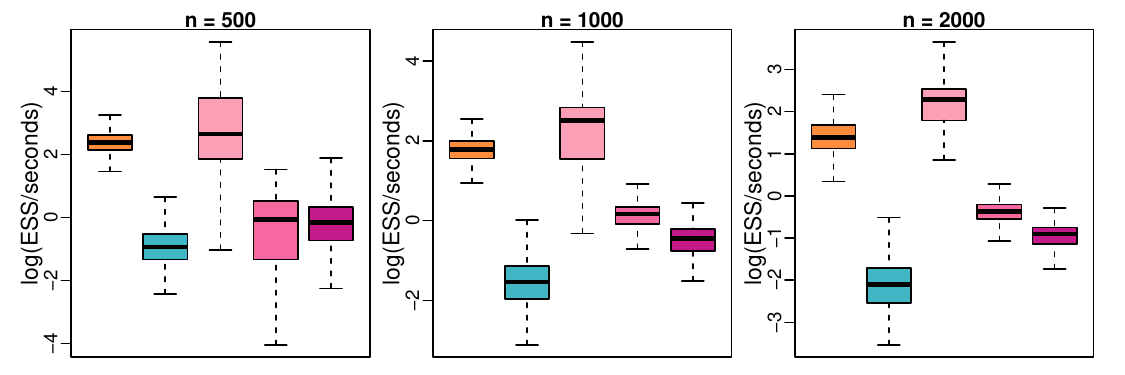}
	\includegraphics[width=9.5cm]{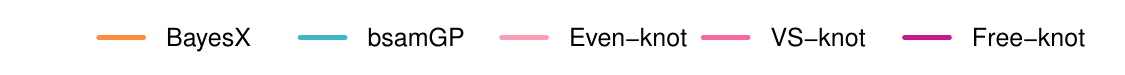}
	\caption{Logarithm of the effective sample sizes of the joint posterior per second of runtime, in the nonparametric logistic regression models with $n=500, 1000, 2000$, obtained from 500 replicated datasets.}
	\label{plot:berESS}
\end{figure}

\section{Application to Pima diabetes data}
\label{sec:realdata}

In this section, we analyze the Pima dataset using the VS-knot spline approach with the intrinsic prior.
The Pima diabetes dataset consists of signs of diabetes and seven potential risk factors for $n=532$ Pima Indian women in Arizona \citep{smith1988using}.
We explore the relationship between the signs of diabetes and these risk factors using a GAM. The response variable $Y_i$ indicates the presence of diabetes (0: negative, 1: positive). For each individual $i$, the predictor variables (risk factors) are $pregnant_i$ (number of times the subject was pregnant),	$glucose_i$ (plasma glucose concentration in two hours in an oral glucose tolerance test, mg/dl), $pressure_i$ (diastolic blood pressure, mm/Hg), $triceps_i$ (triceps skin fold thickness, mm/Hg), $mass_i$ (body mass index, BMI), $pedigree_i$ (diabetes pedigree function), and $age_i$ (age).

To examine the relationship between $Y_i$ and the risk factors, we consider the following GAM with a logit link, 
\begin{align}
	\begin{split}\label{eqn:Pima}
		\log\frac{E(Y_i)}{1-E(Y_i)} &= \alpha + f_1(pregnant_i) + f_2(glucose_i) + f_3 (pressure_i) \\
		&\quad +f_4(triceps_i)+f_5 (mass_i)+f_6 (pedigree_i) +f_7 (age_i). 
	\end{split}
\end{align}
The individuals with missing values are removed from the analysis.
	For VS-knot splines, a set of knot candidates $\xi_j^c=\{\xi_{j1}^c,\dots, \xi_{jM_j}^c\}$ is determined from the unique values of the quantiles for each predictor variable.
	 Some predictors are discrete in the observed data (e.g., $pregnant_i$ and $age_i$). For each $j$, we set $M_j$ as the number of unique values among the 30 quantile values with equal weights, meaning that $M_j<30$ for some $j$.

\begin{figure}[t!]
	\centering
	\includegraphics[width=11cm]{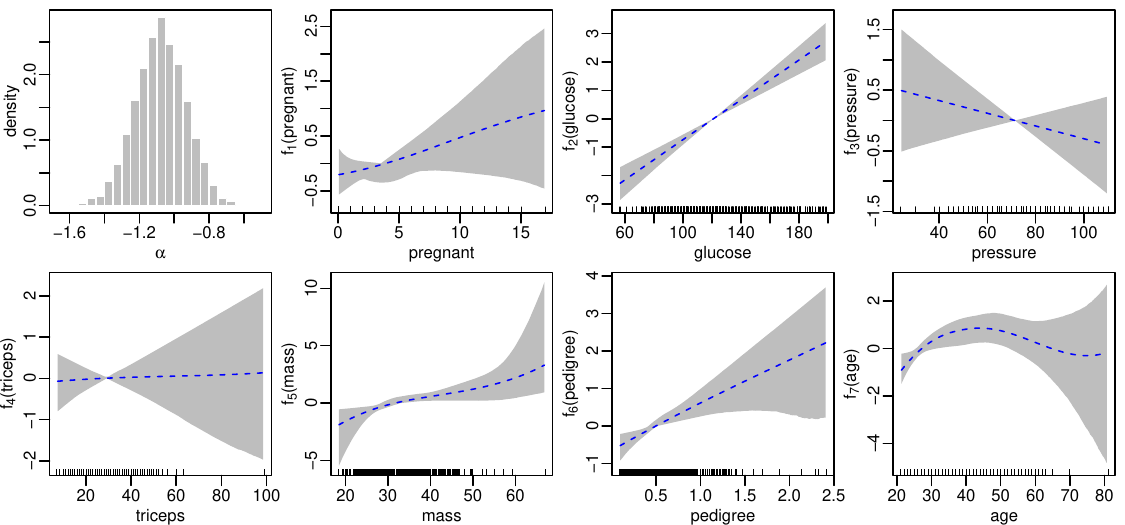}
	\caption{The posterior distribution of $\alpha$ and the pointwise posterior means (blue dashed curve) and pointwise $95\%$ credible bands (gray shade) of the functions $f_j$, $j=1,\dots,7$, for the model in \eqref{eqn:Pima}.}
	\label{plot:PimaSmooth}
\end{figure}

The results summarized in Figure~\ref{plot:PimaSmooth} largely align with intuition. Many variables show near-linear effects, while a few, such as $mass_i$ and $age_i$, exhibit clear nonlinear effects. The effect of $triceps_i$ appears to be negligible, suggesting that it may be worth considering the exclusion of this variable from the analysis.

\section{Discussion}
\label{sec:discussion}

This study examined BMS-based estimation methods for GAMs using the Laplace approximation with mixtures of g-priors. We establish a default prior by analyzing the behavior of the Bayes factor and presenting numerical results. Additionally, this study consolidates existing ideas on priors for knots and mixtures of g-priors.

A significant limitation of the BMS-based approach is that the Laplace approximation depends on the maximum likelihood estimator, which requires extensive computations. Given that the VS-knot spline approach proves effective, one potential solution is to use shrinkage priors for exponential family models in combination with a reasonable MCMC sampling algorithm \citep[e.g.,][]{schmidt2020bayesian}. Another avenue is to explore computationally less expensive likelihood approximations \citep[e.g.,][]{rossell2021approximate}. Additionally, investigating higher-order approximations could offer improved approximation performance \citep{shun1995laplace}.

\bibliographystyle{ba}
\bibliography{ref}

\begin{thebibliography}{68}
\newcommand{\enquote}[1]{``#1''}
\expandafter\ifx\csname natexlab\endcsname\relax\def\natexlab#1{#1}\fi
\expandafter\ifx\csname url\endcsname\relax
  \def\url#1{{\tt #1}}\fi
\expandafter\ifx\csname urlprefix\endcsname\relax\def\urlprefix{URL }\fi
\ifx\endbibitem\undefined \let\endbibitem\relax\fi

\bibitem[{Armero and Bayarri(1994)}]{armero1994prior}
Armero, C. and Bayarri, M. (1994).
\newblock \enquote{Prior assessments for prediction in queues.}
\newblock {\em Journal of the Royal Statistical Society: Series D (The
  Statistician)\/}, 43(1): 139--153.
\endbibitem

\bibitem[{Bayarri et~al.(2012)Bayarri, Berger, Forte, and
  Garc{\'\i}a-Donato}]{bayarri2012criteria}
Bayarri, M.~J., Berger, J.~O., Forte, A., and Garc{\'\i}a-Donato, G. (2012).
\newblock \enquote{Criteria for {B}ayesian model choice with application to
  variable selection.}
\newblock {\em The Annals of Statistics\/}, 40(3): 1550--1577.
\endbibitem

\bibitem[{Berger et~al.(1998)Berger, Pericchi, and
  Varshavsky}]{berger1998bayes}
Berger, J.~O., Pericchi, L.~R., and Varshavsky, J.~A. (1998).
\newblock \enquote{Bayes factors and marginal distributions in invariant
  situations.}
\newblock {\em Sankhy{\=a}: The Indian Journal of Statistics, Series A\/},
  307--321.
\endbibitem

\bibitem[{Brezger and Lang(2006)}]{brezger2006generalized}
Brezger, A. and Lang, S. (2006).
\newblock \enquote{Generalized structured additive regression based on
  {B}ayesian {P}-splines.}
\newblock {\em Computational Statistics \& Data Analysis\/}, 50(4): 967--991.
\endbibitem

\bibitem[{Castellanos et~al.(2021)Castellanos, Garc{\'i}a-Donato, and
  Cabras}]{castellanos2021model}
Castellanos, M.~E., Garc{\'i}a-Donato, G., and Cabras, S. (2021).
\newblock \enquote{{A model selection approach for variable selection with
  censored data}.}
\newblock {\em Bayesian Analysis\/}, 16(1): 271 -- 300.
\endbibitem

\bibitem[{Chan et~al.(2006)Chan, Kohn, Nott, and Kirby}]{chan2006locally}
Chan, D., Kohn, R., Nott, D., and Kirby, C. (2006).
\newblock \enquote{Locally adaptive semiparametric estimation of the mean and
  variance functions in regression models.}
\newblock {\em Journal of Computational and Graphical Statistics\/}, 15(4):
  915--936.
\endbibitem

\bibitem[{Chen and Ibrahim(2003)}]{chen2003conjugate}
Chen, M.-H. and Ibrahim, J.~G. (2003).
\newblock \enquote{Conjugate priors for generalized linear models.}
\newblock {\em Statistica Sinica\/}, 461--476.
\endbibitem

\bibitem[{Chipman et~al.(2010)Chipman, George, and McCulloch}]{chipman2010bart}
Chipman, H.~A., George, E.~I., and McCulloch, R.~E. (2010).
\newblock \enquote{{BART}: {B}ayesian additive regression trees.}
\newblock {\em The Annals of Applied Statistics\/}, 4(1): 266--298.
\endbibitem

\bibitem[{Cox and Snell(1989)}]{cox1989analysis}
Cox, D. and Snell, E. (1989).
\newblock {\em The Analysis of Binary Data\/}, volume~32.
\newblock CRC Press.
\endbibitem

\bibitem[{Crainiceanu et~al.(2007)Crainiceanu, Ruppert, Carroll, Joshi, and
  Goodner}]{crainiceanu2007spatially}
Crainiceanu, C.~M., Ruppert, D., Carroll, R.~J., Joshi, A., and Goodner, B.
  (2007).
\newblock \enquote{Spatially adaptive {B}ayesian penalized splines with
  heteroscedastic errors.}
\newblock {\em Journal of Computational and Graphical Statistics\/}, 16(2):
  265--288.
\endbibitem

\bibitem[{De~Jonge and Van~Zanten(2012)}]{de2012adaptive}
De~Jonge, R. and Van~Zanten, J. (2012).
\newblock \enquote{Adaptive estimation of multivariate functions using
  conditionally {G}aussian tensor-product spline priors.}
\newblock {\em Electronic Journal of Statistics\/}, 6: 1984--2001.
\endbibitem

\bibitem[{Dellaportas et~al.(2002)Dellaportas, Forster, and
  Ntzoufras}]{dellaportas2002bayesian}
Dellaportas, P., Forster, J.~J., and Ntzoufras, I. (2002).
\newblock \enquote{On {B}ayesian model and variable selection using {MCMC}.}
\newblock {\em Statistics and Computing\/}, 12(1): 27--36.
\endbibitem

\bibitem[{Denison et~al.(1998{\natexlab{a}})Denison, Mallick, and
  Smith}]{denison1998automatic}
Denison, D., Mallick, B., and Smith, A. (1998{\natexlab{a}}).
\newblock \enquote{Automatic {B}ayesian curve fitting.}
\newblock {\em Journal of the Royal Statistical Society: Series B (Statistical
  Methodology)\/}, 60(2): 333--350.
\endbibitem

\bibitem[{Denison et~al.(1998{\natexlab{b}})Denison, Mallick, and
  Smith}]{denison1998bayesian}
Denison, D.~G., Mallick, B.~K., and Smith, A.~F. (1998{\natexlab{b}}).
\newblock \enquote{Bayesian {MARS}.}
\newblock {\em Statistics and Computing\/}, 8(4): 337--346.
\endbibitem

\bibitem[{DiMatteo et~al.(2001)DiMatteo, Genovese, and
  Kass}]{dimatteo2001bayesian}
DiMatteo, I., Genovese, C.~R., and Kass, R.~E. (2001).
\newblock \enquote{{B}ayesian curve-fitting with free-knot splines.}
\newblock {\em Biometrika\/}, 88(4): 1055--1071.
\endbibitem

\bibitem[{Fahrmeir and Lang(2001)}]{fahrmeir2001bayesian}
Fahrmeir, L. and Lang, S. (2001).
\newblock \enquote{Bayesian inference for generalized additive mixed models
  based on {M}arkov random field priors.}
\newblock {\em Journal of the Royal Statistical Society: Series C (Applied
  Statistics)\/}, 50(2): 201--220.
\endbibitem

\bibitem[{Fouskakis et~al.(2018)Fouskakis, Ntzoufras, and
  Perrakis}]{fouskakis2018power}
Fouskakis, D., Ntzoufras, I., and Perrakis, K. (2018).
\newblock \enquote{Power-expected-posterior priors for generalized linear
  models.}
\newblock {\em Bayesian Analysis\/}, 13(3): 721--748.
\endbibitem

\bibitem[{Francom and Sans{\'o}(2020)}]{francom2020bass}
Francom, D. and Sans{\'o}, B. (2020).
\newblock \enquote{{BASS}: An {R} package for fitting and performing
  sensitivity analysis of {B}ayesian adaptive spline surfaces.}
\newblock {\em Journal of Statistical Software\/}, 94(8): 1--36.
\endbibitem

\bibitem[{Francom et~al.(2018)Francom, Sans{\'o}, Kupresanin, and
  Johannesson}]{francom2018sensitivity}
Francom, D., Sans{\'o}, B., Kupresanin, A., and Johannesson, G. (2018).
\newblock \enquote{Sensitivity analysis and emulation for functional data using
  {B}ayesian adaptive splines.}
\newblock {\em Statistica Sinica\/}, 791--816.
\endbibitem

\bibitem[{García-Donato et~al.(2023)García-Donato, Cabras, and
  Castellanos}]{garcia2023model}
García-Donato, G., Cabras, S., and Castellanos, M.~E. (2023).
\newblock \enquote{Model uncertainty quantification in Cox regression.}
\newblock {\em Biometrics\/}, 79(3): 1726--1736.
\endbibitem

\bibitem[{Gordy(1998{\natexlab{a}})}]{gordy1998computationally}
Gordy, M.~B. (1998{\natexlab{a}}).
\newblock \enquote{Computationally convenient distributional assumptions for
  common-value auctions.}
\newblock {\em Computational Economics\/}, 12(1): 61--78.
\endbibitem

\bibitem[{Gordy(1998{\natexlab{b}})}]{gordy1998generalization}
--- (1998{\natexlab{b}}).
\newblock \enquote{A generalization of generalized beta distributions.}
\newblock Division of Research and Statistics, Division of Monetary Affairs,
  Federal Reserve Boards.
\endbibitem

\bibitem[{Green(1995)}]{green1995reversible}
Green, P.~J. (1995).
\newblock \enquote{Reversible jump {M}arkov {C}hain {M}onte {C}arlo computation
  and {B}ayesian model determination.}
\newblock {\em Biometrika\/}, 82(4): 711--732.
\endbibitem

\bibitem[{Gressani and Lambert(2021)}]{gressani2021laplace}
Gressani, O. and Lambert, P. (2021).
\newblock \enquote{Laplace approximations for fast {B}ayesian inference in
  generalized additive models based on {P}-splines.}
\newblock {\em Computational Statistics \& Data Analysis\/}, 154: 107088.
\endbibitem

\bibitem[{Gupta and Nadarajah(2004)}]{gupta2004handbook}
Gupta, A.~K. and Nadarajah, S. (2004).
\newblock {\em Handbook of Beta Distribution and Its Applications\/}.
\newblock CRC press.
\endbibitem

\bibitem[{Gupta and Ibrahim(2009)}]{gupta2009information}
Gupta, M. and Ibrahim, J.~G. (2009).
\newblock \enquote{An information matrix prior for {B}ayesian analysis in
  generalized linear models with high dimensional data.}
\newblock {\em Statistica Sinica\/}, 19(4): 1641--1663.
\endbibitem

\bibitem[{Gustafson(2000)}]{gustafson2000bayesian}
Gustafson, P. (2000).
\newblock \enquote{Bayesian regression modeling with interactions and smooth
  effects.}
\newblock {\em Journal of the American Statistical Association\/}, 95(451):
  795--806.
\endbibitem

\bibitem[{Hansen and Yu(2003)}]{hansen2003minimum}
Hansen, M.~H. and Yu, B. (2003).
\newblock \enquote{Minimum description length model selection criteria for
  generalized linear models.}
\newblock {\em Lecture Notes-Monograph Series\/}, 145--163.
\endbibitem

\bibitem[{Hastie and Tibshirani(1986)}]{hastie1986generalized}
Hastie, T. and Tibshirani, R. (1986).
\newblock \enquote{Generalized additive models.}
\newblock {\em Statistical Sicence\/}, 297--318.
\endbibitem

\bibitem[{Hastie et~al.(2009)Hastie, Tibshirani, Friedman, and
  Friedman}]{hastie2009elements}
Hastie, T., Tibshirani, R., Friedman, J.~H., and Friedman, J.~H. (2009).
\newblock {\em The Elements of Statistical Learning: Data Mining, Inference,
  and Prediction\/}.
\newblock Springer, second edition.
\endbibitem

\bibitem[{Held et~al.(2015)Held, Saban{\'e}s~Bov{\'e}, and
  Gravestock}]{held2015approximate}
Held, L., Saban{\'e}s~Bov{\'e}, D., and Gravestock, I. (2015).
\newblock \enquote{Approximate {B}ayesian model selection with the deviance
  statistic.}
\newblock {\em Statistical Science\/}, 242--257.
\endbibitem

\bibitem[{Jeong et~al.(2017)Jeong, Park, and Park}]{jeong2017analysis}
Jeong, S., Park, M., and Park, T. (2017).
\newblock \enquote{Analysis of binary longitudinal data with time-varying
  effects.}
\newblock {\em Computational Statistics \& Data Analysis\/}, 112: 145--153.
\endbibitem

\bibitem[{Jeong and Park(2016)}]{jeong2016bayesian}
Jeong, S. and Park, T. (2016).
\newblock \enquote{Bayesian semiparametric inference on functional
  relationships in linear mixed models.}
\newblock {\em Bayesian Analysis\/}, 11(4): 1137--1163.
\endbibitem

\bibitem[{Jeong et~al.(2022)Jeong, Park, and van Dyk}]{jeong2021bayesian}
Jeong, S., Park, T., and van Dyk, D.~A. (2022).
\newblock \enquote{Bayesian model selection in additive partial linear models
  via locally adaptive splines.}
\newblock {\em Journal of Computational and Graphical Statistics\/}, 31(2):
  324--336.
\endbibitem

\bibitem[{Jeong and Rockova(2023)}]{jeong2023art}
Jeong, S. and Rockova, V. (2023).
\newblock \enquote{The art of BART: Minimax optimality over nonhomogeneous
  smoothness in high dimension.}
\newblock {\em Journal of Machine Learning Research\/}, 24(337): 1--65.
\endbibitem

\bibitem[{Ji and Schmidler(2013)}]{ji2013adaptive}
Ji, C. and Schmidler, S.~C. (2013).
\newblock \enquote{Adaptive {M}arkov {C}hain {M}onte {C}arlo for {B}ayesian
  variable selection.}
\newblock {\em Journal of Computational and Graphical Statistics\/}, 22(3):
  708--728.
\endbibitem

\bibitem[{Jo et~al.(2019)Jo, Choi, Park, and Lenk}]{jo2019bsamgp}
Jo, S., Choi, T., Park, B., and Lenk, P. (2019).
\newblock \enquote{bsamGP: An {R} Package for {B}ayesian Spectral Analysis
  Models Using {G}aussian Process Priors.}
\newblock {\em Journal of Statistical Software\/}, 90(10): 1–41.
\endbibitem

\bibitem[{Jullion and Lambert(2007)}]{jullion2007robust}
Jullion, A. and Lambert, P. (2007).
\newblock \enquote{Robust specification of the roughness penalty prior
  distribution in spatially adaptive {B}ayesian {P}-splines models.}
\newblock {\em Computational Statistics \& Data Analysis\/}, 51(5): 2542--2558.
\endbibitem

\bibitem[{Kass and Raftery(1995)}]{kass1995bayes}
Kass, R.~E. and Raftery, A.~E. (1995).
\newblock \enquote{{B}ayes factors.}
\newblock {\em Journal of the American Statistical Association\/}, 90(430):
  773--795.
\endbibitem

\bibitem[{Kass and Wasserman(1995)}]{kass1995reference}
Kass, R.~E. and Wasserman, L. (1995).
\newblock \enquote{A reference {B}ayesian test for nested hypotheses and its
  relationship to the {S}chwarz criterion.}
\newblock {\em Journal of the American Statistical Association\/}, 90(431):
  928--934.
\endbibitem

\bibitem[{Kohn et~al.(2001)Kohn, Smith, and Chan}]{kohn2001nonparametric}
Kohn, R., Smith, M., and Chan, D. (2001).
\newblock \enquote{Nonparametric regression using linear combinations of basis
  functions.}
\newblock {\em Statistics and Computing\/}, 11(4): 313--322.
\endbibitem

\bibitem[{Lang and Brezger(2004)}]{lang2004bayesian}
Lang, S. and Brezger, A. (2004).
\newblock \enquote{{B}ayesian {P}-splines.}
\newblock {\em Journal of Computational and Graphical Statistics\/}, 13(1):
  183--212.
\endbibitem

\bibitem[{Li and Clyde(2018)}]{li2018mixtures}
Li, Y. and Clyde, M.~A. (2018).
\newblock \enquote{Mixtures of g-priors in generalized linear models.}
\newblock {\em Journal of the American Statistical Association\/}, 113(524):
  1828--1845.
\endbibitem

\bibitem[{Liang et~al.(2008)Liang, Paulo, Molina, Clyde, and
  Berger}]{liang2008mixtures}
Liang, F., Paulo, R., Molina, G., Clyde, M.~A., and Berger, J.~O. (2008).
\newblock \enquote{Mixtures of g priors for {B}ayesian variable selection.}
\newblock {\em Journal of the American Statistical Association\/}, 103(481):
  410--423.
\endbibitem

\bibitem[{Magee(1990)}]{magee1990r}
Magee, L. (1990).
\newblock \enquote{{$R^2$} measures based on {W}ald and likelihood ratio joint
  significance tests.}
\newblock {\em The American Statistician\/}, 44(3): 250--253.
\endbibitem

\bibitem[{Maruyama and George(2011)}]{maruyama2011fully}
Maruyama, Y. and George, E.~I. (2011).
\newblock \enquote{Fully {B}ayes factors with a generalized g-prior.}
\newblock {\em The Annals of Statistics\/}, 39(5): 2740--2765.
\endbibitem

\bibitem[{Nagelkerke(1991)}]{nagelkerke1991note}
Nagelkerke, N.~J. (1991).
\newblock \enquote{A note on a general definition of the coefficient of
  determination.}
\newblock {\em Biometrika\/}, 78(3): 691--692.
\endbibitem

\bibitem[{Nott and Kohn(2005)}]{nott2005adaptive}
Nott, D.~J. and Kohn, R. (2005).
\newblock \enquote{Adaptive sampling for {B}ayesian variable selection.}
\newblock {\em Biometrika\/}, 92(4): 747--763.
\endbibitem

\bibitem[{Park and Jeong(2018)}]{park2018analysis}
Park, T. and Jeong, S. (2018).
\newblock \enquote{Analysis of {P}oisson varying-coefficient models with
  autoregression.}
\newblock {\em Statistics\/}, 52(1): 34--49.
\endbibitem

\bibitem[{Rivoirard and Rousseau(2012)}]{rivoirard2012posterior}
Rivoirard, V. and Rousseau, J. (2012).
\newblock \enquote{Posterior concentration rates for infinite dimensional
  exponential families.}
\newblock {\em Bayesian Analysis\/}, 7(2): 311--334.
\endbibitem

\bibitem[{Rossell et~al.(2021)Rossell, Abril, and
  Bhattacharya}]{rossell2021approximate}
Rossell, D., Abril, O., and Bhattacharya, A. (2021).
\newblock \enquote{Approximate {L}aplace approximations for scalable model
  selection.}
\newblock {\em Journal of the Royal Statistical Society: Series B (Statistical
  Methodology)\/}, 83(4): 853--879.
\endbibitem

\bibitem[{Saban{\'e}s~Bov{\'e} and Held(2011)}]{bove2011hyper}
Saban{\'e}s~Bov{\'e}, D. and Held, L. (2011).
\newblock \enquote{Hyper-$ g $ priors for generalized linear models.}
\newblock {\em Bayesian Analysis\/}, 6(3): 387--410.
\endbibitem

\bibitem[{Saban{\'e}s~Bov{\'e} et~al.(2015)Saban{\'e}s~Bov{\'e}, Held, and
  Kauermann}]{bove2015objective}
Saban{\'e}s~Bov{\'e}, D., Held, L., and Kauermann, G. (2015).
\newblock \enquote{Objective {B}ayesian model selection in generalized additive
  models with penalized splines.}
\newblock {\em Journal of Computational and Graphical Statistics\/}, 24(2):
  394--415.
\endbibitem

\bibitem[{Scheipl and Kneib(2009)}]{scheipl2009locally}
Scheipl, F. and Kneib, T. (2009).
\newblock \enquote{Locally adaptive {B}ayesian {P}-splines with a
  {Normal-Exponential-Gamma} prior.}
\newblock {\em Computational Statistics \& Data Analysis\/}, 53(10):
  3533--3552.
\endbibitem

\bibitem[{Schmidt and Makalic(2020)}]{schmidt2020bayesian}
Schmidt, D.~F. and Makalic, E. (2020).
\newblock \enquote{Bayesian generalized horseshoe estimation of generalized
  linear models.}
\newblock In {\em Joint European Conference on Machine Learning and Knowledge
  Discovery in Databases\/}, 598--613. Springer.
\endbibitem

\bibitem[{Shen and Ghosal(2015)}]{shen2015adaptive}
Shen, W. and Ghosal, S. (2015).
\newblock \enquote{Adaptive {B}ayesian procedures using random series priors.}
\newblock {\em Scandinavian Journal of Statistics\/}, 42(4): 1194--1213.
\endbibitem

\bibitem[{Shun and McCullagh(1995)}]{shun1995laplace}
Shun, Z. and McCullagh, P. (1995).
\newblock \enquote{Laplace approximation of high dimensional integrals.}
\newblock {\em Journal of the Royal Statistical Society Series B: Statistical
  Methodology\/}, 57(4): 749--760.
\endbibitem

\bibitem[{Smith et~al.(1988)Smith, Everhart, Dickson, Knowler, and
  Johannes}]{smith1988using}
Smith, J.~W., Everhart, J.~E., Dickson, W., Knowler, W.~C., and Johannes, R.~S.
  (1988).
\newblock \enquote{Using the {ADAP} learning algorithm to forecast the onset of
  diabetes mellitus.}
\newblock In {\em Proceedings of the Annual Symposium on Computer Application
  in Medical Care\/}, 261. American Medical Informatics Association.
\endbibitem

\bibitem[{Smith and Kohn(1996)}]{smith1996nonparametric}
Smith, M. and Kohn, R. (1996).
\newblock \enquote{Nonparametric regression using {B}ayesian variable
  selection.}
\newblock {\em Journal of Econometrics\/}, 75(2): 317--343.
\endbibitem

\bibitem[{Sohn et~al.(2023)Sohn, Jeong, Cho, and Park}]{sohn2022functional}
Sohn, J., Jeong, S., Cho, Y.~M., and Park, T. (2023).
\newblock \enquote{Functional clustering methods for binary longitudinal data
  with temporal heterogeneity.}
\newblock {\em Computational Statistics \& Data Analysis\/}, 185: 107766.
\endbibitem

\bibitem[{Umlauf et~al.(2015)Umlauf, Adler, Kneib, Lang, and
  Zeileis}]{umlauf2012structured}
Umlauf, N., Adler, D., Kneib, T., Lang, S., and Zeileis, A. (2015).
\newblock \enquote{Structured Additive Regression Models: An R Interface to
  BayesX.}
\newblock {\em Journal of Statistical Software\/}, 63(21): 1--46.
\endbibitem

\bibitem[{Wang et~al.(2011)Wang, Liu, Liang, and Carroll}]{wang2011estimation}
Wang, L., Liu, X., Liang, H., and Carroll, R.~J. (2011).
\newblock \enquote{Estimation and variable selection for generalized additive
  partial linear models.}
\newblock {\em Annals of Statistics\/}, 39(4): 1827.
\endbibitem

\bibitem[{Wang and George(2007)}]{wang2007adaptive}
Wang, X. and George, E.~I. (2007).
\newblock \enquote{Adaptive Bayesian criteria in variable selection for
  generalized linear models.}
\newblock {\em Statistica Sinica\/}, 667--690.
\endbibitem

\bibitem[{Williams and Rasmussen(1995)}]{williams1995gaussian}
Williams, C. and Rasmussen, C. (1995).
\newblock \enquote{Gaussian processes for regression.}
\newblock {\em Advances in Neural Information Processing Systems\/}, 8.
\endbibitem

\bibitem[{Womack et~al.(2014)Womack, Le{\'o}n-Novelo, and
  Casella}]{womack2014inference}
Womack, A.~J., Le{\'o}n-Novelo, L., and Casella, G. (2014).
\newblock \enquote{Inference from intrinsic {B}ayes’ procedures under model
  selection and uncertainty.}
\newblock {\em Journal of the American Statistical Association\/}, 109(507):
  1040--1053.
\endbibitem

\bibitem[{Wood(2017)}]{wood2017generalized}
Wood, S.~N. (2017).
\newblock {\em Generalized Additive Models: an Introduction with R\/}.
\newblock CRC press.
\endbibitem

\bibitem[{Zellner(1986)}]{zellner1986assessing}
Zellner, A. (1986).
\newblock \enquote{On assessing prior distributions and {B}ayesian regression
  analysis with g-prior distributions.}
\newblock {\em Bayesian Inference and Decision Techniques: Essays in Honor of
  Bruno de Finetti\/}, 233--243.
\endbibitem

\bibitem[{Zellner and Siow(1980)}]{zellner1980posterior}
Zellner, A. and Siow, A. (1980).
\newblock \enquote{Posterior odds ratios for selected regression hypotheses.}
\newblock {\em Trabajos de Estad{\'\i}stica Y de Investigaci{\'o}n
  Operativa\/}, 31(1): 585--603.
\endbibitem

\end{thebibliography}

\begin{acks}[Acknowledgments]
	This research was supported by the Yonsei University Research Fund of 2021-22-0032 and the National Research Foundation of Korea (NRF) grant funded by the Korean government (MSIT) (2022R1C1C1006735, RS-2023-00217705). 
\end{acks}

\newpage

\renewcommand{\thesection}{S\arabic{section}}
\renewcommand{\theequation}{S\arabic{equation}}
\renewcommand{\thefigure}{S\arabic{figure}}
\renewcommand{\thetable}{S\arabic{table}}

\setcounter{figure}{0}
\setcounter{section}{0}
\setcounter{equation}{0}

\begin{frontmatter}
	\title{Supplementary Material of `Model Selection-Based Estimation for Generalized Additive Models Using Mixtures of g-priors: Towards Systematization'
	}
	\runtitle{Supplementary Material}
	
	\begin{aug}
		\author{\fnms{Gyeonghun} \snm{Kang}}
		\and
		\author{\fnms{Seonghyun} \snm{Jeong}}
		
		\runauthor{G. Kang and S. Jeong}
		
		
		
		
	\end{aug}
	
	\begin{abstract}
		This supplementary material covers the details of truncated compound confluent hypergeometric (tCCH) distributions, sampling strategies, proofs of propositions, derivation of the posterior through the Laplace approximation, an additional simulation study with Poisson and Gaussian regression models, and instructions for installing the R package.
	\end{abstract}
	
\end{frontmatter}

\section{Truncated compound confluent hypergeometric distributions}

The tCCH distribution, as formally defined by \citet{li2018mixtures}, is a slight modification of the generalized beta distribution defined by \citet{gordy1998generalization}. Specifically, we denote $V \sim \text{tCCH}(a,b,z,s,\nu,\kappa)$ if $V$ has a density of the form
\begin{align}
	f(u) &= \frac{\nu^a u^{a-1} (1-\nu u)^{b-1} [\kappa + (1-\kappa)\nu u]^{-z} e^{s/\nu}e^{-su}}
	{ \Phi_1 (b, z, a+b, s/\nu, 1-\kappa)B(a,b)} \mathbbm 1 \{0<u<1/\nu\},
	\label{seqn:tcchdensity}
\end{align}
where $a>0$, $b>0$, $z\in \mathbb{R}$, $s \in \mathbb{R}$,  $\nu \geq 1$ and $\kappa>0$.
A direct calculation yields the $k$th moment as:
\begin{align}
	E(V^k) &= \nu^{-k}\frac{B(a+k,b) \; \Phi_1\left(b, z, a+b+k, s/\nu, 1-\kappa\right)}
	{B(a,b) \Phi_1\left(b, z, a+b, s/\nu, 1-\kappa\right)}.
	\label{seqn:tcchmoment}
\end{align}

The tCCH distribution can be reduced to several other distributions depending on the parameter values. For example, it can take the form of a Gaussian hypergeometric distribution \citep{armero1994prior}, a confluent hypergeometric distribution \citep{gordy1998computationally}, a beta distribution, or a gamma distribution. For further details, see \citet[p.132, p.279]{gupta2004handbook}. Consequently, the marginal likelihood in \eqref{eqn:mixedg} is simplified based on the parameters of the tCCH prior.

\section{Proofs of the propositions}

\begin{proof}[Proof of Proposition~\ref{prop:ncspace}]
	Consider boundary knots $\{t^L,t^U\}$ and interior knots $\{t_1,\dots, t_M\}$ satisfying $t^L<t_1<\dots<t_M<t^U$. To concatenate the expressions, we write $t_0=t^L$ and $t_{M+1}=t^U$. The common expression of the natural cubic splines derived from the truncated cubic spline basis functions is given by
	\begin{align*}
		N_1^\ast(u)&=u, \\
		N_{k+2}^\ast(u)&=\dfrac{(u-t_k)_+^3 - (u-t_{M+1})_+^3}{t_{M+1} - t_k} - \dfrac{(u-t_{M})_+^3 - (u-t_{M+1})_+^3}{t_{M+1} - t_{M}},\quad k=0,\dots,M-1,
	\end{align*}
	(see, for example, Equations (5.4) and (5.5) in \citet{hastie2009elements}). Letting $N_0^\ast(u)=1$, it is well known that ${\mathcal N^\ast}=\{N_k^\ast,k=0, 1,\dots, M+1\}$ is a basis for the cubic spline space with the natural boundary conditions. Therefore, it suffices to show that there exists an injection $Q:\mathcal N\mapsto {\mathcal N^\ast}$. Given that $N_0= N_0^\ast$, $N_1=N_1^\ast$, $-N_{M+1}=N_2^\ast$ and $N_{k-1}-N_{M+1}=N_k^\ast$, $k=3,\dots,M+1$, we obtain
	\begin{align*}
		Q=\begin{pmatrix}
			1 & 0 & 0 & 0 & \dots & 0 & 0 \\
			0 & 1 & 0 & 0 & \dots & 0 & 0 \\
			0 & 0 & 0 & 0 & \dots & 0 & -1 \\
			0 & 0 & 1 & 0 & \dots & 0 & -1 \\
			0 & 0 & 0 & 1 & \dots & 0 & -1 \\
			\vdots & \vdots & \vdots & \vdots & \ddots & \vdots & \vdots \\
			0 & 0 & 0 & 0 & \dots & 1 & -1 
		\end{pmatrix},
	\end{align*}
	which is clearly nonsingular.
\end{proof}

\begin{proof}[Proof of Proposition~\ref{prop:selection}]
	Each of the basis terms $N(\cdot ; t^L,t^U,t_k)$, $k=1,\dots,M$, in $\mathcal N$ depends solely on $t^L$, $t^U$, and $t_k$.
	Thus, introducing a new knot point $t_\ast$ adds a new basis term $N(\cdot ; t^L,t^U,t_\ast)$ without affecting the existing basis terms. Conversely, removing an existing knot point $t_k$ eliminates the corresponding basis term $N(\cdot ; t^L,t^U,t_k)$ while leaving the other terms unchanged.
\end{proof}

\begin{proof}[Proof of Proposition~\ref{prop:bf}] Suppose the tCCH prior in \eqref{seqn:tcchdensity} is chosen.
	If $\hat\eta_{\xi_{(1)}} = \hat\eta_{\xi_{(2)}}$, then we obtain that
	\begin{align*}
		BF[\xi_{(1)};\xi_{(2)}] &=\nu^{-k/2} \frac{B((a+J_{\xi_{(2)}}+k)/2, b/2) }{B((a+J_{\xi_{(2)}})/2, b/2) }\\
		&\quad\times
		\frac{\Phi_1\left(b/2, r, (a+b+J_{\xi_{(2)}}+k)/2, (s+Q_{\xi_{(2)}})/(2\nu), 1-\kappa\right)  }
		{\Phi_1\left(b/2, r, (a+b+J_{\xi_{(2)}})/2, (s+Q_{\xi_{(2)}})/(2\nu), 1-\kappa\right)  },
	\end{align*}
	using the expression in \eqref{eqn:mixedg}. 
	Given that the posterior of $(g+1)^{-1}$ is the tCCH distribution in the first line of \eqref{eqn:posterior}, the assertion is easily verified using \eqref{seqn:tcchmoment} if the tCCH prior is used for $(g+1)^{-1}$. A similar proof can be extended to the case of the unit information prior.
\end{proof}

\section{Laplace approximation to the marginal likelihood}
\label{sec:postderv}

Following the proof of Proposition~1 in\citet{li2018mixtures}, the Laplace approximation of the likelihood yields
\begin{align*}
	p(Y\mid \alpha, \beta_\xi, \xi) \approx	p(Y\mid \hat\eta_\xi) \exp \!\left\{
	-\frac{(\alpha-\hat\alpha_\xi+m)^2}{2\text{tr}(\mathcal{J}_n(\hat\eta_\xi))}
	-\frac{1}{2}(\beta_\xi - \hat\beta_\xi)^T \tilde{B}_\xi^T \mathcal{J}_n(\hat\eta_\xi) \tilde{B}_\xi (\beta_\xi - \hat\beta_\xi)
	\right\}\!,
\end{align*}
where $m=\text{tr}(\mathcal{J}_n(\hat\eta_\xi))^{-1}1_n^T \mathcal{J}_n(\hat\eta_\xi) B_\xi(\beta_\xi-\hat\beta_\xi)$. Combined with the prior $ \pi(\alpha) \pi(\beta_\xi\mid g,\xi)$ in \eqref{eqn:alphaprior} and \eqref{eqn:gprior}, this verifies the second and third lines of \eqref{eqn:posterior}.
Thus, the verification of \eqref{eqn:fixedg} is straightforward, as shown by
\begin{align*}
	p(Y\mid g, \xi) 
	&\approx
	\int \int \pi(\alpha) \pi(\beta_\xi\mid g,\xi)
	p(Y\mid \alpha, \beta_\xi,\xi) d\alpha d\beta_\xi\nonumber
	\\&=
	p(Y\mid \hat\eta_\xi)\text{tr}(\mathcal{J}_n(\hat{\eta}_\xi))^{-1/2}
	(g+1)^{-{J_\xi}/{2}} 
	\exp \!\left(-\frac{Q_\xi}{2(g+1)}\right).
\end{align*}
Combined with the tCCH prior in \eqref{eqn:tcchprior}, this verifies the first line of \eqref{eqn:posterior} using the density in \eqref{seqn:tcchdensity}.
Now, we can marginalize out $g$, that is,
\begin{align*}
	\begin{split}
		&p(Y\mid \xi)\\
		&~\approx
		\frac{p(Y\mid\hat{\eta}_{\xi})\text{tr}(\mathcal{J}_n(\hat{\eta}_\xi))^{-1/2} \nu^{a/2} e^{s/(2\nu)}/B(a/2, b/2)}{\Phi_1\left(b/2, r, (a+b)/2, s/(2\nu), 1-\kappa\right)}
		\int_0^{1/\nu} 
		\frac{u^{(a+J_\xi)/2-1} (1-\nu u)^{b/2-1} }{[\kappa + (1-\kappa)\nu u]^r e^{(s+Q_\xi)u/2} } du \\
		&~= p(Y\mid\hat{\eta}_{\xi}) \text{tr}(\mathcal{J}_n(\hat{\eta}_\xi))^{-1/2} 
		\nu^{-J_\xi / 2} \exp \!\left(-\frac{Q_\xi}{2\nu}\right)\frac{B((a+J_\xi)/2, b/2) }{B(a/2, b/2) }\\
		&~\quad\times \Phi_1\!\left(\frac{b}{2}, r, \frac{a+b+J_\xi}{2}, \frac{s+Q_\xi}{2\nu}, 1-\kappa\right)  
		\bigg/ \Phi_1\!\left(\frac{b}{2}, r, \frac{a+b}{2}, \frac{s}{2\nu}, 1-\kappa\right),
	\end{split}	
\end{align*}
where the equality holds by the change of variables $v=\nu u$. This verifies \eqref{eqn:mixedg}.

\section{Sampling from tCCH distributions}
\subsection{Exact sampling when $b=1$ and $\kappa=1$}
\label{sec:sampling}

Given \eqref{seqn:tcchdensity}, the density of $ \text{tCCH}(a,1,z,s,\nu,1)$ has the form $f(u)\propto u^{a-1} e^{-su} \mathbbm 1 \{0<u<1/\nu\}$, which is the gamma density truncated to $0<u<1/\nu$. Therefore, exact sampling from tCCH distributions using the inverse transform method is straightforward in this case. Combining the first line of \eqref{eqn:posterior} with Table~\ref{table:mixg} reveals that the posterior distribution $\Pi((g+1)^{-1}\mid Y,\xi)$ simplifies to a truncated gamma distribution when using the uniform, hyper-g, ZS-adapted, or robust prior. 

\subsection{Slice sampling when $z>0$}

The posterior distributions resulting from the hyper-g/n, beta-prime, and intrinsic prior distributions do not simplify to truncated gamma distributions. In these cases, we can utilize a version of slice sampling. 

To generate samples from $V\sim \text{tCCH}(a,b,z,s,\nu,\kappa)$ with $z>0$, we use the change of variable $W=\nu V$ with reparameterization $\xi=\kappa^{-1}-1$ and $\zeta=s/\nu$. The density of $W$ is given by
\begin{align*}
	f(w) &\propto w^{a-1} (1-w)^{b-1} (1 + \xi w)^{-z} e^{-\zeta w}\mathbbm 1\{0<w<1\} \\
	&= w^{a-1} (1-w)^{b-1} e^{-\zeta w}\mathbbm 1\{0<w<1\} \Gamma(z)^{-1}\int_{0}^{\infty}  t^{z-1}e^{-(1+\xi w)t}dt,
\end{align*}
where $\Gamma$ is the gamma function. Therefore, $f$ can be obtained as the marginal density from the joint density
\begin{align*}
	h(w,t,u_1,u_2) &\propto w^{a-1} (1-w)^{b-1}t^{z-1}e^{-t}\\
	&\quad\times\mathbbm 1\{0<u_1<e^{-\zeta w}\}\mathbbm 1\{0<u_2<e^{-\xi w t}\}   \mathbbm 1\{0<w<1\}.
\end{align*}
Given that $\xi>0$ and $\zeta>0$, a slice sampler is constructed as
\begin{align*}
	U_1\mid U_2, T, W &\sim \text{Unif}(0, e^{-\zeta W}),\\
	U_2\mid U_1, T, W &\sim \text{Unif}(0, e^{-\xi T W}),\\
	T\mid U_1,U_2,W &\sim \text{Gamma}(z, 1)\times\mathbbm 1\{-\infty < T < -(\log U_2)/(\xi W)\},\\
	W\mid U_1,U_2, T &\sim \text{Beta}(a,b)\times\mathbbm 1\{-\min\{\log U_1,(\log U_2)/T\}<W<1 \}.
\end{align*}

\section{Gaussian additive regression with unknown precision}
\label{sec:Gaussian}
Thus far, we have focused on GAMs with a known dispersion parameter $\phi$ for the exponential family models. Now, we shift to a more traditional setup by assuming a Gaussian distribution for $Y_i$, treating $\phi$ as an unknown parameter.
In the context of Gaussian additive regression, the response variable $Y_i$ is expressed as
\begin{align}\label{eqn:gam_gau}
	Y_i = \alpha + \sum_{j=1}^p f_j (x_{ij})+\epsilon_i,\quad \epsilon_i \sim {\text N}(0,\phi^{-1}), \quad i = 1, \dots, n,
\end{align}
where the precision parameter $\phi$ is typically unknown. Although model \eqref{eqn:gam_gau} also falls within the GAM framework, the presence of the unknown precision parameter $\phi$ introduces some distinctions.
Let $\eta=(\eta_1,\dots,\eta_n)^T$ be the vector of the mean responses, that is,  $\eta_i=E(Y_i)$.
We parameterize $\eta$ as $\eta=\alpha 1_n + B_\xi\beta_\xi$ using $\alpha$, $B_\xi$, and $\beta_\xi$ defined in Section~\ref{sec2}. 
In line with the convention, an improper prior is assigned to $(\alpha,\phi)$, that is,
\begin{align*}
	\pi(\alpha,\phi)\propto 1/\phi.
\end{align*}
Given that the information matrix of a Gaussian distribution is the identity matrix, we can easily verify that the prior in \eqref{eqn:gprior} simplifies to the standard g-prior distribution \citep{zellner1986assessing,liang2008mixtures},
\begin{align*}
	\beta_\xi\mid \phi,g,\xi &\sim N\big(0, g\phi^{-1} (B_\xi^T B_\xi)^{-1} \big).
\end{align*}
Note that, in the Gaussian case, we have $\tilde{B}_\xi=B_\xi$ because the columns of $B_\xi$ are centered.

By combining the marginal likelihood with one of the priors for $\xi$ discussed in Section~\ref{sec:knotprior}, we derive the marginal posterior of $\xi$, denoted as $\Pi(\xi\mid Y)$. Calculating the marginal likelihood is complex because it requires integrating not only $g$ but also $\phi$.
First, it is well known that
\begin{align}
	p(Y\mid g, \xi) = 
	p(Y\mid \varnothing)
	\frac{(1+g)^{(n - J_\xi - 1)/2}}{[1 + g(1-R_\xi^2)] ^{(n-1)/2}},
	\label{eqn:gausmg}
\end{align}
where $p(Y\mid \varnothing) = 
n^{-1/2} (2\pi)^{-(n-1)/2} \Gamma((n-1)/2) (\|Y - \bar Y 1_n  \|^2/2)^{-(n-1)/2}$ is the marginal likelihood in the intercept-only model, $R_\xi^2 = \|B_\xi( B_\xi^T B_\xi)^{-1}  B_\xi^T  Y \|^2/ \|Y - \bar Y 1_n \|^2$ is the coefficient of determination with $\xi$, and $\bar Y=n^{-1}\sum_{i=1}^n Y_i$ is the average of the observations. For the unit information prior $\Pi(g)=\delta_n(g)$, the marginal likelihood $p(Y\mid \xi)$ is readily available from the expression in \eqref{eqn:gausmg}.
By assigning the tCCH prior in \eqref{eqn:tcchprior} to $(g+1)^{-1}$, \citet{li2018mixtures} shows that if $r=0$ (or $\kappa=1$ equivalently),
\begin{align*}
	\begin{split}
		p(Y\mid \xi) &= 
		\frac{p(Y\mid \varnothing)}{\nu^{J_\xi/2} [1 - (1 - \nu^{-1}) R_\xi^2]^{(n-1)/2}}
		\frac{B((a+J_\xi)/2, b/2)}{B(a/2, b/2)} 
		\\
		&\quad \times 
		\Phi_1 \bigg(\frac{b}{2}, \frac{n-1}{2}, \frac{a+b+J_\xi}{2}, \frac{s}{2\nu}, \frac{R_\xi^2}{\nu - (\nu-1) R_\xi^2}\bigg) \bigg / \;_1F_1\bigg(\frac{b}{2}, \frac{a+b}{2}, \frac{s}{2\nu}\bigg),
	\end{split}
\end{align*}
and if $s=0$,
\begin{align*}
	\begin{split}
		p(Y\mid \xi) &=
		\frac{p(Y\mid \varnothing) \kappa^{(a+J_\xi-2r)/2}}{\nu^{J_\xi/2} (1 - R_\xi^2)^{(n-1)/2}}
		\frac{B((a+J_\xi)/2, b/2)}{B(a/2, b/2)} 
		\\
		&\quad \times 
		F_1 \bigg( \frac{a+J_\xi}{2};\frac{a+b+J_\xi+1-n-2r}{2},\frac{n-1}{2}; \\ 
		&\qquad\qquad\frac{a+b+J_\xi}{2};1-\kappa,1-\kappa-\frac{R_\xi^2\kappa}{(1-R_\xi^2)v}\bigg)\bigg / \,_2F_1\bigg(r,\frac{b}{2}; \frac{a+b}{2}; 1-\kappa\bigg),
	\end{split}
\end{align*}
where $\,_1 F_1 (\alpha,\gamma,x) = \Phi_1(\alpha,0,\gamma,x,0)$, $\gamma>\alpha>0$, is the confluent hypergeometric function, $\,_2 F_1 (\beta,\alpha;\gamma;y) = \Phi_1(\alpha,\beta,\gamma,0,y)$, $\gamma>\alpha>0$, is the Gaussian hypergeometric function, and $F_1$ is the  the Appell hypergeometric function defined as $F_1(\alpha;\beta,\beta';\gamma;x,y)=B(\gamma - \alpha, \alpha)^{-1} \int_{0}^{1} u^{\alpha - 1} (1-u)^{\gamma - \alpha -1}(1-xu)^{-\beta}(1-yu)^{-\beta'}du$ for $\gamma>\alpha>0$. The prior distributions listed in Table~\ref{table:mixg} fall into one of the above two cases. While these expressions can be further simplified depending on the hyperparameters of the tCCH prior, numerical evaluation of the transcendental functions is often required.
The only exception is the beta-prime prior, which offers a closed-form expression for the marginal likelihood without involving hypergeometric-type transcendental functions; see \citet{maruyama2011fully}.

In the Gaussian case, the conditional posterior $\Pi((g+1)^{-1}\mid Y,\xi)$ is no longer a conjugate update of the tCCH prior. Nonetheless, it can be simplified with certain hyperparameter specifications of the tCCH prior, and sampling from $\Pi((g+1)^{-1}\mid Y,\xi)$ can be efficiently performed by introducing auxiliary variables. In particular, the beta-prime prior provides an exact sampling scheme from the beta distribution; see \citet{jeong2021bayesian}. 
The remaining specifications of the joint posterior can be derived through direct calculations as
\begin{align*}
	\phi \mid Y,g,\xi &\sim \text{Gamma} \!
	\left(\frac{n-1}{2}, \frac{\|Y-\bar Y 1_n  \|^2 [1+g(1-R_\xi^2)]}{2(1+g)}\right), \\
	\alpha \mid Y, \phi, g, \xi &\sim 
	\text{N}\!\left(\bar Y , \phi^{-1}/n\right),\\
	\beta_\xi \mid Y, \phi,g,\xi &\sim 
	\text{N}\!\left(\frac{g}{g+1}\hat\beta_\xi,\frac{g \phi^{-1}}{g+1}(B_\xi^T  B_\xi)^{-1}\right).
\end{align*}
The marginal posterior of $\xi$, $\Pi(\xi\mid Y)$, can be readily obtained from the expressions for the marginal likelihood provided above.

\begin{figure}[t!]
	\centering
	\subfloat[\centering Pointwise posterior mean estimates of $f_1$ in $100$ replications]{\includegraphics[width=1\textwidth]{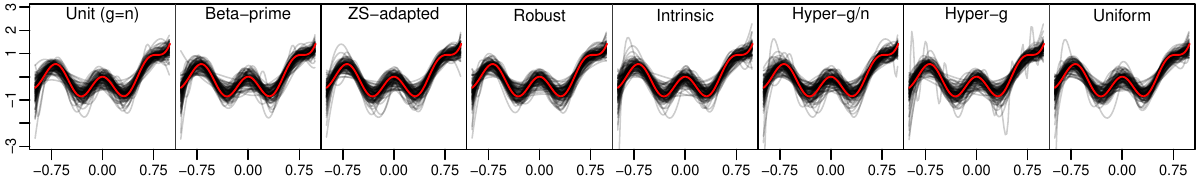}}\\
	\subfloat[\centering Pointwise posterior mean estimates of $f_2$ in $100$ replications]{\includegraphics[width=1\textwidth]{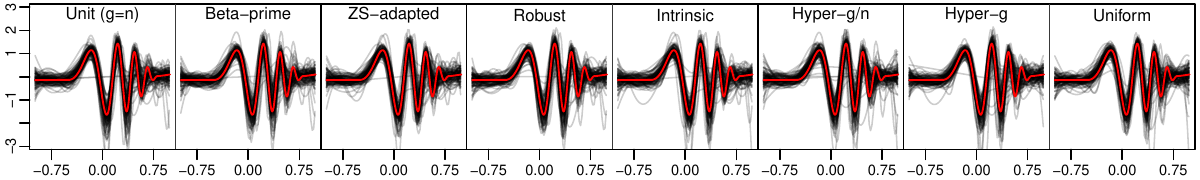}}\\
	\subfloat[\centering Pointwise posterior mean estimates of $f_3$ in $100$ replications]{\includegraphics[width=1\textwidth]{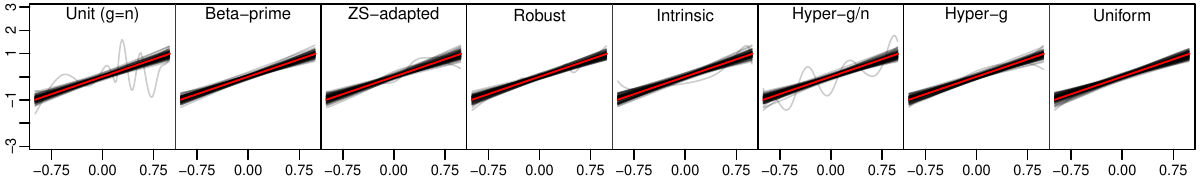}}\\
	\subfloat[\centering Pointwise posterior mean estimates of $f_4$ in $100$ replications]{\includegraphics[width=1\textwidth]{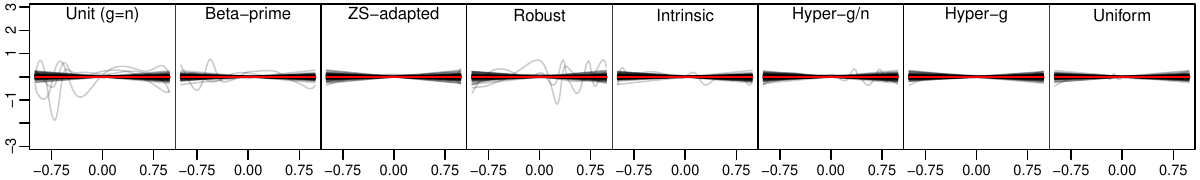}}
	\caption{Pointwise posterior means (gray) of $f_1$, $f_2$, $f_3$, and $f_4$ in the nonparametric Poisson regression model with $n=100$, obtained from randomly chosen 100 replicated datasets, along with the true function (red).}
	\label{plot:poi1}
\end{figure}

\begin{figure}[t!]
	\centering
	\subfloat[\centering Logarithm of RMSE for $f_1$]{\includegraphics[width = 0.47\textwidth]{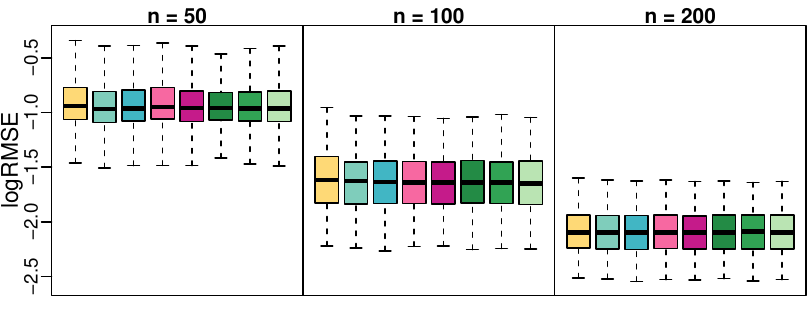}}
	~
	\subfloat[\centering Coverage probabilities for $f_1$]{\includegraphics[width = 0.47\textwidth]{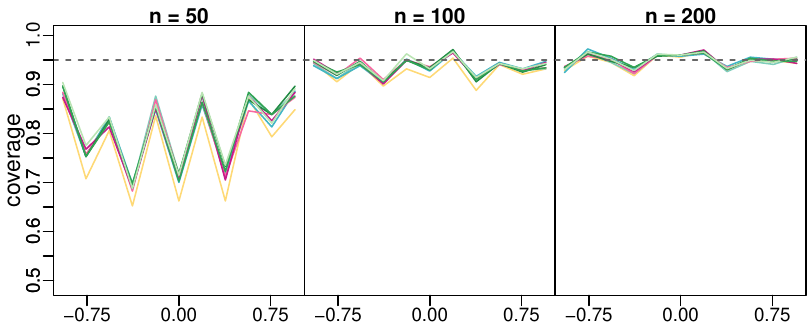}}
	\\
	\subfloat[\centering Logarithm of RMSE for $f_2$]{\includegraphics[width = 0.47\textwidth]{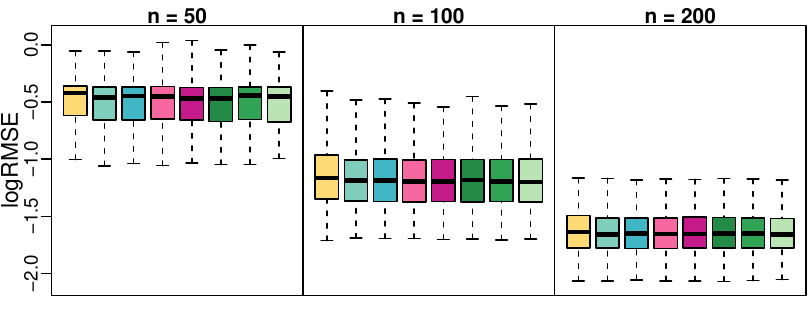}}
	~
	\subfloat[\centering Coverage probabilities for $f_2$]{\includegraphics[width = 0.47\textwidth]{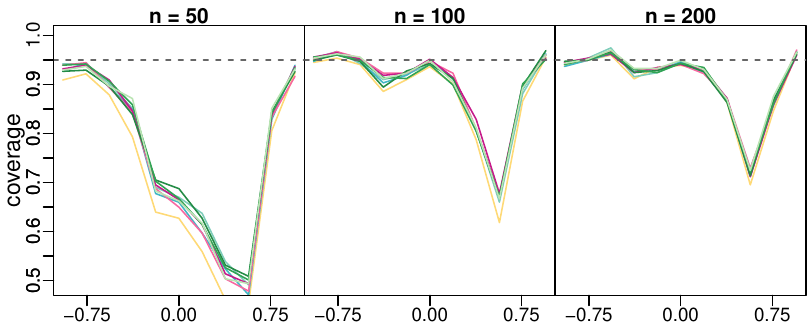}}
	\\
	\subfloat[\centering Logarithm of RMSE for $f_3$]{\includegraphics[width = 0.47\textwidth]{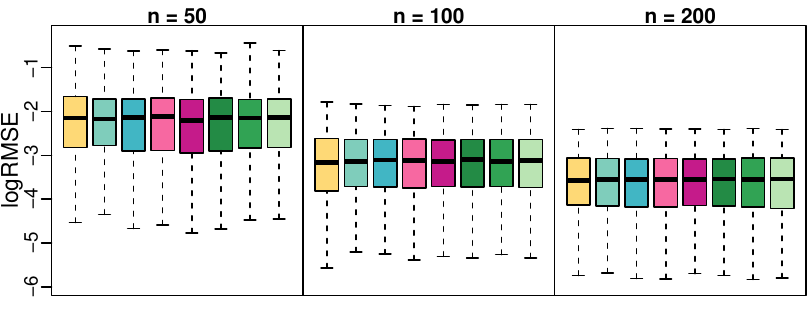}}
	~
	\subfloat[\centering Coverage probabilities for $f_3$]{\includegraphics[width = 0.47\textwidth]{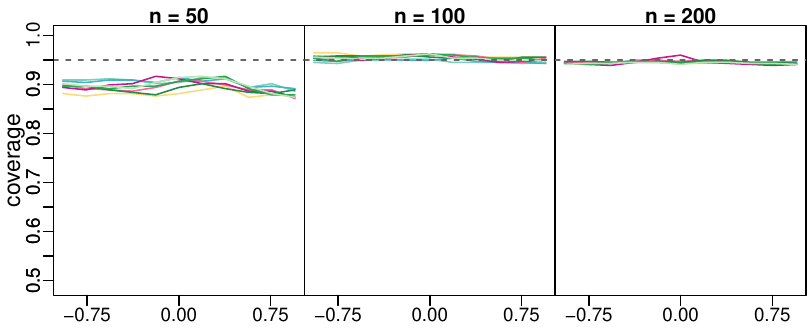}}
	\\
	\subfloat[\centering Logarithm of RMSE for $f_4$]{\includegraphics[width = 0.47\textwidth]{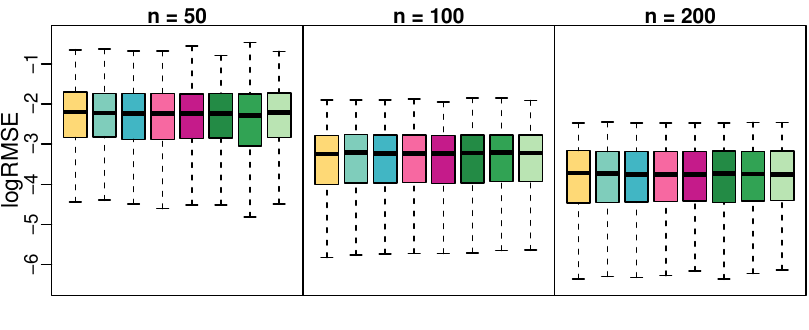}}
	~
	\subfloat[\centering Coverage probabilities for $f_4$]{\includegraphics[width = 0.47\textwidth]{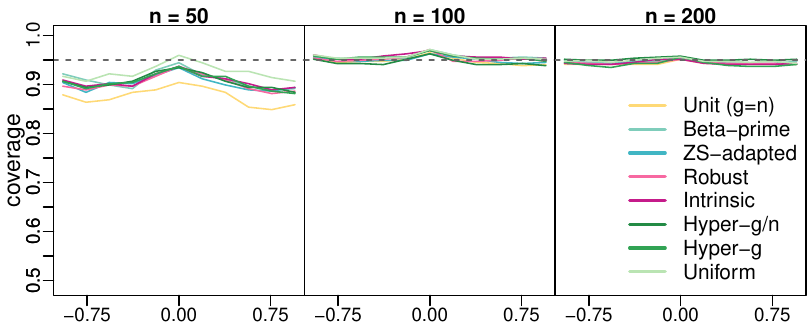}}
	\caption{Logarithm of RMSE and coverage probabilities for $f_1$, $f_2$, $f_3$, and $f_4$ in the nonparametric Poisson regression models with $n=50, 100, 200$, obtained from 500 replicated datasets. Outliers are excluded to improve visualization.}
	\label{plot:poi2}
\end{figure}

\begin{figure}[t!]
	\centering
	\subfloat[\centering Pointwise posterior mean estimates of $f_1$ in $100$ replications]{\includegraphics[width=1\textwidth]{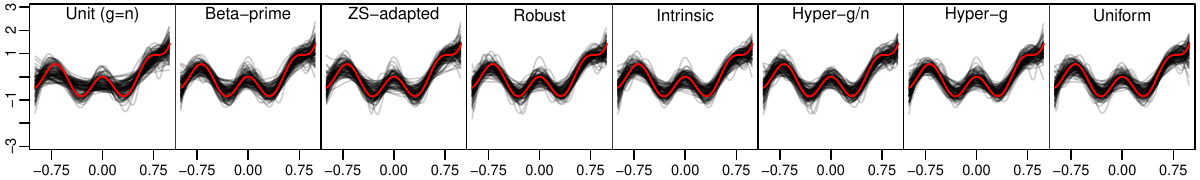}}\\
	\subfloat[\centering Pointwise posterior mean estimates of $f_2$ in $100$ replications]{\includegraphics[width=1\textwidth]{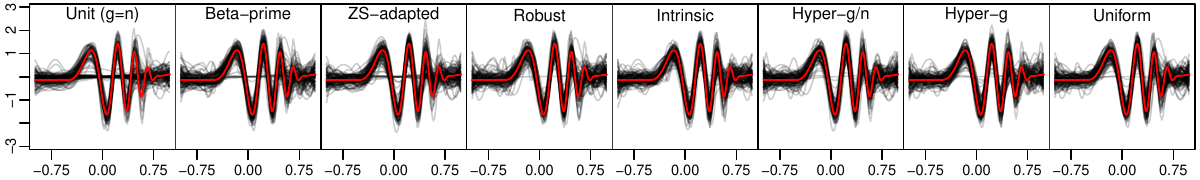}}\\
	\subfloat[\centering Pointwise posterior mean estimates of $f_3$ in $100$ replications]{\includegraphics[width=1\textwidth]{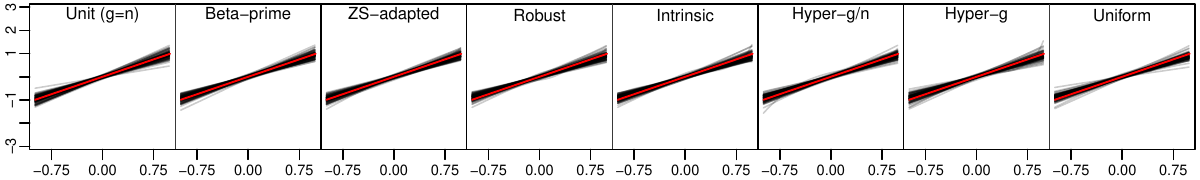}}\\
	\subfloat[\centering Pointwise posterior mean estimates of $f_4$ in $100$ replications]{\includegraphics[width=1\textwidth]{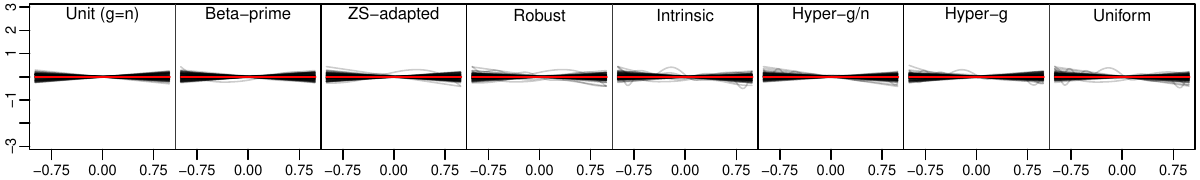}}
	\caption{Pointwise posterior means (gray) of $f_1$, $f_2$, $f_3$ and $f_4$ in the nonparametric Gaussian regression model with $n=200$, obtained from randomly chosen 100 replicated datasets, along with the true function (red).	}
	\label{plot:gau1}
\end{figure}

\begin{figure}[t!]
	\centering
	\subfloat[\centering Logarithm of RMSE for $f_1$]{\includegraphics[width = 0.47\textwidth]{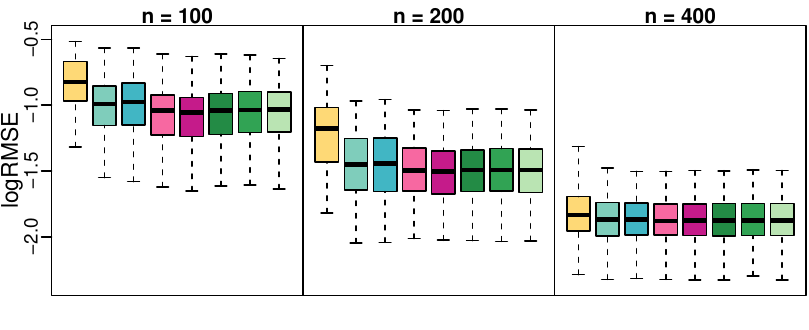}}
	~
	\subfloat[\centering Coverage probabilities for $f_1$]{\includegraphics[width = 0.47\textwidth]{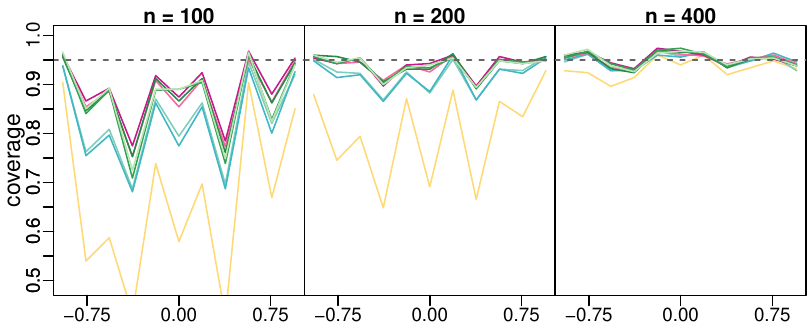}}
	\\
	\subfloat[\centering Logarithm of RMSE for $f_2$]{\includegraphics[width = 0.47\textwidth]{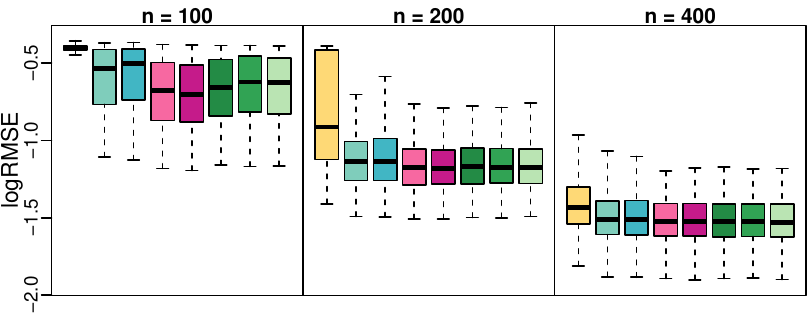}}
	~
	\subfloat[\centering Coverage probabilities for $f_2$]{\includegraphics[width = 0.47\textwidth]{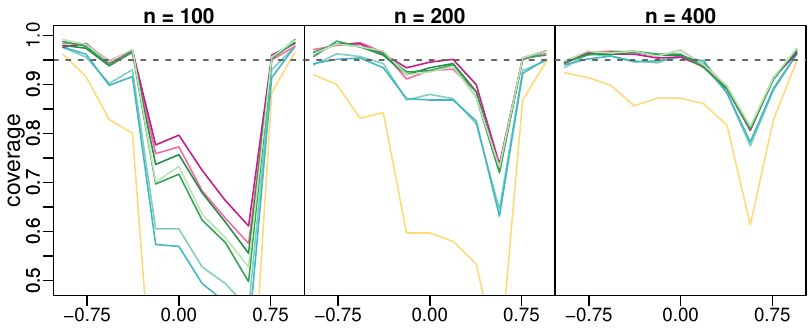}}
	\\
	\subfloat[\centering Logarithm of RMSE for $f_3$]{\includegraphics[width = 0.47\textwidth]{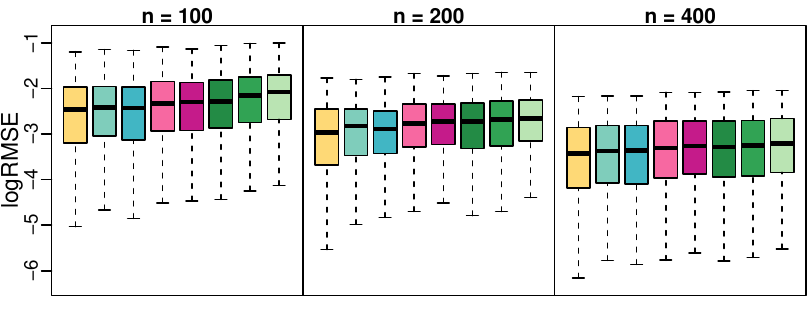}}
	~
	\subfloat[\centering Coverage probabilities for $f_3$]{\includegraphics[width = 0.47\textwidth]{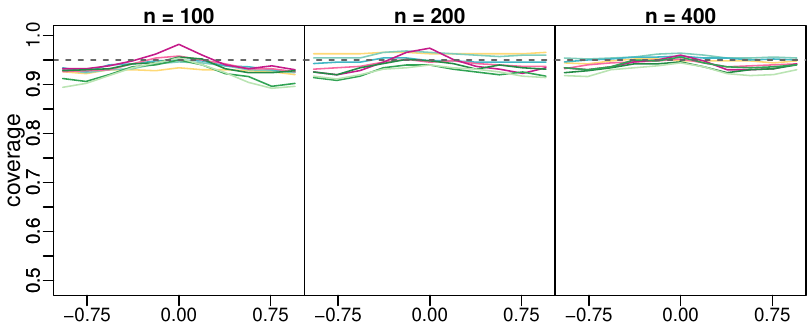}}
	\\
	\subfloat[\centering Logarithm of RMSE for $f_4$]{\includegraphics[width = 0.47\textwidth]{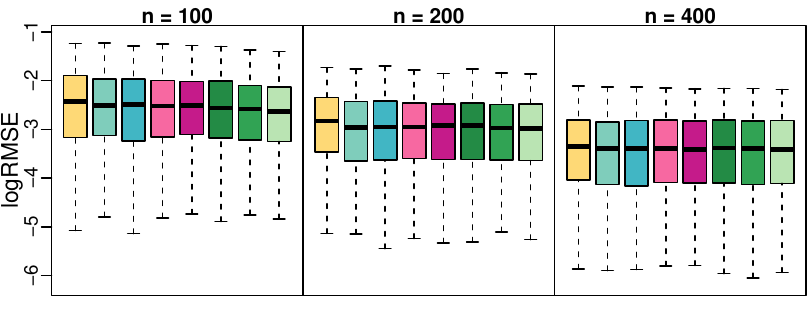}}
	~
	\subfloat[\centering Coverage probabilities for $f_4$]{\includegraphics[width = 0.47\textwidth]{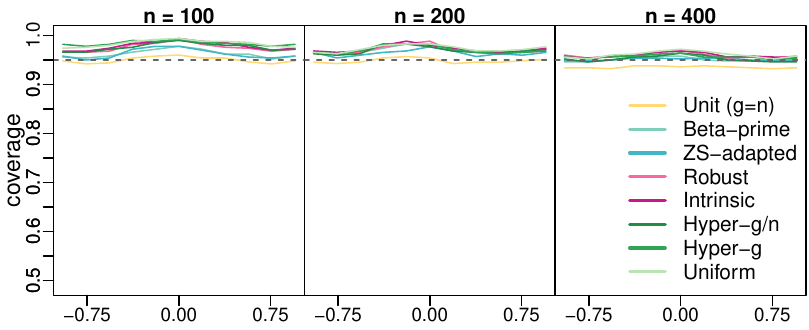}}
	\caption{Logarithm of RMSE and coverage probabilities for $f_1$, $f_2$, $f_3$ and $f_4$ in the nonparametric Gaussian regression models with $n=100, 200, 400$, obtained from 500 replicated datasets. Outliers are excluded to improve visualization.}
	\label{plot:gau2}
\end{figure}

\begin{figure}[t!]
	\centering
	\subfloat[\centering Pointwise posterior mean estimates of $f_1$ in $100$ replications]{\includegraphics[width=1\textwidth]{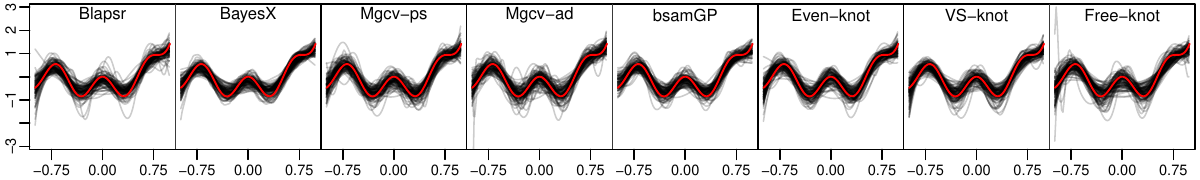}}\\
	\subfloat[\centering Pointwise posterior mean estimates of $f_2$ in $100$ replications]{\includegraphics[width=1\textwidth]{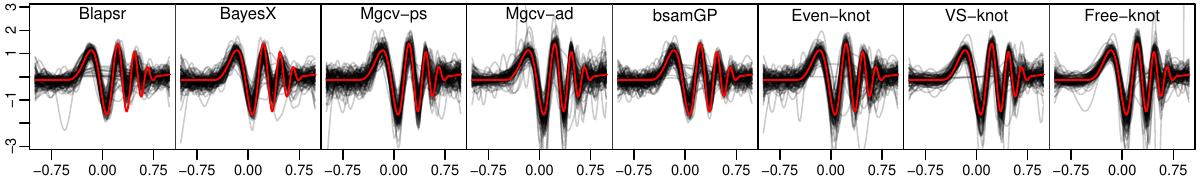}}\\
	\subfloat[\centering Pointwise posterior mean estimates of $f_3$ in $100$ replications]{\includegraphics[width=1\textwidth]{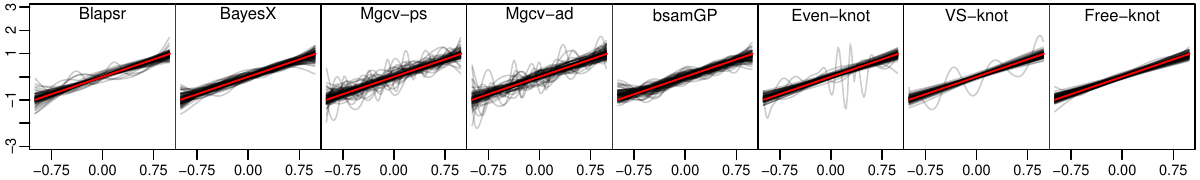}}\\
	\subfloat[\centering Pointwise posterior mean estimates of $f_4$ in $100$ replications]{\includegraphics[width=1\textwidth]{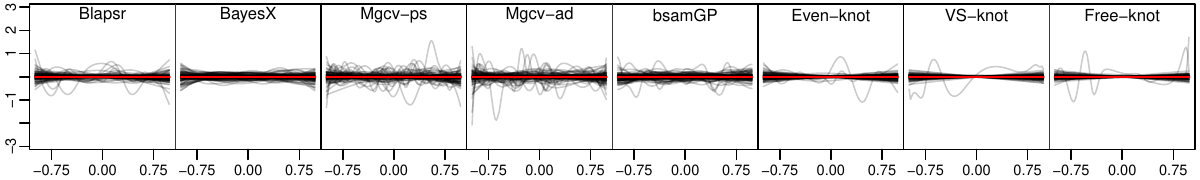}}
	\caption{Pointwise posterior means (gray) of $f_1$, $f_2$, $f_3$, and $f_4$ in the nonparametric Poisson regression model with $n=100$, obtained from randomly chosen 100 replicated datasets, along with the true function (red).}
	\label{plot:poi3}
\end{figure}

\begin{figure}[t!]
	\centering
	\subfloat[\centering Logarithm of RMSE for $f_1$]{\includegraphics[width = 0.47\textwidth]{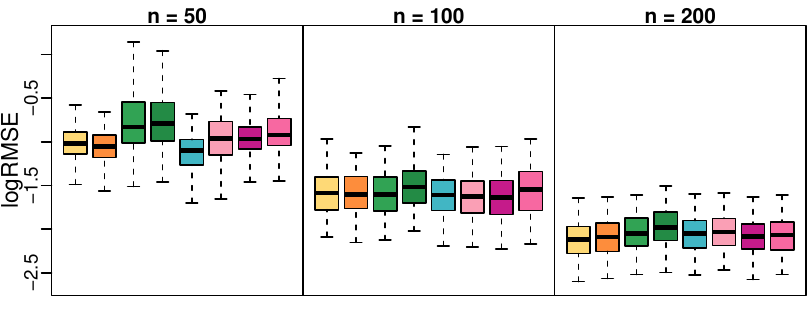}}
	~
	\subfloat[\centering Coverage probabilities for $f_1$]{\includegraphics[width = 0.47\textwidth]{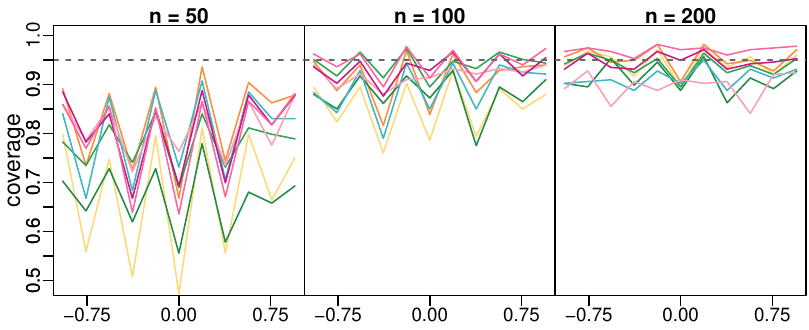}}
	\\
	\subfloat[\centering Logarithm of RMSE for $f_2$]{\includegraphics[width = 0.47\textwidth]{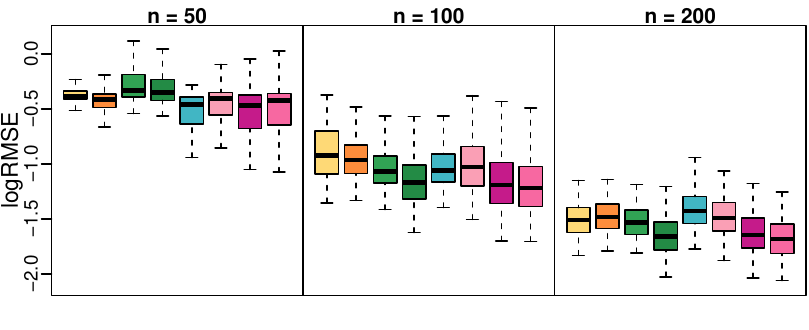}}
	~
	\subfloat[\centering Coverage probabilities for $f_2$]{\includegraphics[width = 0.47\textwidth]{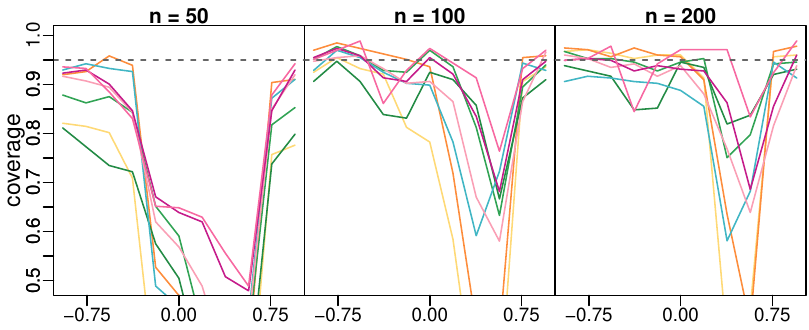}}
	\\
	\subfloat[\centering Logarithm of RMSE for $f_3$]{\includegraphics[width = 0.47\textwidth]{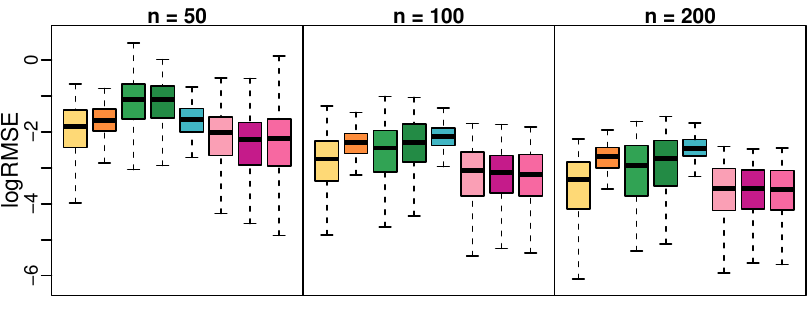}}
	~
	\subfloat[\centering Coverage probabilities for $f_3$]{\includegraphics[width = 0.47\textwidth]{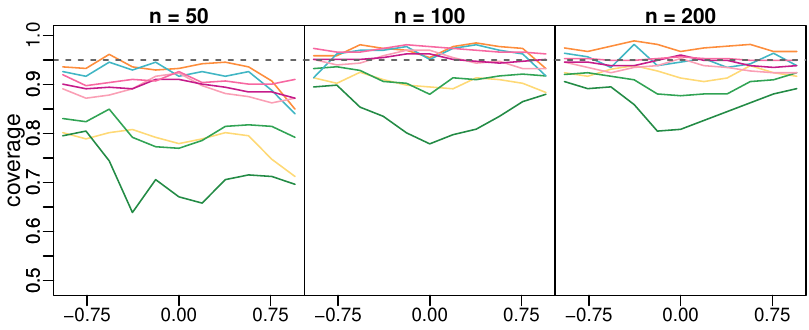}}
	\\
	\subfloat[\centering Logarithm of RMSE for $f_4$]{\includegraphics[width = 0.47\textwidth]{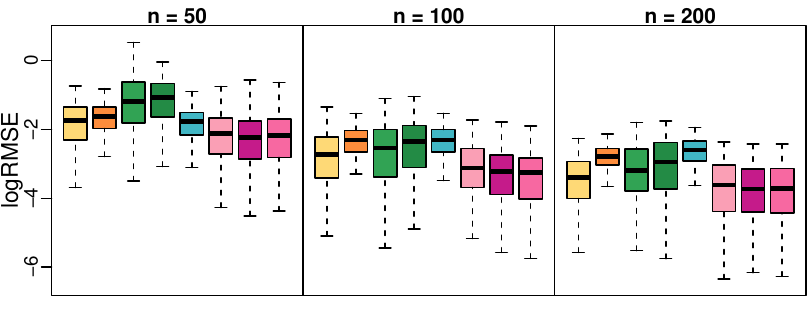}}
	~
	\subfloat[\centering Coverage probabilities for $f_4$]{\includegraphics[width = 0.47\textwidth]{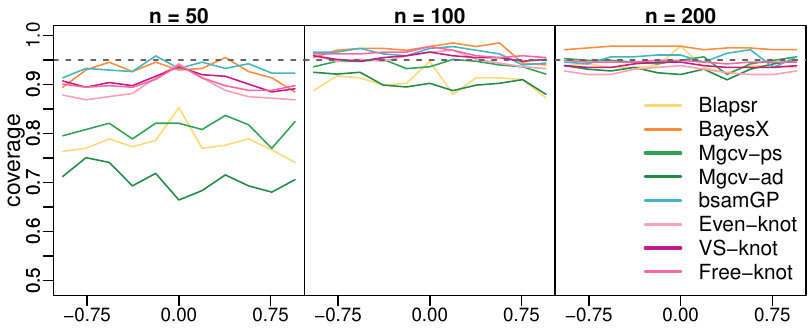}}
	\caption{Logarithm of RMSE and coverage probabilities for $f_1$, $f_2$, $f_3$, and $f_4$ in the nonparametric Poisson regression models with $n=50, 100, 200$, obtained from 500 replicated datasets. Outliers are excluded to improve visualization.}
	\label{plot:poi4}
\end{figure}

\begin{figure}[t!]
	\centering
	\subfloat[\centering Pointwise posterior mean estimates of $f_1$ in $100$ replications]{\includegraphics[width=1\textwidth]{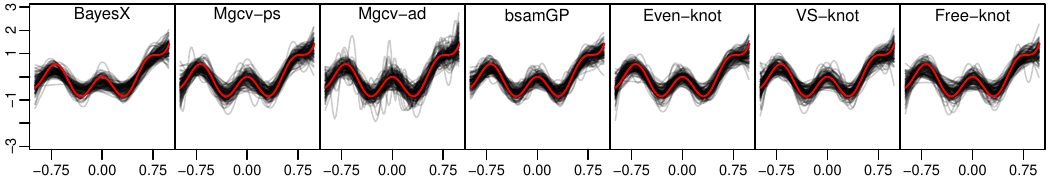}}\\
	\subfloat[\centering Pointwise posterior mean estimates of $f_2$ in $100$ replications]{\includegraphics[width=1\textwidth]{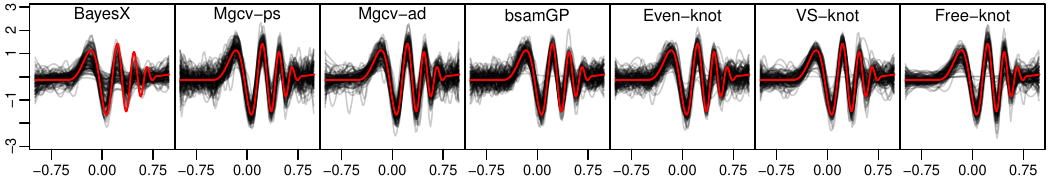}}\\
	\subfloat[\centering Pointwise posterior mean estimates of $f_3$ in $100$ replications]{\includegraphics[width=1\textwidth]{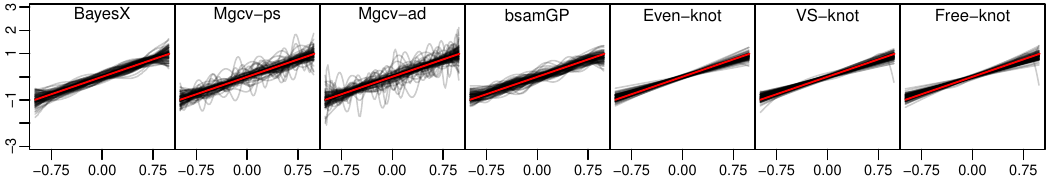}}\\
	\subfloat[\centering Pointwise posterior mean estimates of $f_4$ in $100$ replications]{\includegraphics[width=1\textwidth]{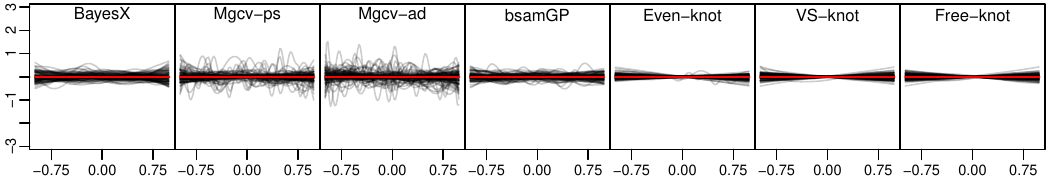}}
	\caption{Pointwise posterior means (gray) of $f_1$, $f_2$, $f_3$ and $f_4$ in the nonparametric Gaussian regression model with $n=200$, obtained from randomly chosen 100 replicated datasets, along with the true function (red).	}
	\label{plot:gau3}
\end{figure}

\begin{figure}[t!]
	\centering
	\subfloat[\centering Logarithm of RMSE for $f_1$]{\includegraphics[width = 0.47\textwidth]{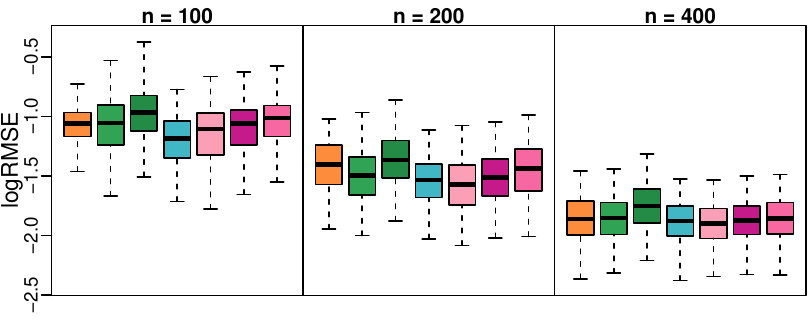}}
	~
	\subfloat[\centering Coverage probabilities for $f_1$]{\includegraphics[width = 0.47\textwidth]{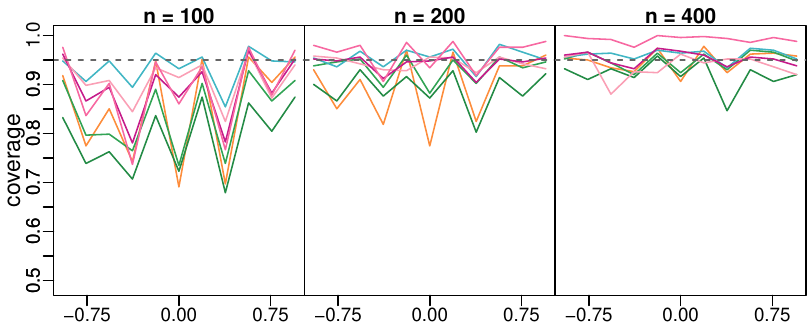}}
	\\
	\subfloat[\centering Logarithm of RMSE for $f_2$]{\includegraphics[width = 0.47\textwidth]{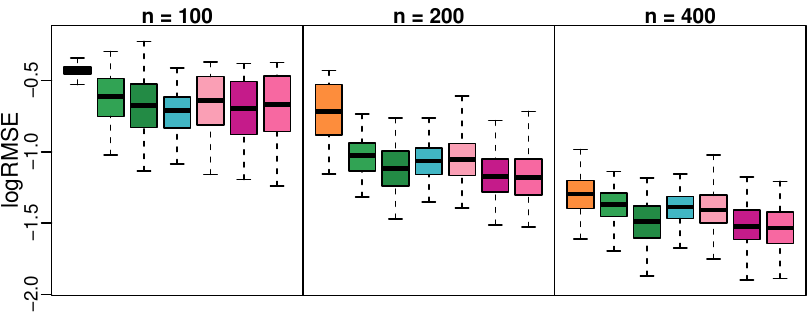}}
	~
	\subfloat[\centering Coverage probabilities for $f_2$]{\includegraphics[width = 0.47\textwidth]{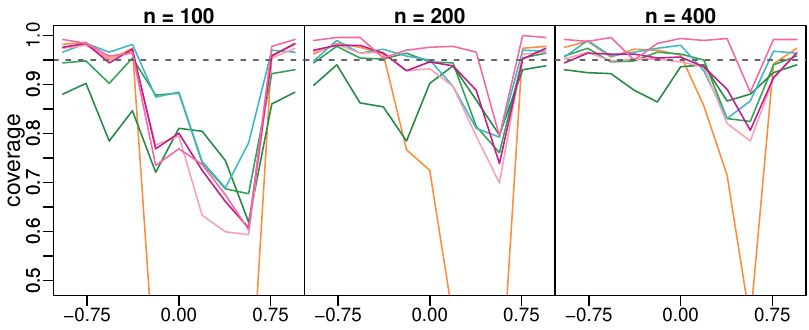}}
	\\
	\subfloat[\centering Logarithm of RMSE for $f_3$]{\includegraphics[width = 0.47\textwidth]{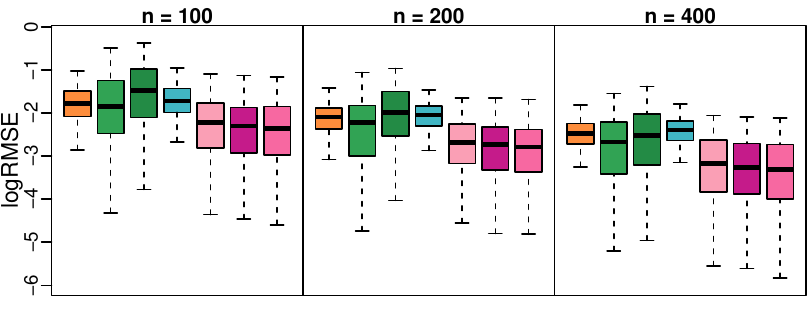}}
	~
	\subfloat[\centering Coverage probabilities for $f_3$]{\includegraphics[width = 0.47\textwidth]{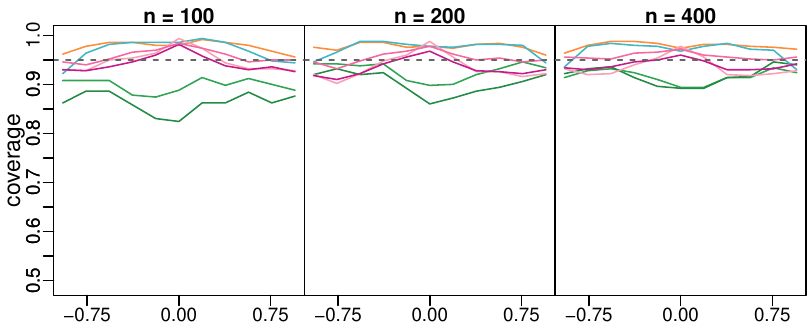}}
	\\
	\subfloat[\centering Logarithm of RMSE for $f_4$]{\includegraphics[width = 0.47\textwidth]{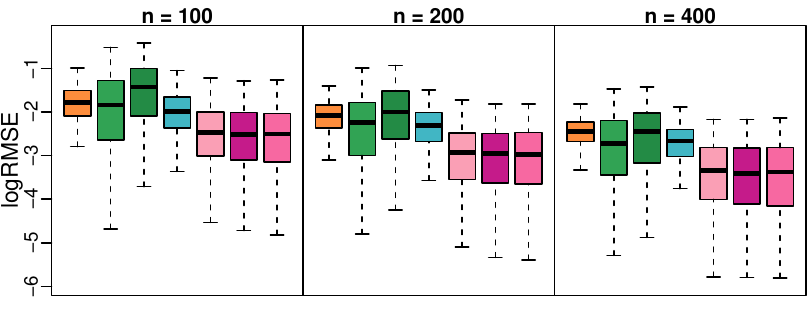}}
	~
	\subfloat[\centering Coverage probabilities for $f_4$]{\includegraphics[width = 0.47\textwidth]{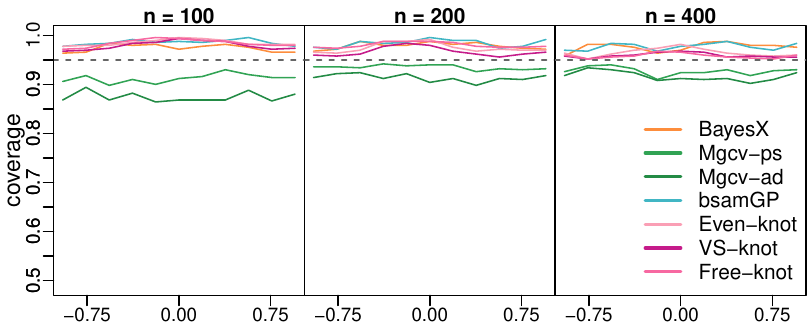}}
	\caption{Logarithm of RMSE and coverage probabilities for $f_1$, $f_2$, $f_3$ and $f_4$ in the nonparametric Gaussian regression models with $n=100, 200, 400$, obtained from 500 replicated datasets. Outliers are excluded to improve visualization.}
	\label{plot:gau4}
\end{figure}

\begin{figure}[t!]
	\centering
	\includegraphics[width = 10cm]{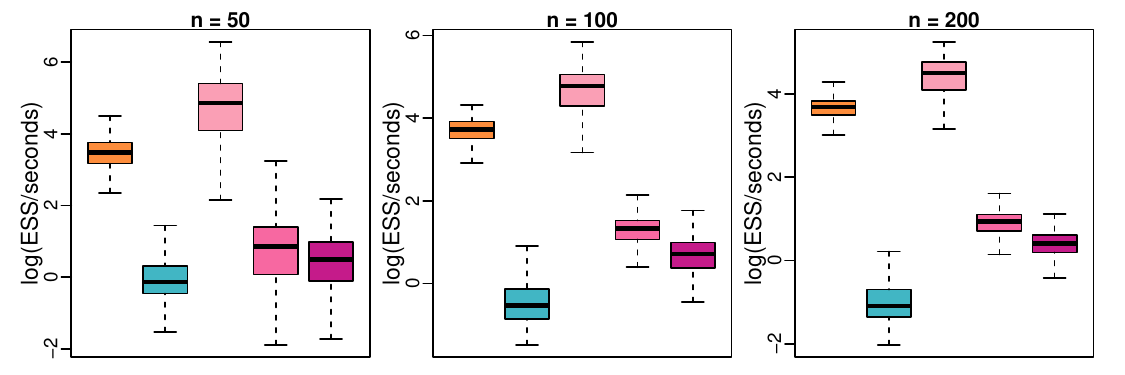} \\ 	\vspace{-0.4cm}
	\subfloat[\centering Sampling efficiency in the Poisson regression models]{\includegraphics[width=10cm]{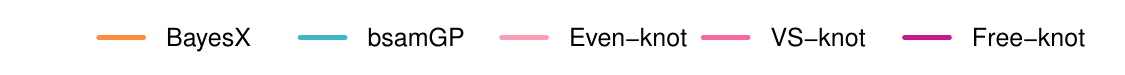}} \\ \vspace{0.6cm}
	\includegraphics[width = 10cm]{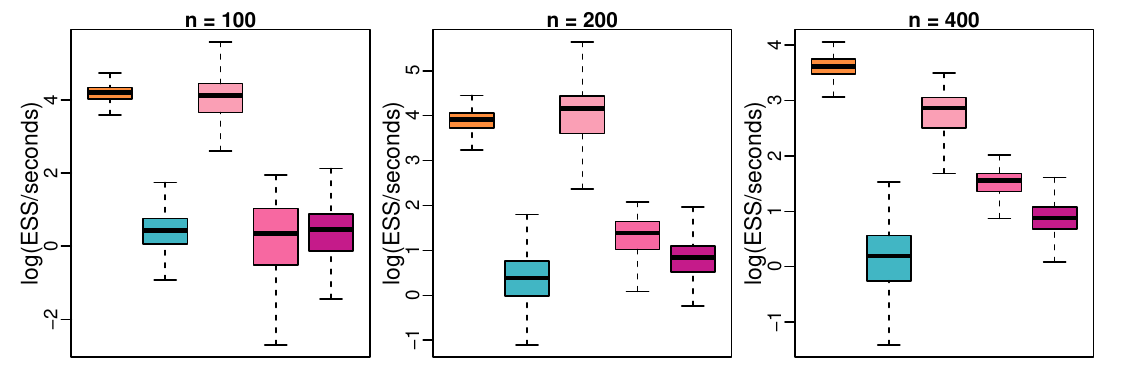} \\ 	\vspace{-0.4cm}
	\subfloat[\centering Sampling efficiency in the Gaussian regression models]{\includegraphics[width=10cm]{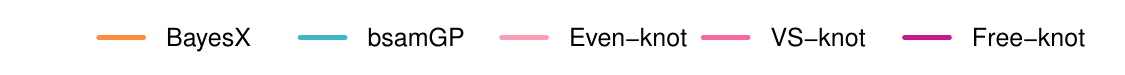}}
	\caption{Logarithm of the effective sample sizes of the joint posterior per second of runtime, in the Poisson regression models with $n=50, 100, 200$ and the Gaussian regression models with $n=100, 200, 400$, obtained from 500 replicated datasets.}
	\label{plot:ESS}
\end{figure}

\section{Simulation for Poisson and Gaussian regression}
\label{sec:sim1_poigau}

Section~\ref{sec:sims} presents simulations focused solely on a nonparametric logistic regression model. Given that the modeling framework encompasses a broad range of exponential family models, investigating the properties of BMS-based methods across other GAMs is important.
In this section, we extend our analysis to include simulation results for Poisson and Gaussian regression models. 
As in Section~\ref{sec:sims}, the observations are generated using the linear predictor $\eta_i=\alpha+\sum_{j=1}^4 f_j(x_{ij})$, where $f_j$ is the centered version of $f_j^\ast$ in \eqref{eqn:functions} and $\alpha$ represents the intercept introduced by centering. For Poisson regression, $Y_i\sim \text{Poi}(e^{\eta_i})$. For Gaussian regression, $Y_i = \eta_i +\epsilon_i$ where $\epsilon_i \sim N(0, 1)$. We generate 500 replicated datasets with sizes $n=50, 100, 200$ for Poisson regression and $n = 100, 200, 400$ for Gaussian regression. For each dataset, we apply the VS-knot spline approach to evaluate the differences among mixtures of g-priors. 
Additionally, we compare BMS-based methods with the intrinsic prior to other Bayesian methods to validate the effectiveness of BMS-based approaches.

The simulation results are presented in Figures~\ref{plot:poi1}--\ref{plot:gau4}. Specifically, Figures~\ref{plot:poi1}--\ref{plot:gau2} display the performance differences among the mixtures of g-priors in the VS-knot splines (Poisson regression in Figures~\ref{plot:poi1}--\ref{plot:poi2} and Gaussian regression in Figures~\ref{plot:gau1}--\ref{plot:gau2}). In contrast, Figures~\ref{plot:poi3}--\ref{plot:gau4} illustrate the comparison between the BMS-based methods and other Bayesian approaches (Poisson regression in Figures~\ref{plot:poi3}--\ref{plot:poi4} and Gaussian regression in Figures~\ref{plot:gau3}--\ref{plot:gau4}).
Since \texttt{Blapsr} requires a known value of $\phi$ in Gaussian regression, it is excluded from the comparison for Gaussian regression. The overall simulation performance aligns with the results from the logistic regression model in Section~\ref{sec:sims}, leading to similar conclusions. As in Section~\ref{sec:sims}, computational efficiency is assessed using the effective sample sizes of the joint posterior per second of runtime. Efficiency measures are summarized in Figure~\ref{plot:ESS}, confirming results consistent with those presented in Section~\ref{sec:sims}.

\section{Simulation for basis construction}
Proposition~\ref{prop:selection} suggests that the natural cubic spline basis in \eqref{eqn:ncs} is advantageous for both VS-knot and free-knot splines. To demonstrate the computational efficiency of this proposed basis construction, we conduct a numerical study. The simulation setups are identical to those described in Sections~\ref{sec:sim1} and~\ref{sec:sim1_poigau}. Along with the basis construction in \eqref{eqn:ncs}, we also consider the commonly used truncated power natural cubic splines as detailed in Equations (5.4) and (5.5) of \citet{hastie2009elements}. Both basis constructions are applied to VS-knot splines in our simulations.

Figure~\ref{plot:basis1} compares the computation times for both basis constructions across over 500 replications. Since the two basis constructions yield identical performance, we focus solely on computational runtime. The measurements were taken using a system equipped with an AMD Ryzen 9 7950X3D CPU. The results indicate that our proposed basis construction leads to faster computation times. Notably, the relative time improvement is more significant in Gaussian regression compared to logistic and Poisson regression models. 
This is due to the more extensive computation required by logistic and Poisson regression models, as they must calculate the maximum likelihood estimates in every MCMC iteration.
In contrast, Gaussian regression is less computationally intensive, as the maximum likelihood estimate is not necessary, allowing a larger proportion of the computation time to be allocated to basis construction.

\begin{figure}[t!]
	\centering
	\includegraphics[width = 9cm]{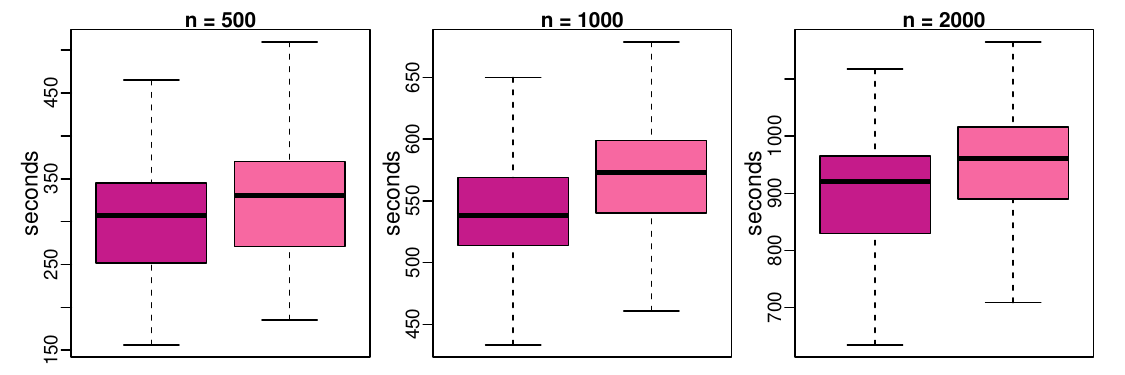} \\ \vspace{-0.4cm}
	\subfloat[\centering Logistic regression]{\includegraphics[width = 10cm]{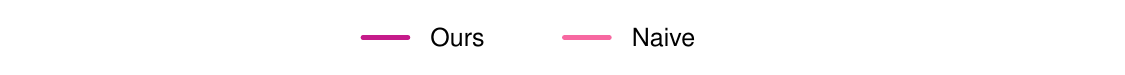}}\\ \vspace{0.4cm}
	\includegraphics[width = 9cm]{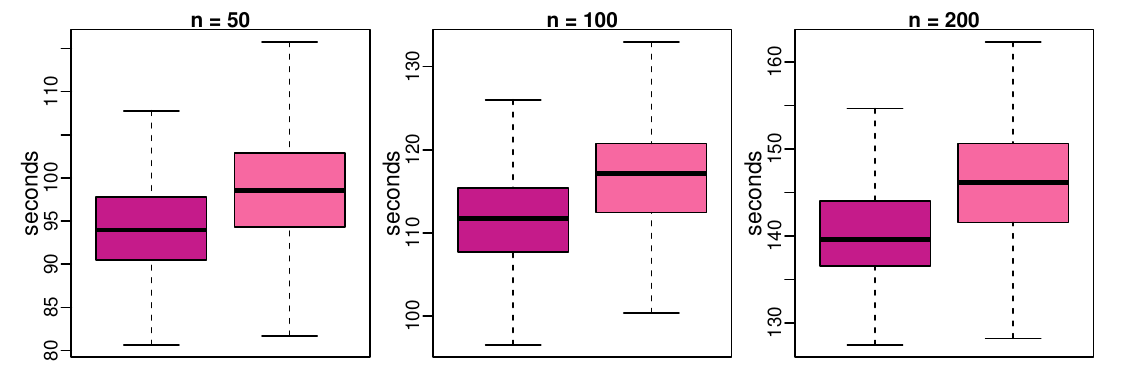} \\ \vspace{-0.4cm}
	\subfloat[\centering Poisson regression]{\includegraphics[width = 10cm]{plots/sim2_basistime_legend.pdf}}\\ \vspace{0.6cm}
	\includegraphics[width = 9cm]{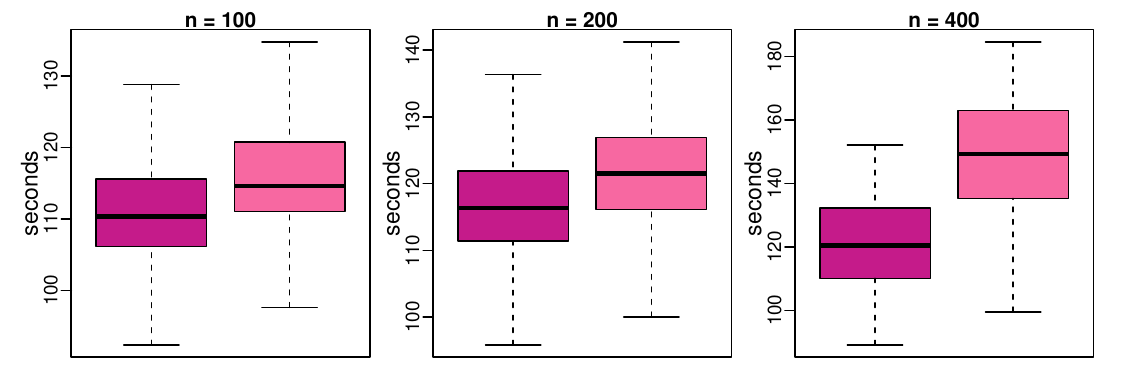} \\ \vspace{-0.4cm}
	\subfloat[\centering Gaussian regression]{\includegraphics[width = 10cm]{plots/sim2_basistime_legend.pdf}}
	\caption{Comparison of computation time between our basis construction (Ours) between the naive one given in  \citet{hastie2009elements} (Naive) using the VS-knot splines.}
	\label{plot:basis1}
\end{figure}

\section{R package \texttt{GAMBMS}}
Here, we demonstrate how to use the R package for BMS-based approaches to GAMs. To install and load our R package using the \texttt{devtools} package available on CRAN, run the following code:
\begin{verbatim}
	devtools::install_github("hun-learning94/gambms")
	library(gambms)
\end{verbatim}
The results presented in Sections~\ref{sec:sims} and \ref{sec:realdata} can be reproduced by running the examples provided on the help page of the R function \texttt{gambms}.

\end{document}